\documentclass[runningheads]{llncs}
\usepackage[T1]{fontenc}
\usepackage{newtxtext}       %
\usepackage[varvw]{newtxmath}       

\usepackage{amsmath}
\usepackage{graphicx}
\usepackage{algorithm}
\usepackage{latexsym}
\usepackage{times}
\usepackage{color}
\usepackage{subfigure}

\begin{document}

\title{An Asynchronous Maximum Independent Set Algorithm by Myopic Luminous Robots on Grids}

\author{Sayaka Kamei \inst{1}\and S\'ebastien Tixeuil\inst{2}}
\institute{Graduate School of Advanced Science and Engineering, Hiroshima University, Japan. \email{s10kamei@hiroshima-u.ac.jp} \and
Sorbonne University, CNRS, LIP6, Institut Universitaire de France, Paris, France.
\email{Sebastien.Tixeuil@lip6.fr} }

\authorrunning{S. Kamei and S. Tixeuil}

\maketitle
\begin{abstract}
We consider the problem of constructing a maximum independent set with mobile myopic luminous robots on a grid network whose size is finite but unknown to the robots.
In this setting, the robots enter the grid network one-by-one from a corner of the grid, and they eventually have to be disseminated on the grid nodes so that the occupied positions form a maximum independent set of the network.
We assume that robots are asynchronous, anonymous, silent, and they execute the same distributed algorithm.
In this paper, we propose two algorithms: 
The first one assumes that the number of light colors of each robot is three and the visible range is two, but uses the additional assumption that a local edge-labeling exists for each node.
To remove this assumption, the second one assumes that the number of light colors of each robot is seven, and that the visible range is three.
In both algorithms, the number of movements is $O(n(L+l))$ steps, where $n$ is the number of nodes and $L$ and $l$ are the grid dimensions.

\end{abstract}

\section{Introduction}
Swarm robotics envisions groups of mobile robots self-organizing and cooperating toward the resolution of common objectives, such as patrolling, exploring and mapping disaster areas, constructing ad hoc mobile communication infrastructures to enable communication with rescue teams, etc.
Our focus in this paper is the autonomous deployment of mobile robots in an unknown size rectangular area, \emph{e.g.}, for the purpose of establishing a communication infrastructure (if robots carry antennas) or a surveillance device (if robots carry intrusion sensors). 
When considering the rectangular area as a discrete structure (\emph{i.e.}, a graph, that depends on the antenna/sensor range: two nodes in the graph are adjacent if and only if they are within the range of the antenna/sensor), one can consider several placement strategies. Given that every location in the area must be covered by an antenna/sensor, there are two competing metrics:
\begin{enumerate}
\item The \emph{number} of deployed robots: The cost of the deployment obviously depends linearly from the number of robots deployed. 
\item The \emph{resilience} of the infrastructure in the case robots fail unpredictably: This amounts to the number of locations that are left uncovered when a robot (or a set of robots) ceases to perform its algorithm.
\end{enumerate}

Assuming full coverage is necessary, two extreme placement strategies are possible: A complete filling of each location by a robot enables maximum resilience (uncovering one location, say $C$ in Figure~\ref{fig:placements}(a), requires to disable five robots, at positions $A$, $B$, $C$, $D$ and $E$ in Figure~\ref{fig:placements}(a)) but also requires deploying one robot per location (so, the cost is highest), while a minimum dominating set strategy yields minimum cost, but poor resilience (disabling a single robot, say at $C$ in Figure~\ref{fig:placements}(b), uncovers five locations, $A$, $B$, $C$, $D$, and $E$ in Figure~\ref{fig:placements}(b)).
Maximal and maximum independent set placements are somewhat more balanced, despite the fact that any robot failure uncovers its location: a maximal independent set may use as little as one-third of the robots required for a complete filling, while retaining decent resilience (\emph{e.g.}, in Figure~\ref{fig:placements}(c), at least two robot failures $C$ and $D$ are required to uncover locations $A$ and $B$ beyond those initially hosting a robot); finally, a maximum independent set placement policy yields a resilience that is close to optimal (\emph{e.g.}, four robot failures, $A$, $B$, $C$, and $D$ uncover only one additional location, $E$ in Figure~\ref{fig:placements}(d)) while using half of the robots required for a complete filling.
In this paper, we concentrate on placing the robots according to a maximum independent set organization.

\begin{figure}[t]\centering
\subfigure[Filling Placement]{
\centering
\includegraphics[width=.3\textwidth]{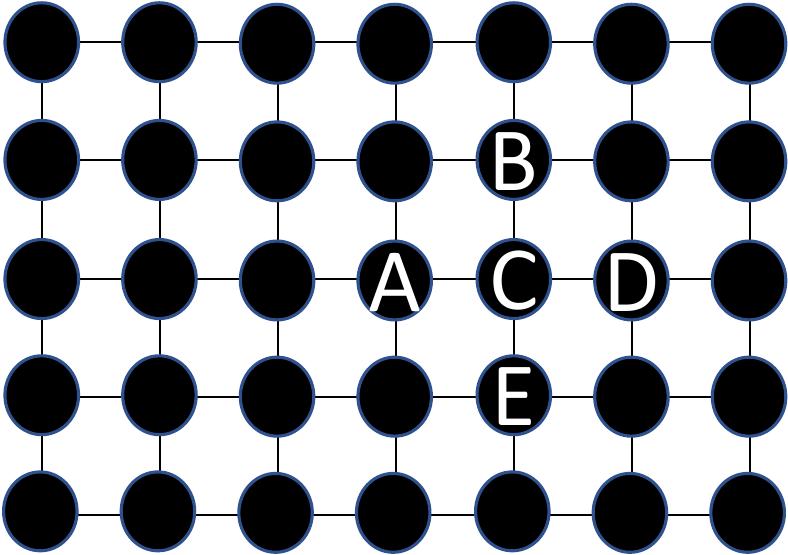}
}\hspace{1cm}
\subfigure[Minimum Dominating Set Placement]{ 
\centering
\includegraphics[width=.3\textwidth]{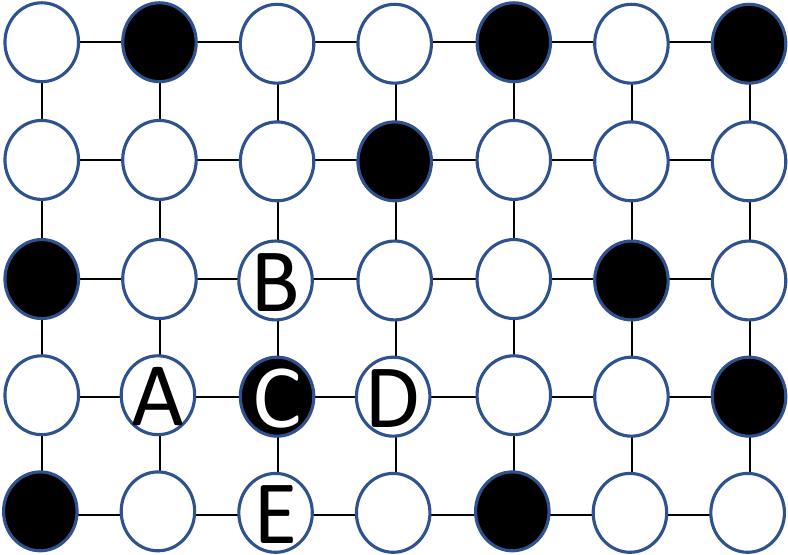}
}\\

\subfigure[Maximal Independent Set Placement]{
\centering
\includegraphics[width=0.3\textwidth]{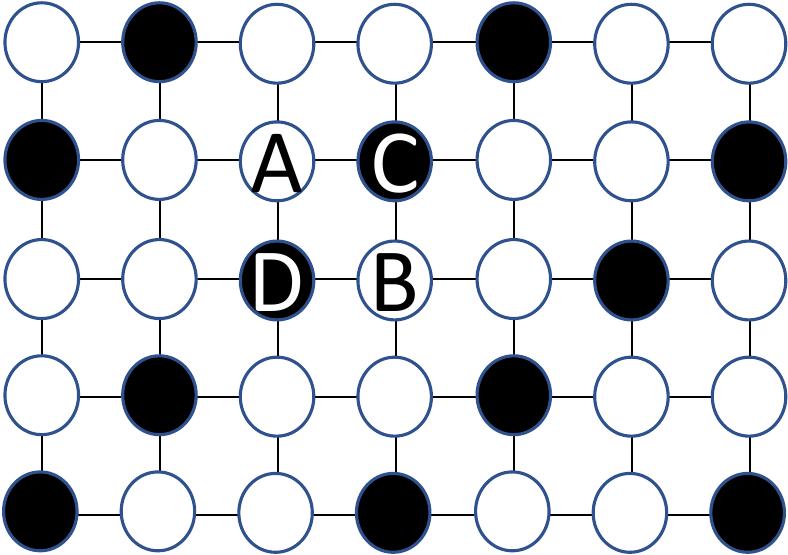}
}
\hspace{1cm}
\subfigure[Maximum Independent Set Placement]{
\centering
\includegraphics[width=0.3\textwidth]{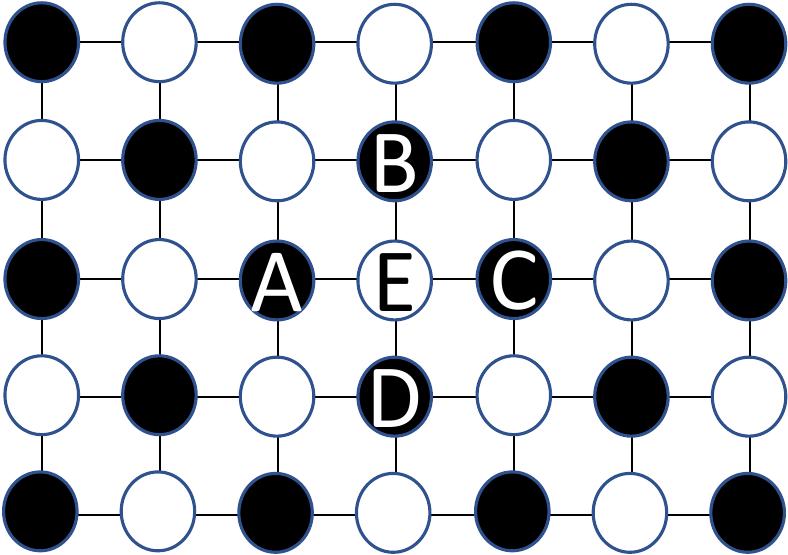}
}

\caption{Possible mobile robot placements on a $5\times 7$ grid.}
\label{fig:placements}
\end{figure}

\paragraph{Related Works.}

The seminal paper for studying robotic swarms from a distributed computing perspective is due to Suzuki and Yamashita~\cite{SY99}.
In the initial model, robots are represented as dimensionless points evolving in a bidimensional Euclidean space, and operate in \emph{Look-Compute-Move} (LCM) cycles.
In each cycle, a robot “Looks” at its surroundings and obtains (in its own coordinate system) a snapshot containing some information about the locations of all robots.
Based on this visual information, the robot “Computes” a destination location (still in its own coordinate system), and then “Moves” towards the computed location.
When the robots are \emph{oblivious}, the computed destination in each cycle only depends on the snapshot obtained in the current cycle (and not on the past history of actions).
Then, an \emph{execution} of a distributed algorithm by a robotic swarm consists in having every robot repeatedly execute its LCM cycle.
Although this mathematical model is perfectly precise, it allows a great number of variants (developed over a period of 20 years by different research teams~\cite{FPS19}), according to various dimensions, namely: sensors, memory, actuators, synchronization, and faults.

Although the seminal paper~\cite{SY99} focused on continuous spaces, many recent papers~\cite{FPS19p1,FPS19p2,FPS19p3,FPS19p4,FPS19p5} consider robots evolving on a discrete graph (that is, robots are located on a discrete set of locations, the nodes of the graph, and may move from one location to the next if an edge exists between the two locations). This interest for discrete spaces was further motivated by Balabonski et al.~\cite{balabonski19netys} that proved that discrete spaces are equivalent to discrete sensors (which better model actual sensing devices) in relevant settings. For the particular topology we consider, the grid, many problems were previously investigated, \emph{e.g.}, exploration~\cite{DLPRT12,SAFPS20}, perpetual exploration~\cite{bonnet11opodis}, scattering~\cite{barriere11ijfcs}, dispersion~\cite{WALCOM20}, gathering~\cite{DSKN16}, mutual visibility~\cite{ABKS2018}, pattern formation~\cite{BAKS2020}, and convex hull formation~\cite{HVST2020}. 
Similarly, the initial model considers unlimited visibility range, but actual sensors have a limited range, which makes solutions that only assume limited visibility more practical. When the evolving space is discrete, robots that can only see at a constant range (in the locations graph) are called \emph{myopic}. Myopic robots have successfully solved ring exploration~\cite{DLLP13}, gathering in bipartite graphs~\cite{GP13}, and gathering in ring networks~\cite{KLO14}.
Finally, another characteristic of the initial model, obliviousness, was recently dropped out in favor of a more realistic setting: \emph{luminous} robots. Oblivious robots were not able to remember past actions (each new Look-Compute-Move cycle reset the local memory of the robot), while luminous robots are able to remember and communicate\footnote{In the literature, this is refereed to as the Full Light model.} a finite value between two consecutive LCM cycles, using a visible light that is maintained by the robot. The number of values a robot is able to remember is tantamount to the number of different colors its light is able to show. Luminous robots were used to circumvent classical impossibility results in the oblivious model, mainly for gathering~\cite{G13,DFPSY16,heriban18icdcn}.
In this paper, we consider the particular combination of myopic and luminous robot model, that was previously used for ring exploration~\cite{SF-Ex,NOI19}, grid exploration~\cite{BDL19,APDCM2021}, and gathering on rings~\cite{OPODIS19}. 

The maximum independent set placement we consider in this paper is related to the benchmarking problem of geometric pattern formation initially proposed by Suzuki and Yamashita~\cite{SY99}. A key difference is that the target pattern is usually given explicitly to all robots (see the recent survey by Yamauchi~\cite{formation-survey}), while the maximum independent set pattern we target is only given as a constraint (as the dimensions of the grid are unknown to the robots, the exact pattern cannot be given to the robots).
Unconstrained placement of robots is also known as \emph{scattering} (in a continuous bidimensional Euclidean space~\cite{DP09j,CDPIM10,BT17j}, robots simply have to eventually occupy distinct positions). Evenly spreading robots in a unidimensional Euclidean space was previously investigated by Cohen and Peleg~\cite{CP08j} and by Flocchini~\cite{F16bc} and Flocchini et al.~\cite{FPS08j}. The bidimensional case was tackled mostly by means of simulation by Cohen and Peleg~\cite{CP08j} and by Casteigts et al.~\cite{casteigts12cc}. 
Most related to our setting is the barrier coverage problem investigated by Hesari et al.~\cite{HFNOS14c}: robots have to move on a continuous line so that each portion of the line is covered by robot sensors (whose range is a fixed value) despite the robots having limited vision (whose range is twice the range of the sensor). A key difference besides the robots evolving space (continuous segment \emph{versus} discrete grid) with our approach is that they consider oblivious robots and a common orientation, while we assume luminous robots and no orientation.
Another closely related problem was studied by Barri\`{e}re et al.~\cite{barriere11ijfcs}: uniform scattering on square grids. For uniform scattering to be solved, robots, initially at arbitrary positions, must reach a configuration where they are evenly spaced on a grid. Similarly to Hesari et al.~\cite{HFNOS14c}, Barri\`{e}re et al.~\cite{barriere11ijfcs} assume a common orientation (on both axes), that the length of every border is $kd+1$, where $k\geq 2$, $d\geq 2$, the number of robots is $(k+1)^2$, and each robot knows both $k$ and $d$.
They also assume that each robot has internal lights with six colors and that their visible radius is $2d$. 
Under the same assumptions as Barri\`{e}re et al.~\cite{barriere11ijfcs}, Poudel et al.~\cite{PS2020} proposed an algorithm 
needing $O(1)$ bit memory per robot, assuming a visibility radius of $2\max\{d,k\}$.
By contrast, we don't assume a common orientation, we use seven or three full lights colors, and the size of the grid is arbitrary and unknown to the robots. 

Finally, the placement method we describe as the fill placement (see Figure~\ref{fig:placements}(a)) was investigated as \emph{filling} (or uniform dispersal) problem \cite{Hsiang,filling2008,filling2013,filling2017,filling2018,filling2}.
Several of these solutions~\cite{Hsiang,filling2017,filling2018} assume the fully synchronous model (FSYNC).
It is known that~\cite{DFPSY16,DE2016} there exists some tasks that can be solved in FSYNC, but not in the asynchronous model with lights (ASYNC$^{O(1)}$), and also other tasks that can be solved in ASYNC$^{O(1)}$ but not in FSYNC. 
For ASYNC$^{O(1)}$, the algorithms by Barrameda et al.~\cite{filling2008,filling2013} assume common top-down and left-right directions, and the algorithm by Hideg et al.~\cite{filling2} assumes that each node has unique persistent edge labels on its adjacent edges.
In these algorithms~\cite{filling2008,filling2013,filling2}, the strategy of \emph{following the leader} is used. 
By the assumptions on common directions or edge-labeling, robots can recognize and memorize the node position they came from when at their current position, and are also able to recognize their predecessor and successor in their local views.

\paragraph{Our contribution.}
We propose the first two solutions to the maximum independent set placement of mobile myopic luminous robots on a grid of unknown size. 
Robots enter at a corner of the grid, and do not share a common orientation nor chirality. 
In the first algorithm $\mathcal{MIS\_LAB}$, each robot light can take $3$ different colors, and the visibility range of each robot is two. Similarly to previous work~\cite{filling2} for the filling problem, $\mathcal{MIS\_LAB}$ assumes that “local" edge-labeling\footnote{The edge-labeling are local in the sense that there is no coordination between adjacent nodes to label their common edge.} are available at each node, so that each robot can recognize its previous node from where it moved to its current node. 
The second algorithm $\mathcal{MIS\_NOLAB}$ gets rid of the edge-labeling assumption, and executes in a completely anonymous graph. It turns out that weakening this assumption has a cost on the number of colors ($7$ instead of $3$) and on the visibility radius ($3$ instead of $2$).
In both cases,
the placement process takes $O(n(L+l))$ steps of computation, where $n$ is the number of nodes and $L$ and $l$ are the grid dimensions.

As pointed out in the above, a maximum independent set placement yields good resilience in case of robot failures for the purpose of the target application, yet makes use of half of the robots needed for a complete filling of the grid. 

\section{Model}\label{model}
We consider an anonymous, undirected connected network $G^\prime=(V, E)$, where $V$ is a finite set of $n$ nodes $v_1, v_2, \cdots v_n$, and a specific node $v^\prime$ (discussed below), and $E$ is a finite set of edges.
We assume that the subgraph $G$ of $G^\prime$ induced by the nodes in $V \setminus \{v'\}$ is a two-dimensional $(l,L)$-grid, where $l\geq 3$ and $L\geq 3$ are two positive integers such that $l=n/L$, \emph{i.e.}, $l$ and $L$ are the size of borders.
We assume that these sizes $l$, $L$ and $n$ are unknown to the robots.
Let $\delta(v)$ be the degree of node $v$ in $G^\prime$.

The specific node $v^\prime$ is called a \emph{Door node}.
Each robot enters the grid $G$ one-by-one through the Door node.
We assume that $\delta(v^\prime)=1$, and the Door node is connected to a corner node of the grid (the particular corner $v^\prime$ is connected to is decided by an adversary). We refer to this corner as the \emph{Door corner}.
A robot at the Door node has to disperse through the grid while avoiding collisions. That is, two or more robots cannot occupy the same node.
When the Door node becomes empty, a new robot can be placed there immediately.
We use 
${\it Enter\_Grid}(r_i)$ to denote an operation that makes robot $r_i$ move from the Door node to the Door corner, and ${\it Move}(r_i)$ to denote an operation that makes $r_i$ move to an adjacent node in its direction.
We assume that robots have no orientation, \emph{i.e.}, a robot in unaware of $x$ and $y$ axes of the grid.
For each robot $r_i$, we call the robot $r_j$ that enters the grid just before $r_i$ is placed the Door node $r_i$'s \emph{predecessor}. Reciprocally, we call $r_i$ $r_j$'s \emph{successor}.

The distance between two nodes $v$ and $u$ is the number of edges in a shortest path connecting them.
The distance between two robots $r_1$ and $r_2$ is the distance between two nodes occupied by $r_1$ and $r_2$, respectively.
Two robots or two nodes are adjacent if the distance between them is one.

We assume that robots are homogeneous and anonymous.
That is, all robots have the same capabilities and execute the same algorithm, and they cannot distinguish two robots \emph{a priori}.

We assume that robots have limited visibility: an observing robot $r_i$ at node $u$ can only sense the robots that occupy nodes within a certain distance, denoted by $\phi$, from $u$.
For each robot $r_i$, we use ${\it view}(r_i)$ to denote the view of $r_i$.
Then, we call each robot $r_j$ in ${\it view}(r_i)$ a \emph{neighboring robot} of $r_i$. 

Each robot $r_i$ maintains a variable $c(r_i)$ called \emph{light}, which spans a finite set of states called \emph{colors}.
A light is persistent from one computational cycle to the next: the color is not automatically reset at the end of the cycle (see below how cycles drive the life of robots).
Robot $r_i$ knows its own current color of light and can detect colors of lights of other robots in the visibility range.
Robots are unable to communicate with each other explicitly (\emph{e.g.}, by sending messages), however, they can observe their environment, including the positions and colors of other robots, in their visibility range.

Each robot $r_i$ executes Look-Compute-Move cycles infinitely many times: \emph{(i)} first, $r_i$ takes a snapshot of the environment and obtains an ego-centered view of the current configuration (Look phase), \emph{(ii)} according to its view, $r_i$ decides to move or to stay idle and possibly changes its light color (Compute phase), \emph{(iii)} if $r_i$ decided to move,
it moves to one of its adjacent nodes depending on the choice made in the Compute phase (Move phase). At each time instant $t$, a subset of robots is activated by an entity known as the scheduler. This scheduler is assumed to be fair, \emph{i.e.}, all robots in $G^\prime$ are activated infinitely many times in any infinite execution. In this paper, we consider
the most general \emph{asynchronous} model: the time between Look, Compute, and Move phases is finite but unbounded. 
Similarly to existing literature~\cite{FPS19p1,FPS19p2,FPS19p3,FPS19p4,FPS19p5}, we assume that the move operation in Move phase is atomic, that is, if a robot moves to a neighboring node, its movement is completed instantaneously. 
Thus, when any robot takes a snapshot, it sees robots on nodes and not on edges\footnote{The assumption that moves are atomic was shown equivalent~\cite{balabonski19netys} to the assumption that moves are not atomic but the sensors see the robot either at the starting node or at the destination node, and no inversion of the observations is possible. 
For the sake of proof readability, we retain the former hypothesis.}.
Since the scheduler is allowed to interleave the different phases between robots, some robots may decide to move according to a view that is different from the current configuration. Indeed, during the Compute phase, other robots may move. We call a view that is different from the current configuration an \emph{outdated view}, and a robot with an outdated view an \emph{outdated robot}.

In this paper, the set of robots that enter the grid $G$ from a Door node constructs a maximum independent set of $G$. 

\begin{definition}
An independent set $I$ of $G$ is a subset of $V\setminus\{v^\prime\}$ such that no two nodes in $I$ are adjacent on $G$. 
A maximum independent set is an independent set containing the largest possible number of nodes for $G$.
\end{definition}

\section{Proposed Algorithms}

In this section, we present two algorithms $\mathcal{MIS\_LAB}$ and $\mathcal{MIS\_NOLAB}$ to construct a maximum independent set when the Door node is connected to a corner node.
The first algorithm $\mathcal{MIS\_LAB}$ makes the assumption that outgoing edges are labeled ``locally'' (that is, 
each node must assign distinct labels to different outgoing edges, however, any given edge may be associated to two different labels (as the assignment of a label is performed independently by each adjacent node)), 
and assumes that each robot is endowed with a light enabling $3$ colors, has visibility radius $2$.
Additionally, each robot can recognize its previous node by edge labeling.
To remove the edge labeling assumption, the second algorithm $\mathcal{MIS\_NOLAB}$ makes use of more colors (7 colors are used) and a larger visibility radius (\emph{i.e.}, $3$). As a result, it operates in the ``vanilla'' Look-Compute-Move model (no labeling of nodes or edges, etc.). 
In both algorithms, we assume no agreement on the grid axes or directions.



\subsection{Algorithm $\mathcal{MIS\_LAB}$ with 3 colors lights, $\phi=2$, and edge-labeling}\label{sec:3colors}
First, we propose algorithm $\mathcal{MIS\_LAB}$ that assumes that three light colors are available (and referred to as $F$, $p_1$, and $p_2$), and that $\phi=2$.
The color $F$ means that the robot finished the execution of the algorithm, and stops its execution. Colors $p_1$ and $p_2$ are used when the robot has not yet finished its execution.
We say that if the light color of a robot $r_i$ is $F$, then $r_i$ is Finished.
Initially, the color of the light $c(r_i)$ for each robot $r_i$ is $p_1$. 

For this algorithm, we add the following assumptions to the model in Section~\ref{model}:
\begin{itemize}
\item For each node of the grid, adjacent nodes (except the Door node) are arranged in a fixed order, and this order is only visible for robots on the node as \emph{edge labels}. The order does not change during the execution.  
\item Each robot can recognize the node it came from when at its current node (We call it \emph{previous node}). 
\end{itemize}
These assumptions are those considered in related work for the filling problem~\cite{filling2}.
Note that, the latter assumption can be implemented using an additional light of four colors to remember the edge label of the previous node.

Now, we explain the strategy of the routing scheme we use to construct a maximum independent set placement using Figure~\ref{fig:st1}. 
We call the border that is connected to the Door corner by an edge with a maximal edge label among the two edges of the Door corner ``the first border", and the border that is connected to the first border and is not-connected to the Door corner ``the second border", as per Figure~\ref{fig:st1}(b).
Without loss of generality, let $L$ be the size of the first border and $l$ be the size of the second border.
Additionally, we call the second border \emph{0-line}, and count the lines that are parallel to the second border in the following way: the line that is adjacent to 0-line is \emph{1-line}, and the border that is connected to the Door corner but not the first border is \emph{$(L-1)$-line}.
As the strategy of positioning of robots, first, we position the robots on the second border. After we position the robots on $m$-line, we position the robots on $(m+1)$-line.
To this purpose, because each robot $r_i$ starts with $c(r_i)=p_1$ from the Door node (Figure~\ref{fig:st1}(a)), each robot moves along the first border and enters a grid line where it should settle. To recognize the line where each robot should settle, we use the light colors.

In the following, we explain how this strategy can be put into practice in $\mathcal{MIS\_LAB}$.
In Figure~\ref{fig:st1}, the thick white circle represents a Finished robot with light color $F$, and the diagonal (resp. horizontal) striped circle represents a robot with light color $p_1$ (resp. $p_2$).
On the Door corner, each robot chooses an adjacent node according to the edge with the locally maximal label.
Each robot moves on the first border to keep the distance from its predecessor two or more hops.
Each robot arrives at the first corner, then changes $c(r_i)$ to $p_2$.
After that, the first robot $r_1$ goes through the second border, eventually arrives at the second corner (Figure~\ref{fig:st1}(b)), and changes $c(r_1)$ to $F$.
We call this second corner \emph{the diagonal corner}.
The successor $r_i$ follows its predecessor $r_j$ while striving to keep a distance of at least two.
If $r_i$ has $c(r_i)=p_2$, and $r_j$ is Finished two hops away, then $r_i$ changes $c(r_i)$ to $F$ (Figure~\ref{fig:st1}(c)).
If $r_i$ with $c(r_i)=p_1$ observes that $r_j$ is Finished two hops away, $r_i$ changes $c(r_i)$ to $p_2$ and makes the next line (Figure~\ref{fig:st1}(d)).
By repeating such elementary steps, eventually, a maximum independent set can be constructed (Figure~\ref{fig:st1}(f)).
Because each robot can recognize its previous node by the assumption, it can recognize its predecessor and its successor if there are two neighboring non-Finished robots.
In this strategy, to distinguish which border or corner the robot stays on, the light colors $p_1$ and $p_2$ are used. 
Although the described strategy reminds the \emph{follow the leader} strategy previously used for the filling problem~\cite{filling2008,filling2013,filling2}, a noticeable difference is that in the proposed algorithm, a robot may loose visual contact with its predecessor.
\begin{figure}[t]
\centering
    \includegraphics[width=\textwidth]{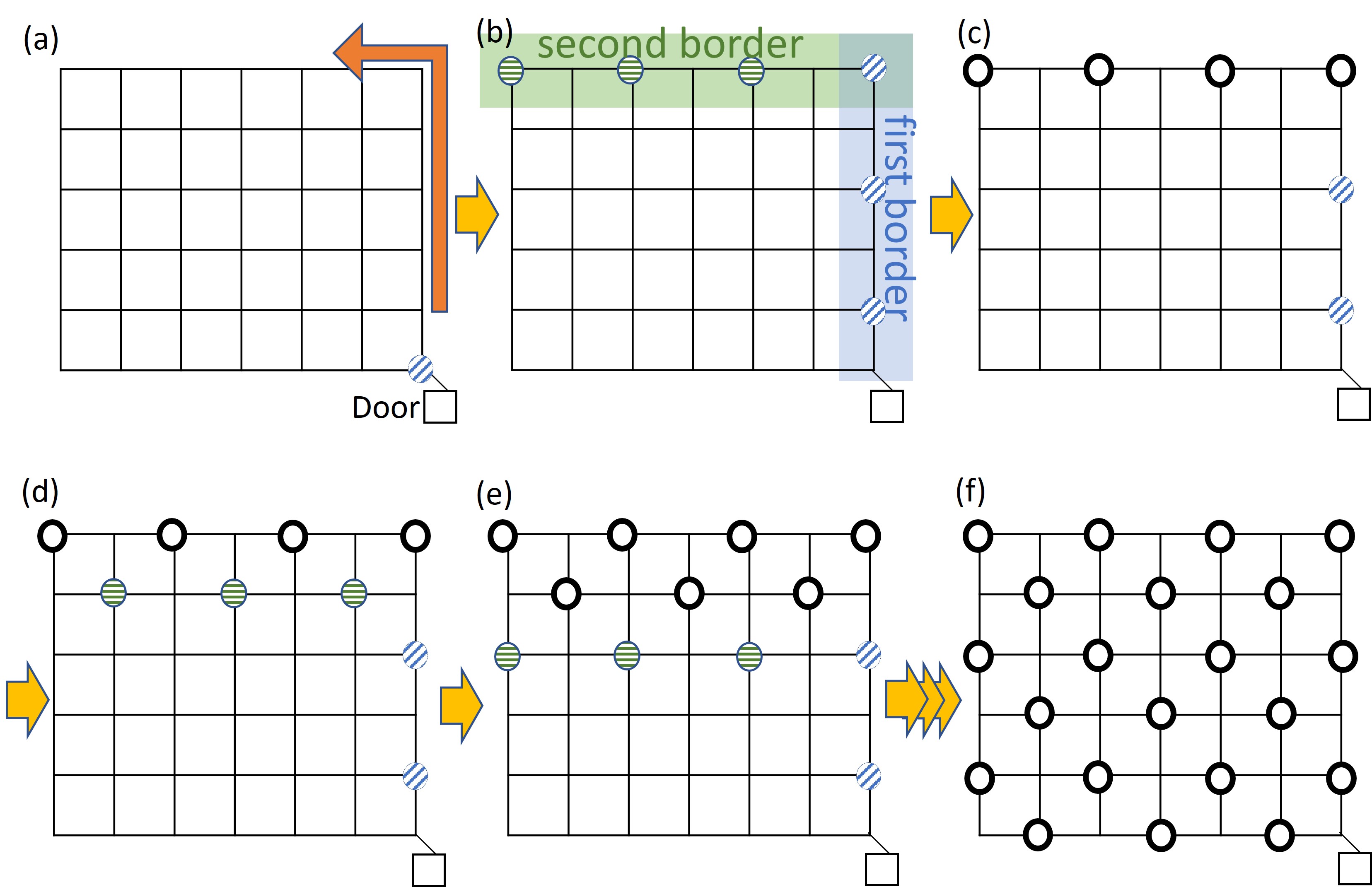}
    \caption{Strategy of the maximum independent set placement from a Door node on a corner.}
    \label{fig:st1}
\end{figure}

To describe $\mathcal{MIS\_LAB}$, we use the definitions of some view types 
in Figures~\ref{fig:p1D}--\ref{fig:p2F}. 
The notations in these figures are summarized in Table~\ref{note1}.
In Figures~\ref{fig:p1D}--\ref{fig:p2F}, each view ${\it view}(r_i)$ is from robot $r_i$ represented by the center filled circle.
The dotted circle without frame border represents the previous node from where $r_i$ moved to the current node. 
If ${\it view}(r_i)$ is a view with an arrow from $r_i$ not on the Door node (\emph{i.e.}, Figures~\ref{fig:p1M}, \ref{fig:p1C} and \ref{fig:p2M}), the arrow represents the direction of ${\it Move}(r_i)$ operation.
If $r_i$ is on the Door node (\emph{i.e.}, Figure~\ref{fig:p1D}), the arrow represents the operation ${\it Enter\_Grid}(r_i)$.
If the previous node is not the Door node, the circle with diagonal stripes or horizontal stripes (which is adjacent to the previous node) represents the node where $r_i$'s successor robot may be hosted.
So, such a node may: \emph{(i)} host the successor of $r_i$, \emph{(ii)} be empty, or \emph{(iii)} do not exist. 
If the successor robot is on the diagonal (resp. horizontal) striped node, it has $p_1$ (resp. $p_1$ or $p_2$). 
The diagonal striped node can be the Door node under the grid size hypothesis, \emph{i.e.}, the diagonal striped node in Figure~\ref{fig:p1F}(a) can be the Door node but in Figure~\ref{fig:p1F}(b) cannot be the Door node due the grid size hypothesis.
If the previous node is the Door node, the successor robot can be the previous node by assumption.
The thick white circle represents a Finished robot that must be on the node.
The circle with vertical stripes represents either a node hosting a Finished robot, or no node.
A node with a dashed white square represents an empty node, if the node exists on the grid. All other empty nodes must exist on the grid and host no robot.
In our classification, each type of views may include several possible views. 
For example, in Figure~\ref{fig:p1F}(a), the upper adjacent empty node may be a corner, then the top node with vertical stripes does not exist. Thus, in Figure~\ref{fig:p1F}(a), there are five possible views, depending on whether the top node exists or not, the bottom node exists or not, and whether the successor is in the view or not. Note that all combinations are not feasible, since \emph{e.g.} if the bottom node does not exist, then the previous node is the Door node, and the successor is on the Door node, which in turn implies that the top Finished robot must exist (due to the grid size hypothesis). 

\begin{table}[t]
    \centering
    \caption{Notations for $\mathcal{MIS\_LAB}$.}
   \includegraphics[width=10cm]{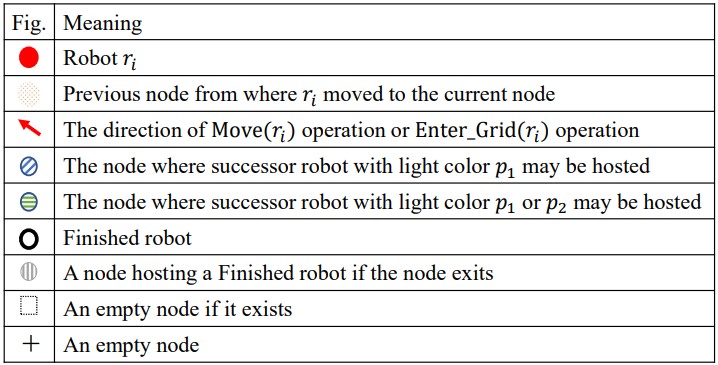}
    \label{note1}
\end{table}

Next, we explain the definition of each view type (except the direction of the movement) in detail.
While $c(r_i)=p_1$, we use the view types in Figures~\ref{fig:p1D}-\ref{fig:p1C}. 
In set ${\it Door}$ of view types in Figure~\ref{fig:p1D}, there are two types of views such that $r_i$ is on the Door node: 
    \begin{itemize}
        \item {\sf DoorA}: $r_i$ cannot see other robots. 
        \item {\sf DoorB}: A node two hops away is occupied by a Finished robot, and other nodes including the Door corner are empty. 
    \end{itemize}
\begin{figure}[t]
    \centering
    \subfigure[{\sf DoorA}]{\includegraphics[height=1.5cm]{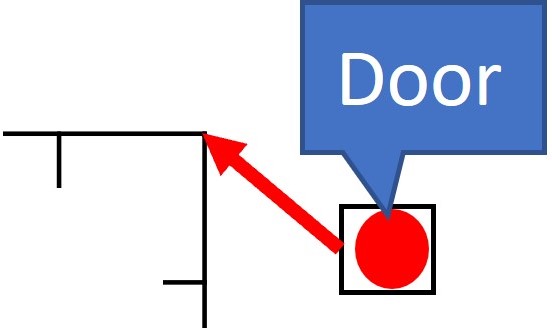}}\hspace{0.5cm}
    \subfigure[{\sf DoorB}]{\includegraphics[height=1.5cm]{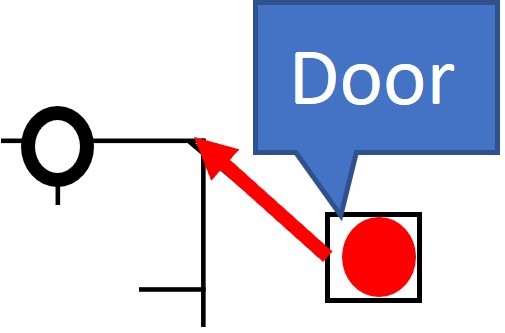}}
    \caption{Definition of views in ${\it Door}$ while $p_1$.}
    \label{fig:p1D}
\end{figure}
In set ${\it MovP1}$ in Figure~\ref{fig:p1M}, there are three types of views:
\begin{itemize}
    \item {\sf MovP1A}: $r_i$ is on the Door corner, and cannot see other robots on the grid (that is, not on the Door node). 
    \item {\sf MovP1B}: $r_i$ is on a border, and cannot see other robots except its successor. 
    \item {\sf MovP1C}: $r_i$ is on a border, and can see only a Finished robot two hops away on the border on the opposite side from its previous node. Note that $r_i$ may be on the Door corner in this view type.
\end{itemize}
\begin{figure}[t]
    \centering
    \subfigure[{\sf MovP1A}]{\includegraphics[height=1.6cm]{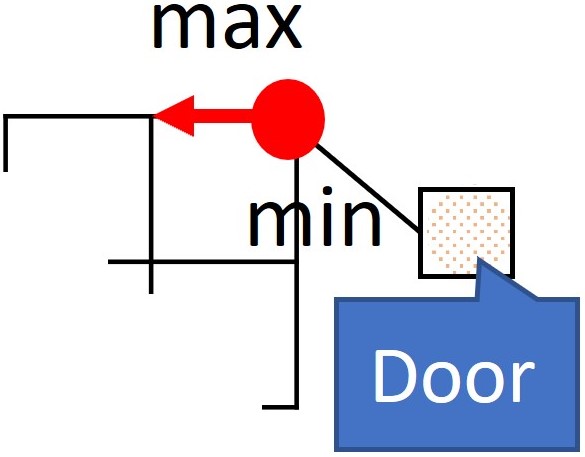}}\hspace{0.5cm}
    \subfigure[{\sf MovP1B}]{\includegraphics[height=2.1cm]{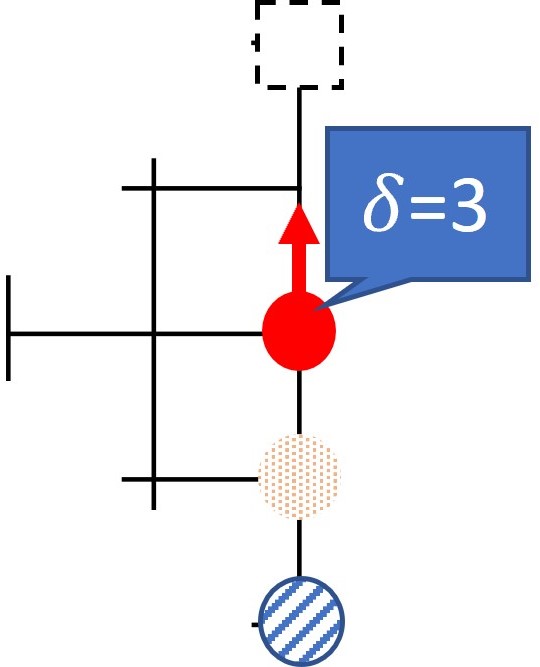}}\hspace{0.5cm}
    \subfigure[{\sf MovP1C}]{\includegraphics[height=2.1cm]{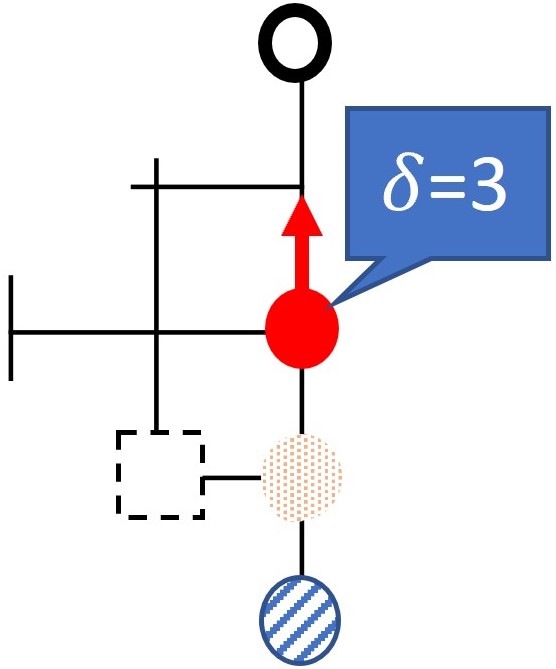}}
    \caption{Definition of views in ${\it MovP1}$ while $p_1$.}
    \label{fig:p1M}
\end{figure}
In set ${\it StopP1}$ in Figure~\ref{fig:p1F}, there are two types of views:
\begin{itemize}
    \item {\sf StopP1A}: $r_i$ is on a border, and all nodes two hops away in the direction opposite to $r_i$'s previous node are occupied by Finished robots. Because $c(r_i)=p_1$, the border $r_i$ stays on is the first border. 
    Note that $r_i$ may be on the Door corner in this view type. 
    \item {\sf StopP1B}: $r_i$ is on a corner, and can see a Finished robot on a border that does not include $r_i$'s previous node. In this case, because $c(r_i)=p_1$, the corner $r_i$ stays on is the first corner.
\end{itemize}
\begin{figure}[t]
\centering
    \subfigure[{\sf StopP1A}]{\includegraphics[height=2.1cm]{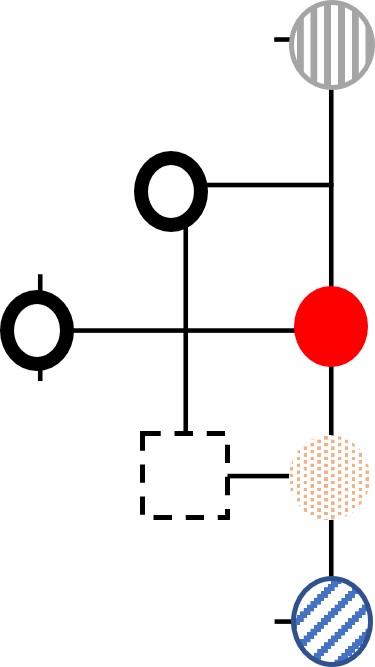}}\hspace{0.5cm}
    \subfigure[{\sf StopP1B}]{\includegraphics[height=1.8cm]{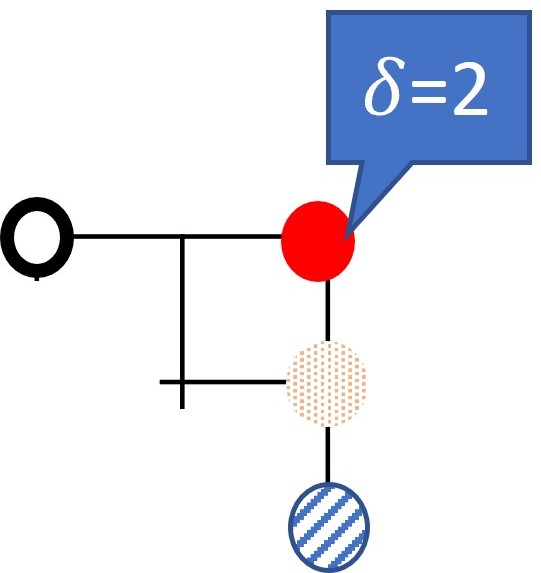}}
    \caption{Definition of views in ${\it StopP1}$ while $p_1$.}
\label{fig:p1F}
\end{figure}
In set ${\it Change}$ in Figure~\ref{fig:p1C}, there are three types of views:
\begin{itemize}
    \item {\sf ChangeA}: $r_i$ is on a corner, and cannot see other robots except its successor.
    \item {\sf ChangeB}: $r_i$ is on a border, the adjacent node opposite to $r_i$'s previous node on the border is occupied by a Finished robot, and $r_i$ cannot see other non-Finished robots than its successor. 
    \item {\sf ChangeC}: $r_i$ is on a border, all nodes two hops away in the direction opposite to $r_i$'s previous node except on the same line as $r_i$ stays are occupied by Finished robots, and $r_i$ cannot see other robots on the same line. 
\end{itemize}
\begin{figure}[t]
\centering
    \subfigure[{\sf ChangeA}]{\includegraphics[height=1.8cm]{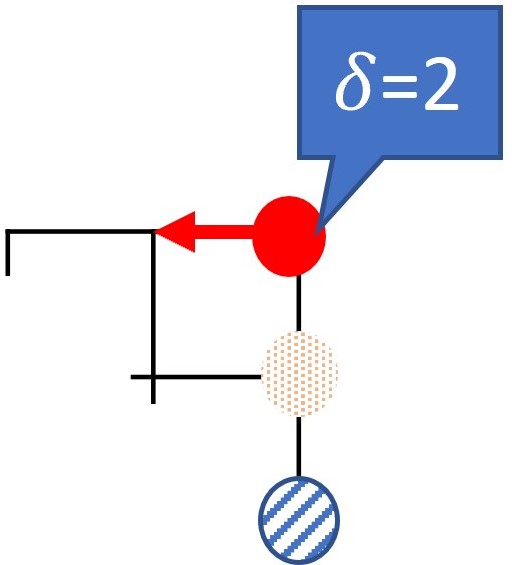}}\hspace{0.5cm}
    \subfigure[{\sf ChangeB}]{\includegraphics[height=2.1cm]{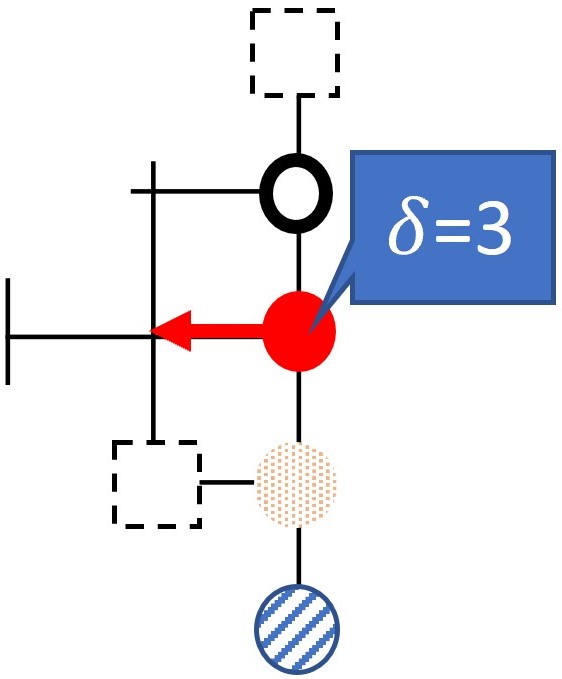}}\hspace{0.5cm}
    \subfigure[{\sf ChangeC}]{\includegraphics[height=2.1cm]{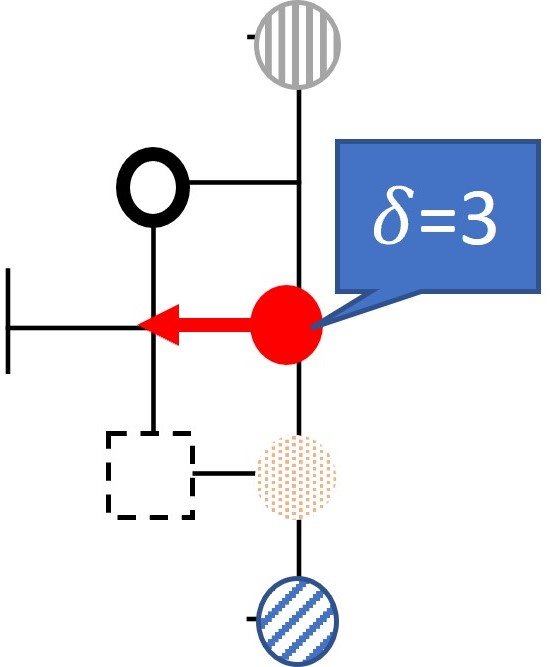}}
\caption{Definition of views in ${\it Change}$ while $p_1$.}
    \label{fig:p1C}
\end{figure}
While $c(r_i)=p_2$, we use the view types in Figures~\ref{fig:p2M}-\ref{fig:p2F}.
In set ${\it MovP2}$ in Figure~\ref{fig:p2M}, there are three types of views:
\begin{itemize}
    \item {\sf MovP2A}: $r_i$ is on a border, and cannot see other robots except its successor. 
    \item {\sf MovP2B} and {\sf MovP2C}: $r_i$ can see Finished robots on an adjacent line, and cannot see other non-Finished robots except its successor. The difference between these view types is that the Finished robot is adjacent to $r_i$ only in {\sf MovP2C}.
\end{itemize}
\begin{figure}[t]\centering
    \subfigure[{\sf MovP2A}]{\includegraphics[height=2.1cm]{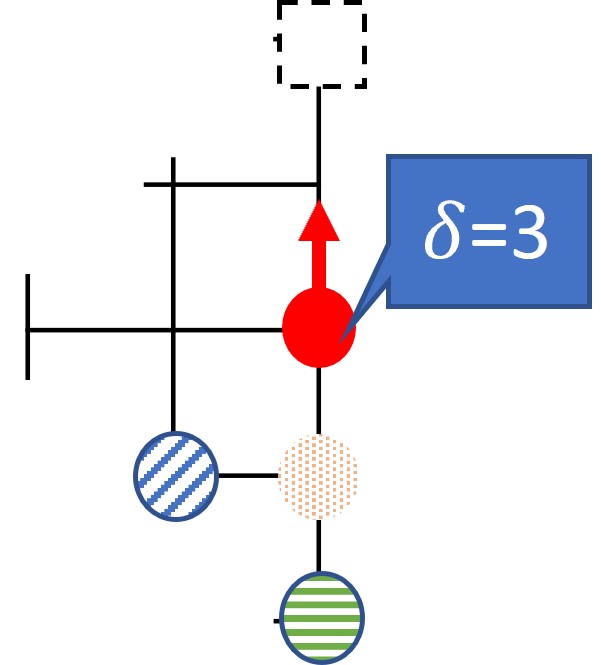}}\hspace{1cm}
    \subfigure[{\sf MovP2B}]{\includegraphics[height=2.1cm]{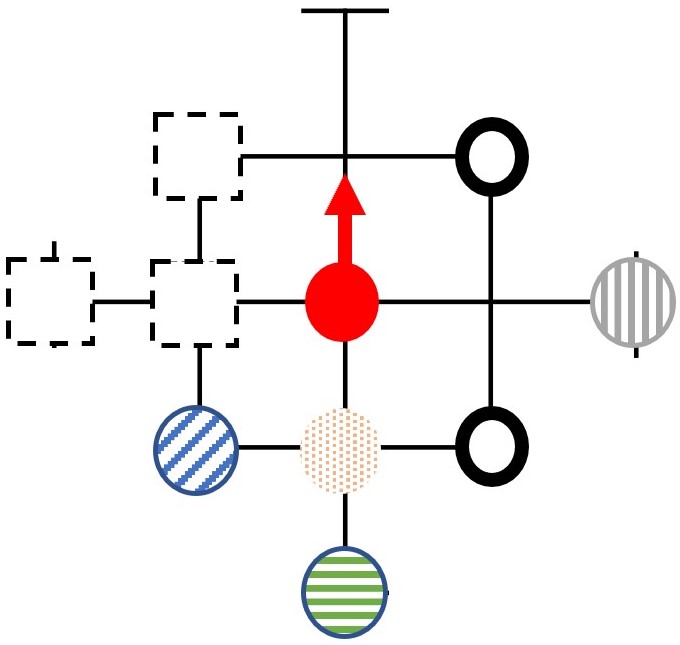}}\hspace{1cm}
    \subfigure[{\sf MovP2C}]{\includegraphics[height=2.1cm]{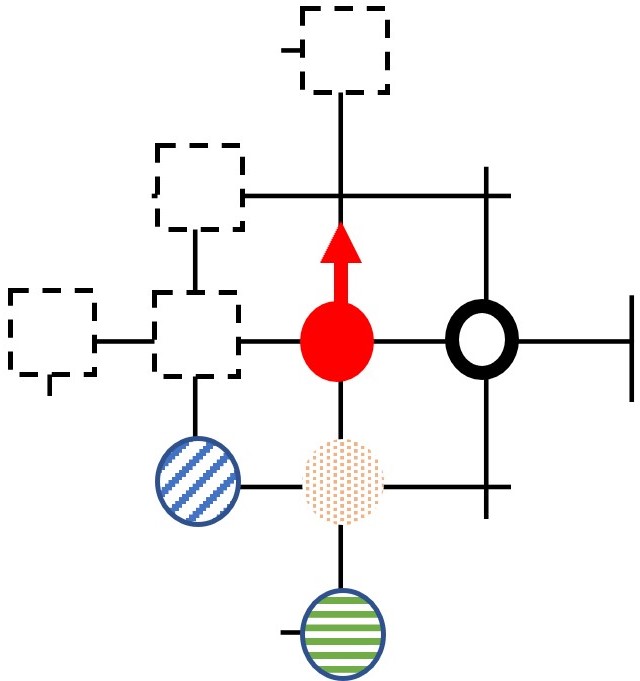}}
    \caption{Definition of views in ${\it MovP2}$ while $p_2$.}
    \label{fig:p2M}
\end{figure}
In set ${\it StopP2}$ in Figure~\ref{fig:p2F}, there are five types of views:
\begin{itemize}
    \item {\sf StopP2A}: $r_i$ is on a corner (\emph{i.e.}, it reached the diagonal corner).
    \item {\sf StopP2B}: $r_i$ is on a border, and can see a Finished robot on the border two hops away on the opposite side of its previous node. Since $c(r_i)=p_2$, the border is the second border.
    \item {\sf StopP2C} (resp. {\sf StopP2D}): $r_i$ is on the end of the line (resp. at the node one hop from the end of the line), and can see Finished robots on an adjacent line. 
    \item {\sf StopP2E}: $r_i$ can see a Finished robot on the same line as $r_i$ stays two hops away on the opposite side of its previous node.
\end{itemize}
\begin{figure}[t]
    \centering
    \subfigure[{\sf StopP2A}]{\includegraphics[height=2.0cm]{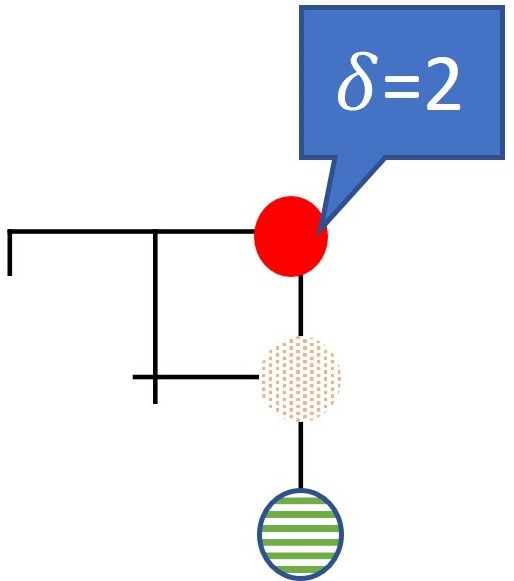}}\hspace{0.4cm}
    \subfigure[{\sf StopP2B}]{\includegraphics[height=2.1cm]{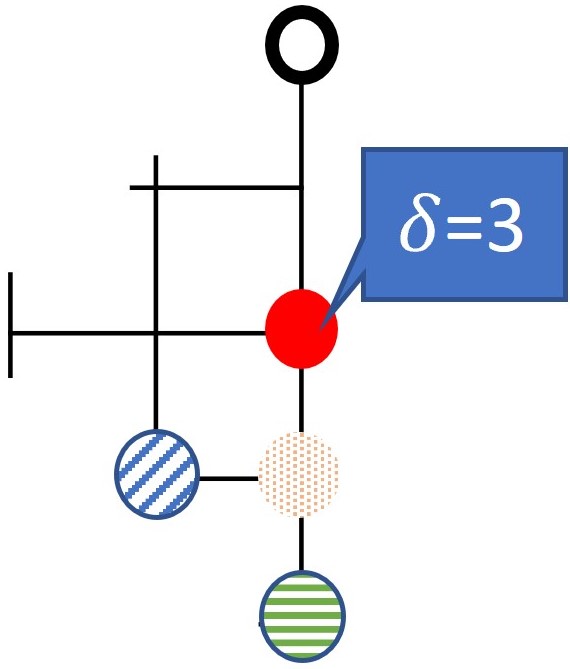}}\hspace{0.4cm}
    \subfigure[{\sf StopP2C}]{\includegraphics[height=1.7cm]{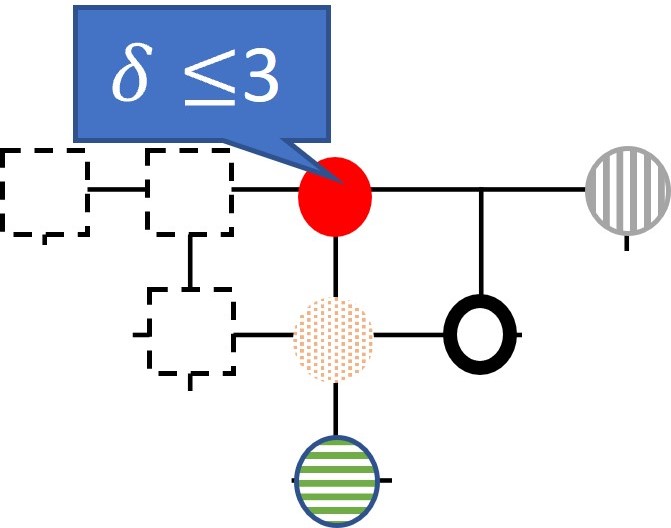}}\hspace{0.4cm}
    \subfigure[{\sf StopP2D}]{\includegraphics[height=2.1cm]{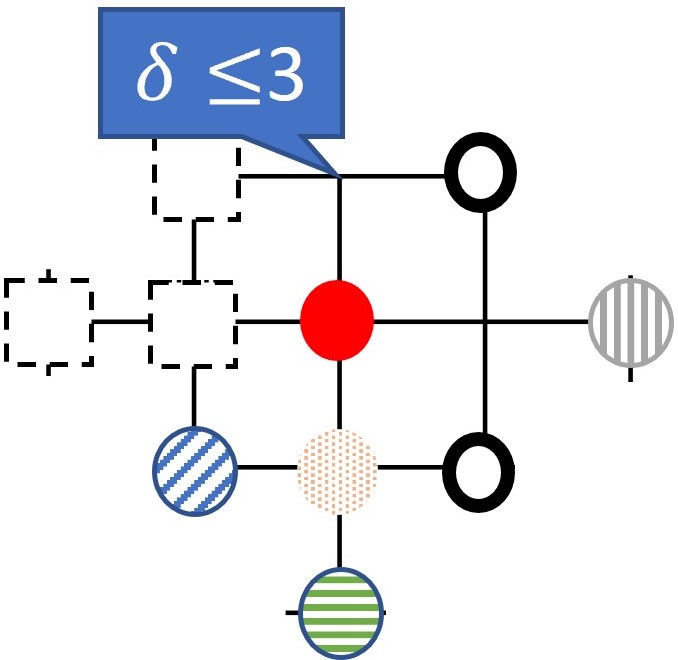}}\hspace{0.4cm}
    \subfigure[{\sf StopP2E}]{\includegraphics[height=2.1cm]{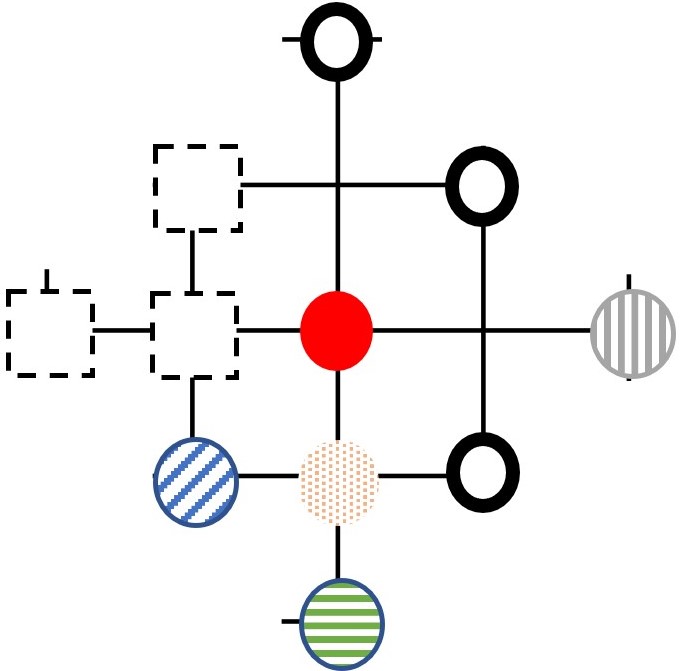}}
    \caption{Definition of views in ${\it StopP2}$ while $p_2$.}
    \label{fig:p2F}
\end{figure}
Note that, in any type of views, $r_i$'s predecessor and successor can be outside $r_i$'s view.

The algorithm description is presented as  Algorithm~\ref{code:alg1}.
The number of each rule represents its priority, a smaller number denoting a higher priority.
In the following, we explain the algorithm in more detail. 
Initially, each robot $r_i$ has color $p_1$ at the Door node.
Then, if ${\it view}(r_i)$ satisfies one of view types in ${\it Door}$ (Figure~\ref{fig:p1D}), 
$r_i$ enters the grid from the Door node by Rule 0.
After that, ${\it view}(r_i)$ belongs to ${\it MovP1}$ (Figure~\ref{fig:p1M}).
In each view type of ${\it MovP1}$, the direction of ${\it Move}(r_i)$ operation by Rule 3 is as follows:
\begin{itemize}
    \item {\sf MovP1A}: $r_i$ moves through the edge with the locally maximum label.
    \item {\sf MovP1B} and {\sf MovP1C}: $r_i$ moves to the border in the direction that is opposite from its previous node.
\end{itemize}
Thus, $r_i$ moves on the first border by Rule 3 until ${\it view}(r_i)$ belongs to ${\it StopP1}$ (Figure~\ref{fig:p1F}) or {\it Change} (Figure~\ref{fig:p1C}). 
If ${\it view}(r_i)$ is in ${\it StopP1}$, then $r_i$ changes its color to $F$ by Rule 1, and finishes its execution. 
If ${\it view}(r_i)$ is in ${\it Change}$, $r_i$ changes its color to $p_2$ and strays from the first border by Rule 2.
Then, in each view type of ${\it Change}$, the direction of ${\it Move}(r_i)$ operation is as follows:
\begin{itemize}
    \item {\sf ChangeA}:  $r_i$ moves toward the other border node than its previous node. 
    \item {\sf ChangeB} and {\sf ChangeC}: $r_i$ moves to the line. 
\end{itemize}
After $r_i$ strays from the first border, ${\it view}(r_i)$ belongs to ${\it MovP2}$ (Figure~\ref{fig:p2M}), and $r_i$ moves along a line by Rule 5 until ${\it view}(r_i)$ belongs to ${\it StopP2}$ (Figure~\ref{fig:p2F}).
Then, in each view type of ${\it MovP2}$, $r_i$ moves in a direction opposite from its previous node.
If ${\it view}(r_i)$ belongs to ${\it StopP2}$, then $r_i$ changes its color to $F$, and finishes its execution by Rule 4. 

\begin{algorithm}[t]
\begin{tabbing}
\quad\=\qquad\=\qquad\=\qquad\=\quad\=\qquad\=\quad\=\kill
{\bf Colors}~ $F$, $p_1$, $p_2$\\
{\bf Iinitialization}~ $c(r_i)=p_1$\\
{\bf Rules on node $v$ of robot $r_i$}\\
{\bf if} $c(r_i)=p_1$ {\bf then}$\{$\\
\> {\bf Rule 0}: {\bf if }${\it view}(r_i)\in {\it Door}$ {\bf then} ${\it Enter\_Grid}(r_i)$;\\
\> {\bf Rule 1}: {\bf if }${\it view}(r_i)\in {\it StopP1}$ {\bf then} $c(r_i):=F$;\\
\> {\bf Rule 2}: {\bf if }${\it view}(r_i)\in{\it Change}$ {\bf then} $c(r_i):=p_2$; ${\it Move}(r_i)$;\\
\> {\bf Rule 3}: {\bf if }${\it view}(r_i)\in {\it MovP1}$ {\bf then} ${\it Move}(r_i)$;\\
$\}$\\
{\bf if} $c(r_i)=p_2$ {\bf then}$\{$\\
\> {\bf Rule 4}: {\bf if }${\it view}(r_i)\in {\it StopP2}$ {\bf then} $c(r_i):=F$;\\
\> {\bf Rule 5}: {\bf if }${\it view}(r_i)\in {\it MovP2}$ {\bf then} ${\it Move}(r_i)$;\\
$\}$
\end{tabbing}
\caption{Algorithm $\mathcal{MIS\_LAB}$ for a maximum independent set placement with 3 colors light.}
\label{code:alg1}
\end{algorithm}

\subsubsection{Proof of Correctness}
First, we show that robots cannot collide.
\begin{lemma}\label{collide}
Robots cannot collide when executing $\mathcal{MIS\_LAB}$.
\end{lemma}
\begin{proof}
If there exists an outdated robot $r_o$ that is to move using its outdated view, the outdated view is one of the types in ${\it Door}$, ${\it Change}$, ${\it MovP1}$, or ${\it MovP2}$ by the definition of the algorithm.
Thus, if a collision with $r_o$ occurs, then $r_o$'s view is one of the types in ${\it Door}$, ${\it Change}$, ${\it MovP1}$, or ${\it MovP2}$.
In that case, because a Finished robot does not move at all, a non-Finished robot in $r_o$'s view may have moved, 
or other non-Finished robot came into the visible region of $r_o$.

On the Door corner, if a robot $r_i$ cannot see other robots, then its view is {\sf MovP1A}. 
Then, because $r_i$ has $c(r_i)=p_1$ initially, its view becomes {\sf MovP1B}.
By the definition of {\sf MovP1B}, $r_i$ moves only on the first border according to the degree of nodes until it arrives at the end of the first border, or it can see Finished robots.
Then, the first border is one-way because each robot can recognize its previous node.
Additionally, if $r_i$ can see other non-Finished robots than its successor, then $r_i$ cannot move because there is no such rule. 
That is, $r_i$ keeps the distance from its predecessor (if it exists) two or more hops on the first border.
Thus, on the first border (while $c(r_i)=p_1$), if $r_i$ has an outdated view in ${\it Door}$, ${\it Change}$ or ${\it MovP1}$, the current configuration can be the same view type as its outdated view, because only $r_i$'s successor is allowed to move toward $r_i$.
Thus, $r_i$ cannot collide with other robots while $c(r_i)=p_1$.

Consider when $r_i$ arrives at the end of the first border, or can see Finished robots. 
\begin{itemize}
\item If ${\it view}(r_i)$ becomes {\sf StopP1A} or {\sf StopP1B}, then $r_i$ changes its color to $F$ by Rule 1. 
\item If ${\it view}(r_i)$ becomes {\sf ChangeA}, then $r_i$ changes its color to $p_2$, and moves to the second border by Rule 2.
After that, ${\it view}(r_i)$ becomes {\sf MovP2A} until it arrives at the diagonal corner (\emph{i.e.}, {\sf StopP2A}), or it can see Finished robots on the second border (\emph{i.e.}, {\sf StopP2B}).
By the definition of {\sf MovP2A}, each robot moves on the second border according to the degree of nodes, and the second border is one-way.
By the definition of the algorithm, there is no rule to make $r_i$ stray from the second border.
Then, $r_i$ keeps the distance from its predecessor (if it exists) two or more hops on the second border. 
\item If ${\it view}(r_i)$ becomes {\sf MovP1C}, then $r_i$ moves to the adjacent node occupied by a Finished robot on the first border and ${\it view}(r_i)$ becomes {\sf ChangeB}.
\item If ${\it view}(r_i)$ becomes {\sf ChangeB} or {\sf ChangeC}, then $r_i$ changes its color to $p_2$, and moves to one of the lines. Without loss of generality, let the line be $m$-line where $m>0$. 
Then, robots on $(m-1)$-line are Finished and $(m+1)$-line is empty (if it exists on the grid) by the definition of {\sf ChangeB} or {\sf ChangeC}. 
Thus, after that, ${\it view}(r_i)$ becomes {\sf MovP2B} or {\sf MovP2C} until it arrives at the end of $m$-line (\emph{i.e.}, {\sf StopP2C} or {\sf StopP2D}), or it can see a Finished robot on $m$-line (\emph{i.e.}, {\sf StopP2E}).
By the definition of the algorithm, there is no rule to make $r_i$ stray from $m$-line.
By the definitions of {\sf MovP2B} and {\sf MovP2C}, $m$-line is also one-way, and $r_i$ keeps the distance from its predecessor (if it exists on the line) two or more hops.
\end{itemize}
In any case, on each line (while $c(r_i)=p_2$), if $r_i$ has an outdated view in ${\it Change}$ or ${\it MovP2}$, then the current configuration is the same view type as its outdated view, because only $r_i$'s successor is allowed to move toward $r_i$.
Thus, $r_i$ cannot collide with other robots while $c(r_i)=p_2$.

Thus, the lemma holds.
\end{proof}


Next, we show that $\mathcal{MIS\_LAB}$ constructs a maximum independent set.
\begin{lemma}\label{1st}
The first robot $r_1$ moves to the diagonal corner, and $c(r_1)$ becomes $F$ on the corner. 
\end{lemma}
\begin{proof}
When the first robot $r_1$ is in the Door node, then its view is {\sf DoorA}.
Thus, it moves to the Door corner by ${\it Enter\_Grid}(r_1)$, and its view becomes {\sf MovP1A}.
Then, $r_1$ moves to the adjacent node through the edge with the maximal edge label by Rule 3.
Then, because it is the first robot, its view becomes {\sf MovP1B} and moves to the adjacent node on the first border by Rule 3.


By the proof of Lemma~\ref{collide}, the first border and second border are one-way, and any other robot cannot pass $r_1$ on these borders.
Thus, ${\it view}(r_1)$ remains {\sf MovP1B}, and $r_1$ moves on the first border according to the node degree.
Therefore, $r_1$ arrives at the end of the first border eventually, and then ${\it view}(r_1)$ becomes {\sf ChangeA}.
After that, by Rule 2, $r_1$ changes its color to $p_2$ and moves to the adjacent node on the second border.
${\it view}(r_1)$ becomes {\sf MovP2A}, $r_1$ moves towards the next (diagonal) corner through the second border according to the node degree by Rule 5, eventually ${\it view}(r_1)$ becomes {\sf StopP2A}.
Because $c(r_1)=p_2$ holds, $r_1$ changes its color to $F$ on the corner by Rule 4.

Thus, the lemma holds.
\end{proof}

\begin{lemma}\label{1stB}
The first $\lceil l/2 \rceil$ robots move to the second border, and their color becomes $F$.
Additionally, nodes on the second border are occupied by a robot or empty alternately from the diagonal corner. 
\end{lemma}
\begin{proof}
By Lemma~\ref{1st}, the first robot $r_1$ eventually becomes Finished on the diagonal corner.

First, we consider the second robot $r_2$, which follows $r_1$ in the case that $l>3$.
By the assumption, $r_2$ appears to the Door node just after $r_1$ enters into the grid.
Because $r_1$ does not become Finished before it arrives at the diagonal corner, $r_2$ can move from the Door node only when ${\it view}(r_2)={\sf DoorA}$ holds.
After that, by the definition of the algorithm, $r_2$ moves in the same way as $r_1$, $c(r_2)$ becomes $p_2$ on the end of the first border eventually and $r_2$ moves on the second border.
Finally, $r_2$ can see $r_1$ on the diagonal corner two hops away.
Then, there is no rule such that $r_2$ executes before $c(r_1)$ becomes $F$.
Because of Lemma~\ref{1st}, ${\it view}(r_2)$ eventually becomes {\sf StopP2B}.
Then, by Rule 4, $c(r_2)$ becomes $F$.

In the case that $l=3$, when $r_2$ arrives at the end of the first border and $r_1$ becomes Finished, ${\it view}(r_2)$ becomes {\sf StopP1B} and $c(r_2)$ becomes $F$ by Rule 1.
Note that, in any case, the distance between $r_1$ and $r_2$ is two hops when they are Finished.

For the successors of $r_2$, we can discuss their movements in the same way as $r_2$.
Thus, by the definitions of {\sf StopP1B} and {\sf StopP2B}, 
the distance between a robot and its successor is two hops when they are Finished on the second border.
Therefore, on the second border, beginning with the diagonal corner, every even node is occupied, and the number of robots is $\lceil l/2 \rceil$.
If $l$ is odd, when the $\lceil l/2 \rceil$-th robot arrives at the end of the first border, its view becomes {\sf StopP1B} and it changes its color to $F$ by Rule 1.
Otherwise, it changes its color to $p_2$ and moves to the second border.

Thus, the lemma holds.
\end{proof}

\begin{figure}[t]
\centering
\subfigure[even-even dimensions]{
\centering
\includegraphics[height=.07\textheight]{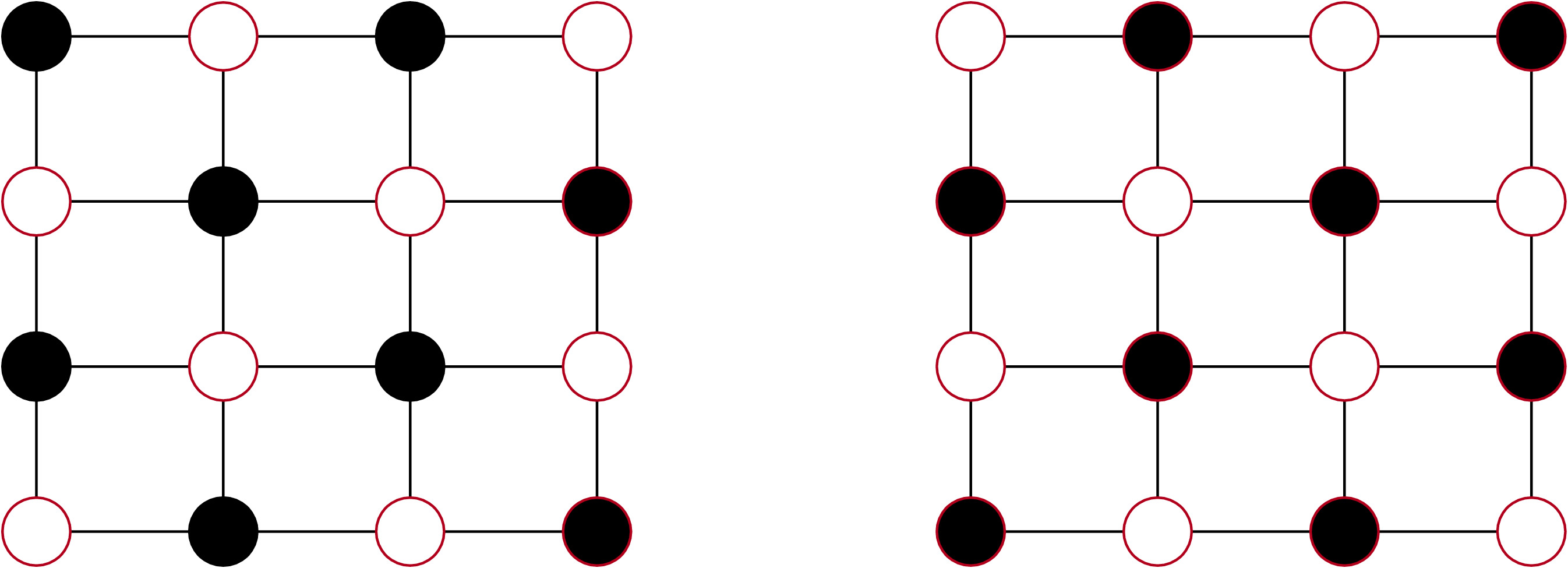}
}\hspace{1cm}
\subfigure[even-odd dimensions]{
\centering
\includegraphics[height=.09\textheight]{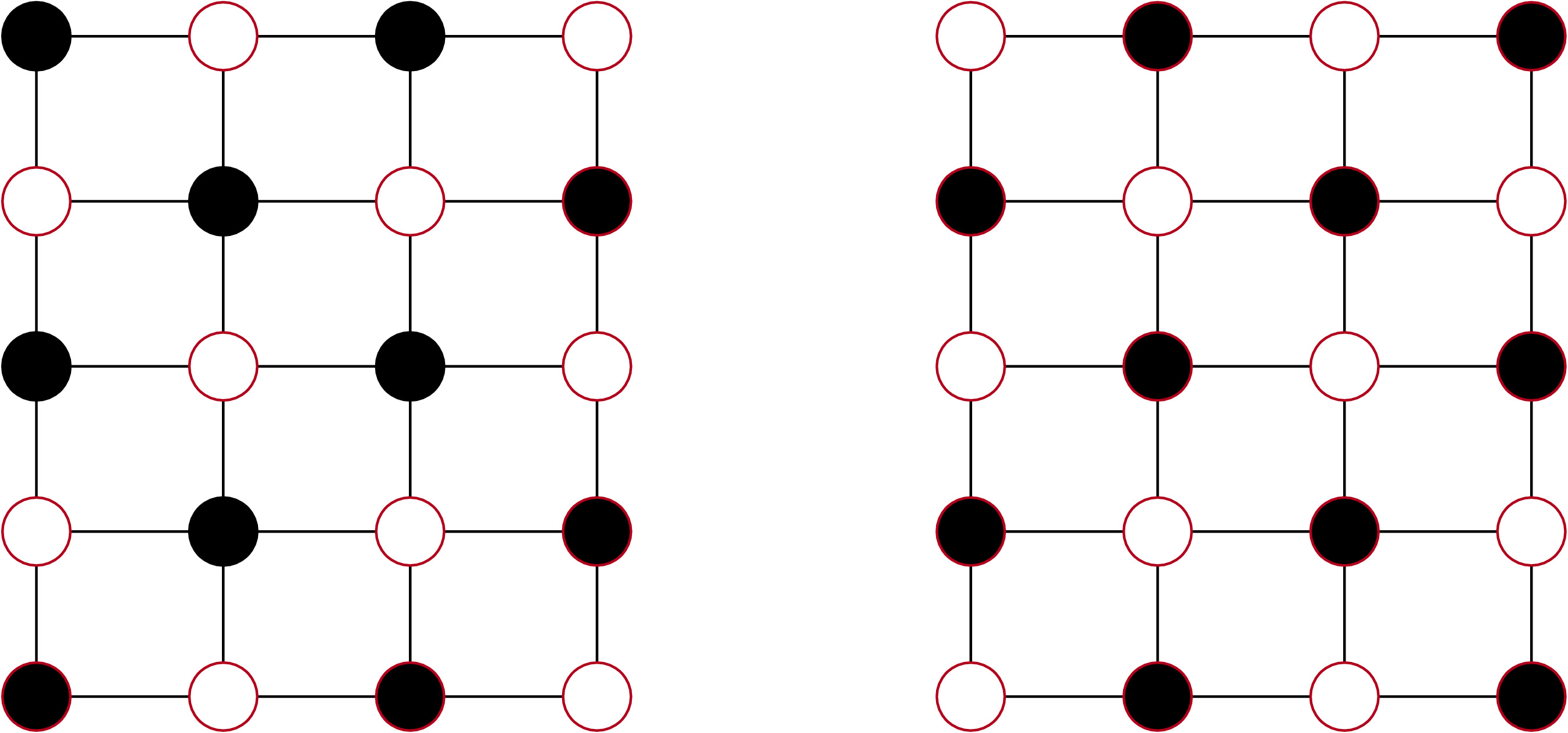}
}
\subfigure[odd-odd dimensions]{
\centering
\includegraphics[height=.09\textheight]{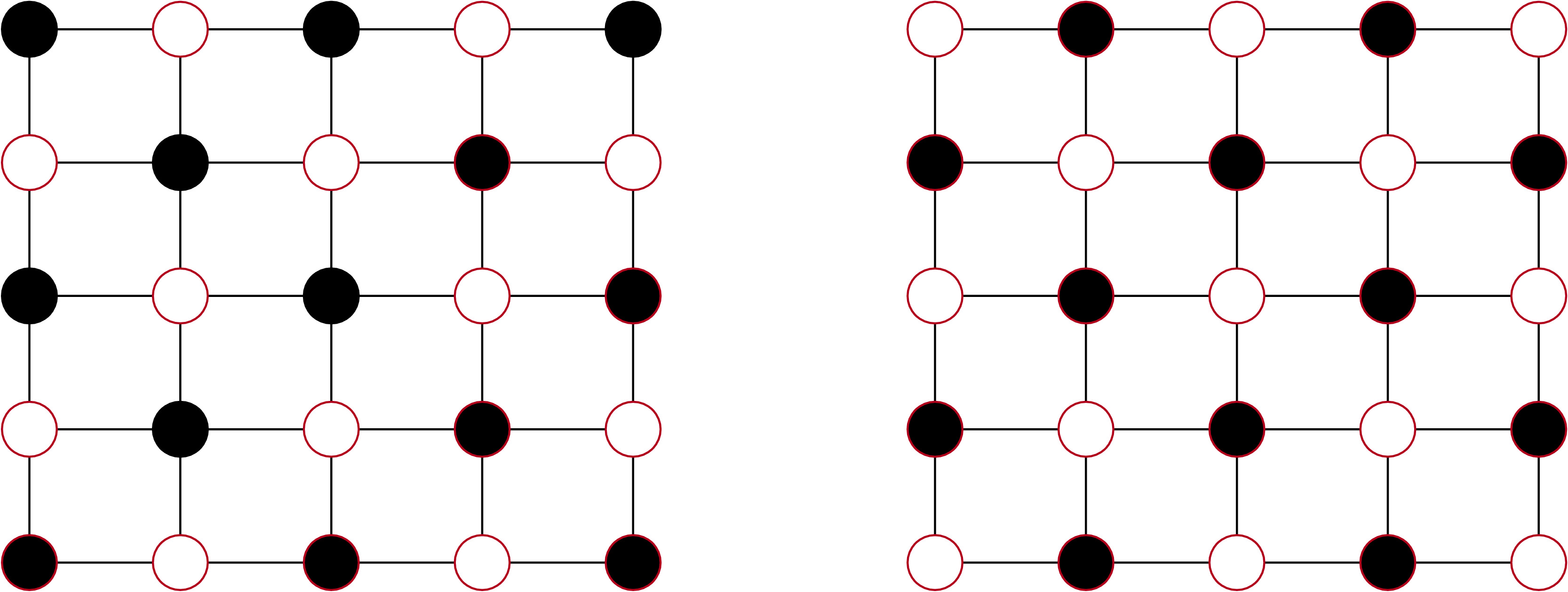}
}
\caption{Checkers patterns in grids.}
\label{fig:checkers}
\end{figure}

\begin{lemma}\label{1-line}
From the $(\lceil l/2 \rceil+1)$-th to the $l$-th robots, each robot moves to the $1$-line, and its color becomes $F$.
Additionally, nodes on the $1$-line are empty or occupied by a robot alternately, beginning with an empty node.
\end{lemma}
\begin{proof}
By Lemma~\ref{1stB}, $\lceil l/2 \rceil$ robots on $0$-line eventually become Finished.
Let $r_i$ be the $(\lceil l/2 \rceil+1)$-th robot, $r_j$ be the $(\lceil l/2 \rceil)$-th robot that is the predecessor of $r_i$.
$r_i$ moves from the Door node in the same way as $r_j$ while $c(r_i)=p_1$.
Because robots on $0$-line become Finished eventually, one of the following two
cases occurs: \emph{(1)} if $l$ is odd, ${\it view}(r_i)$ becomes {\sf MovP1C}, because the end of the first border is occupied by $r_j$, or \emph{(2)} if $l$ is even, ${\it view}(r_i)$ becomes {\sf ChangeC}, because the end of the first border is empty but its adjacent node on the $0$-line is occupied by $r_j$.

In case \emph{(1)}, by Rule 3, $r_i$ moves to the node in front of the end of the first border, ${\it view}(r_i)$ becomes {\sf ChangeB}.
Then, by Rule 2, $c(r_i)$ becomes $p_2$ and $r_i$ moves to $1$-line.
After that, if $l=3$, ${\it view}(r_i)$ becomes {\sf StopP2D} and $r_i$ changes its color to $F$ by Rule 4.
Otherwise, because $r_i$ can see Finished robots on $0$-line, ${\it view}(r_i)$ becomes {\sf MovP2B}.
Then, because the nodes on $0$-line are occupied alternately by Lemma~\ref{1stB}, ${\it view}(r_i)$ becomes {\sf MovP2B} and {\sf MovP2C} alternately by the execution of Rule 5.
Thus, $r_i$ moves toward the other side border that is parallel to the first border by Rule 5, and ${\it view}(r_i)$ eventually becomes {\sf StopP2D} because the diagonal corner is occupied by a Finished robot (Lemma~\ref{1st}). 
Then, by Rule 4, $c(r_i)$ eventually becomes $F$.
Because $l$ is odd, $\lfloor l/2 \rfloor-1$ successors of $r_i$ follow $r_i$, and eventually, their views become {\sf StopP2E}, and they change their colors to $F$ by Rule 4 on $1$-line.

In case \emph{(2)}, $r_i$ also changes its color to $p_2$, and moves to $1$-line by Rule 2.
After that, because $r_i$ can see Finished robots on $0$-line, ${\it view}(r_i)$ becomes {\sf MovP2C}. 
Then, in the same way as for case \emph{(1)}, $l/2-1$ robots including $r_i$ become Finished on $1$-line.
After that, the view of the next robot $r_l$ ($l$-th robot) becomes {\sf StopP1A} on the intersection between the first border and $1$-line, and $r_l$ becomes Finished by Rule 1. 

Thus, the lemma holds.
\end{proof}

\begin{lemma}\label{2hops}
The distance between any two robots on the grid is two hops after every robot becomes Finished.
\end{lemma}  
\begin{proof}
By 
the definitions of ${\it StopP1}$ and ${\it StopP2}$, the distance between a robot $r_i$ and its predecessor is two hops after $r_i$ becomes Finished if the predecessor is on the same line as $r_i$ stays.
Thus, when the robots on $m$-line ($0<m<L-1$) become Finished, if there are two adjacent Finished robots, then there is a robot $r_r$ on $m$-line that cannot move from the node that is adjacent to a node occupied by a Finished robot on $(m-1)$-line.
However, by the same argument as in Lemmas~\ref{1stB} and \ref{1-line}, if $m$ is odd (resp. even), the nodes on $m$-line are occupied alternately beginning with an empty node (resp. occupied node) because the nodes on $(m-1)$-line are also occupied alternately beginning with an occupied node (resp. empty node).
Thus, before such $r_r$ becomes Finished, $r_r$ has a view of type {\sf MovP2C} and can move by Rule 5, \emph{i.e.}, it cannot exist.

Now, to consider the end of the execution of the algorithm, we consider $(L-1)$-line when nodes on $(L-2)$-line are occupied by Finished robots. 
The $(L-1)$-line is a border connected to the Door node.
Then, if both $l$ and $L$ are odd or both are even, the view from the Door node becomes {\sf DoorA}, otherwise {\sf DoorB} (See the left figures of Figure~\ref{fig:checkers}).
\begin{itemize}
\item If the view from the Door node is {\sf DoorA}, the robot on the Door node moves to the Door corner by Rule 0.
Then, the view from the Door corner is {\sf ChangeC}.
By the above discussion, the view from the Door corner eventually becomes {\sf StopP1A}, thus the final robot on the Door corner becomes Finished by Rule 1.
Then, any other robot cannot enter into the grid because there is no such rule.
\item If the view from the Door node is {\sf DoorB}, the view from the Door corner is {\sf ChangeB}.
Then, the empty node $v$ that is adjacent to the Door corner is eventually occupied by a Finished robot on $(L-1)$-line.
After that, any other robots on the Door node cannot enter into the grid because there is no such rule. 
\end{itemize}
Thus, the lemma holds.
\end{proof}

\begin{lemma} 
Every robot on the grid is eventually Finished.
\end{lemma}
\begin{proof}
By the proofs of Lemmas~\ref{collide}-\ref{2hops}, the transitions of the view type of each robot are shown as Figure~\ref{fig:Alg1}. 
Thus, the lemma holds.
\begin{figure}[bt]
    \centering
    \includegraphics[width=0.7\textwidth]{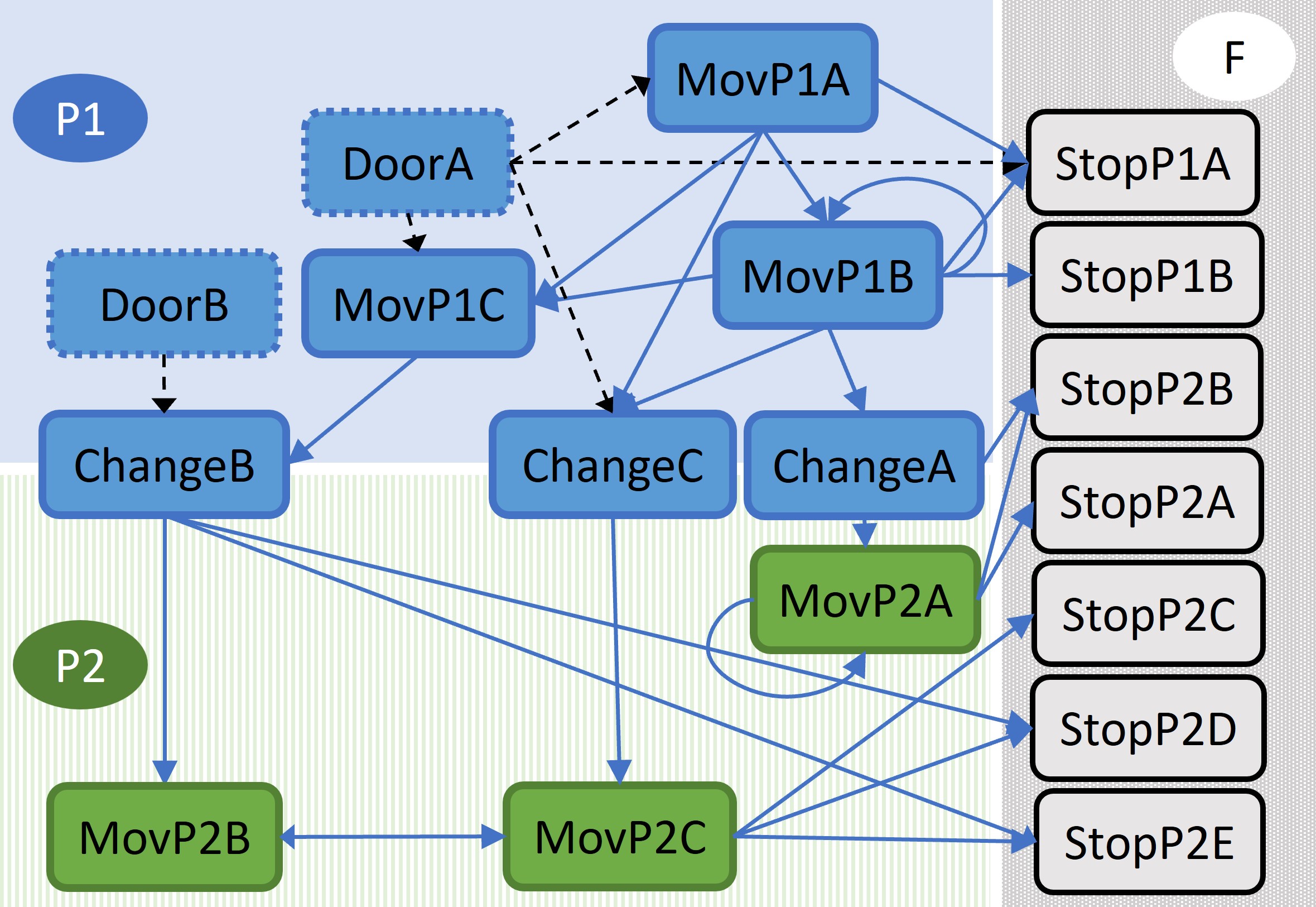}
    \caption{View type transitions of $\mathcal{MIS\_LAB}$. Each solid arrow (resp. dotted arrow) represents a transition by a {\it Move} (resp. {\it Enter\_Grid}) operation. }
    \label{fig:Alg1}
\end{figure}
\end{proof}

\begin{theorem}\label{1MIS}
Algorithm $\mathcal{MIS\_LAB}$ constructs a maximum independent set of occupied locations on the grid.
\end{theorem}
\begin{proof}
By Lemma~\ref{2hops}, distances between any two occupied nodes are two. On a grid, only checkers patterns satisfy this constraint. 
When at least one dimension of even, there are as many occupied locations as non-occupied locations, so any checkers pattern is a maximum independent set (see Figure~\ref{fig:checkers}(a) and Figure~\ref{fig:checkers}(b)). When both dimensions are odd, there may be either one more occupied locations than non-occupied locations, or the contrary (See Figure~\ref{fig:checkers}(c)). The situation that corresponds to the maximum independent set is the one with occupied locations in the corners, which is what our algorithm constructs. 
Hence, the theorem holds.
\end{proof}

\begin{lemma}\label{number}
When a maximum independent set is constructed, $\lceil n/2 \rceil$ robots are on the grid.
\end{lemma}
\begin{proof}
By the proofs of Lemmas~\ref{1stB}--\ref{2hops}, nodes on the even-numbers (resp. odd-numbers) lines are occupied by $\lceil l/2\rceil$ (resp. $\lfloor l/2 \rfloor$) robots.
Therefore, if $L$ is even, the number of robots in the maximum independent set is $lL/2$.
If $L$ is odd, the number of robots in the maximum independent set is $l\lfloor L/2\rfloor+\lceil l/2\rceil$.
Because $n=lL$, the lemma holds.
\end{proof}

To analyze the time complexity of the algorithm, we count the sum of individual executions of rules.

\begin{theorem}\label{time}
The time complexity by $\mathcal{MIS\_LAB}$ is $O(n(L+l))$ steps.
\end{theorem}
\begin{proof}
The first robot moves $L+l-1$ steps and becomes Finished, thus it executes $L+l$ steps. 
The first robot moves the longest way.
Therefore, by Lemma~\ref{number}, the sum of the number of steps 
is $O(n(L+l))$.
Thus, the theorem holds.
\end{proof}

\subsection{Algorithm $\mathcal{MIS\_NOLAB}$ with 7 colors lights, and $\phi=3$}\label{sec:7colors}
In this section, we relax both additional hypotheses made in Section~\ref{sec:3colors}. So, there is no local labeling of edges, and robots cannot recognize the node they came from when at a particular node. 
Instead, we assume that $\phi=3$, and that seven light colors are available for each robot $r_i$, whose colors are named $F$, $(p_1,k_i)$, and $(p_2,k_i)$ ($k_i\in\{0, 1, 2\}$).
The value of $k_i$ represents the \emph{order} of the robot (the notion of order is explained in detail in the sequel).
Initially, the color of light $c(r_i)$ for each robot $r_i$ on the Door node is $(p_1,0)$, that is, $k_i=0$.

The strategy to construct a maximum independent set is the same as $\mathcal{MIS\_LAB}$.
The light colors $p_1$ and $p_2$ are used to distinguish which border or corner the robot stays. 
However, on the Door corner, the first robot chooses an adjacent node arbitrarily (this choice is decided by an adversary), 
and other robots just follow it while each robot keeps its predecessor in its view.
This strategy is similar to of the existing algorithms \cite{filling2,filling2008,filling2013} for the filling problem (\emph{i.e.}, following the reader strategy), yet those algorithms assume edge labeling or common directions. 
As a byproduct of our approach, filling algorithms based on the ``following the leader'' construction \cite{filling2,filling2008,filling2013} can be improved using our scheme to remove their unnecessary labeling or common directions assumptions.
In the following, we call the border chosen by the first robot on the Door corner ``the first border", and the border that is connected to the first border and is not connected to the Door corner ``the second border".
Without loss of generality, we define the size of these borders and lines in the same way as $\mathcal{MIS\_LAB}$.

By the strategy of the routing described above, each robot enters the grid one-by-one and walks in line on the grid.
Therefore, each robot has a successor, and each robot except the first one has a predecessor. 
In this algorithm, each robot has a variable $k_i$ to distinguish them.
On the Door node, each robot $r_i$ sets its $k_i$ as follows.
If $r_i$ is the first robot, $r_i$ maintains $k_i=0$.
Otherwise, if $r_i$'s predecessor robot $r_j$ on the Door corner has $(p_1,k_j)$, then $r_i$ sets $k_i$ to $(k_j+1)\mod 3$.
Then, $r_i$ retains the value of $k_i$ in $c(r_i)$ until $c(r_i)$ becomes $F$.
On the Door corner, each robot $r_j$ waits for its successor $r_i$ on the Door node to set its value $k_i$ before $r_j$ moves. 
By this mechanism, each robot $r_i$ recognizes that its neighboring non-Finished robot $r_j$ with smaller (resp. larger)\footnote{If $k_i=2$ (resp. 0, 1), it is larger than 1 (resp. 2, 0), but smaller than 0 (resp. 1, 2).} $k_j$ value than $k_i$ is its predecessor (resp. successor).
Let ${\it SetC}(r_i)$ be the operation such that $c(r_i):=(p_1,(k_j+1)\mod 3)$, where $r_j$ is the robot on the Door corner and $c(r_j)=(p_1,k_j)$.
In $\mathcal{MIS\_NOLAB}$, only the first robot selects its way arbitrarily from the Door corner.
After that, the other robot $r_i$ can move only when the distance from its predecessor is three (unless the predecessor becomes Finished) and the distance from its successor is two (unless $r_i$ is not on the Door corner).
By this mechanism, $r_i$ can recognize which border is the first border chosen by the first robot on the Door corner.

\begin{table}[t]
    \centering
    \caption{Notations for $\mathcal{MIS\_NOLAB}$.}
   \includegraphics[width=10cm]{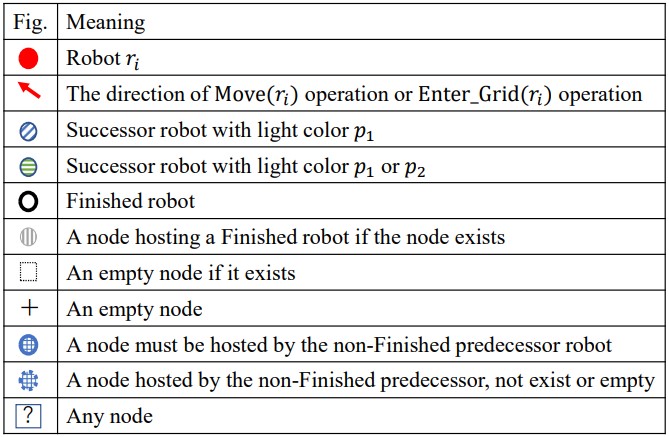}
    \label{note2}
\end{table}

To describe $\mathcal{MIS\_NOLAB}$, we use the definitions of some view types in Figures~\ref{fig:my_label}--\ref{fig:p2F3}. 
The notations in these figures are summarized in Table~\ref{note2}.
Unlike $\mathcal{MIS\_LAB}$, we do not use dotted circles without frames, since the previous node can no longer be recognized by the robot.
The circle with diagonal stripes or horizontal stripes represents $r_i$'s successor robot, which must be there. 
If the successor robot is on the diagonal (resp. horizontal) striped node, it has $p_1$ (resp. $p_1$ or $p_2$).
While there are two types of successor in each view type of Figures~\ref{fig:p2M3}--\ref{fig:p2F3}, exactly one of them must be present.
The waffle circle represents $r_i$'s non-Finished predecessor robot.
If the waffle circle is with a thick border, the predecessor must be there.
Otherwise, it may be an empty node or non-existent node.
For example, in Figure~\ref{fig:p1D3}(b), when the predecessor has just become Finished on the upper node with the thick white circle, the waffle circle with the dotted border actually represents an empty node.
The square with a question mark represents any node.


\begin{figure}[t]
    \centering
    \includegraphics[height=2cm]{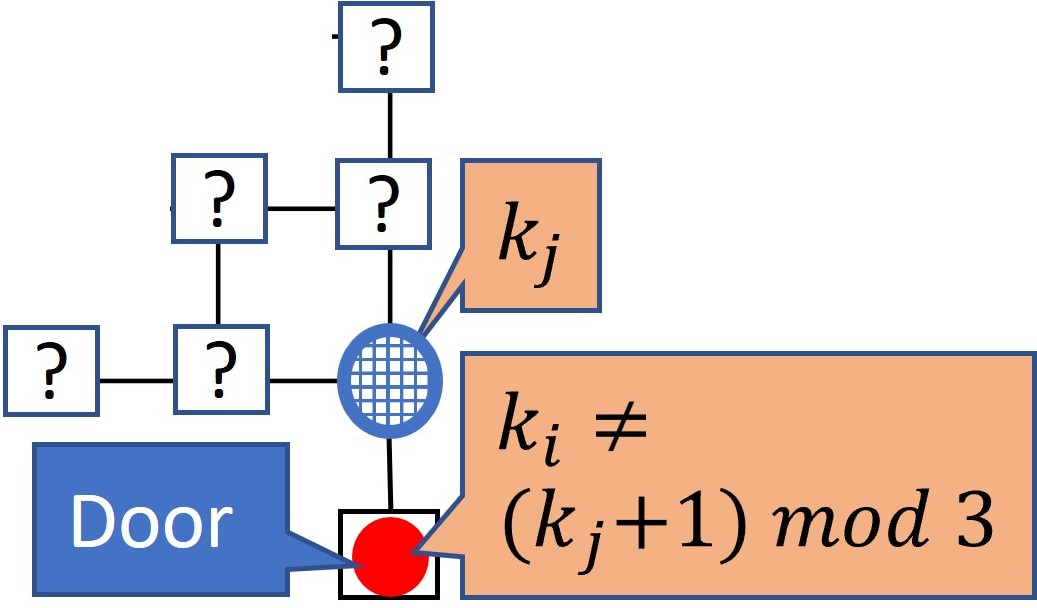}
    \caption{Definition of {\sf Door0}}
    \label{fig:my_label}
\end{figure}

Next, we explain the definition of each view type (except the direction of the movement) in detail.
While $c(r_i)=(p_1,k_i)$, we use the view types in Figures~\ref{fig:my_label}-\ref{fig:p1C3}.
A view type {\sf Door0} in Figure~\ref{fig:my_label} represents a view such that $r_i$ is on the Door node before $k_i$ is set.
In set ${\it Door}^\prime$ in Figure~\ref{fig:p1D3}, there are four types of views such that $r_i$ is on the Door node after $k_i$ is set:
\begin{itemize}
    \item {\sf DoorA}: $r_i$ cannot see Finished robots, and the distance to non-Finished robot is at least three.
    \item {\sf DoorB}: $r_i$ can see a Finished robot only on a border two hops away, and the distance to non-Finished robot is at least three.  
    \item {\sf DoorC}: $r_i$ can see a Finished robot only on a border three hops away, the other border is empty, and $r_i$ can see a non-Finished robot three hops away.
    \item {\sf DoorD}: $r_i$ can see two Finished robots three hops away, and a non-Finished robot three hops away.
\end{itemize}
\begin{figure}[t]
    \centering
\subfigure[{\sf DoorA}]{\includegraphics[height=1.8cm]{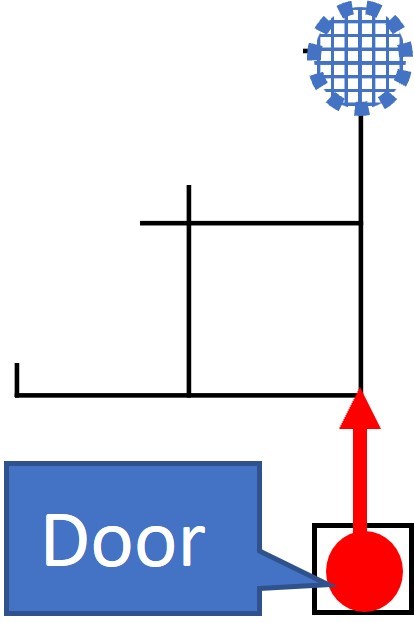}}
    \hspace{0.45cm}
\subfigure[{\sf DoorB}]{\includegraphics[height=1.8cm]{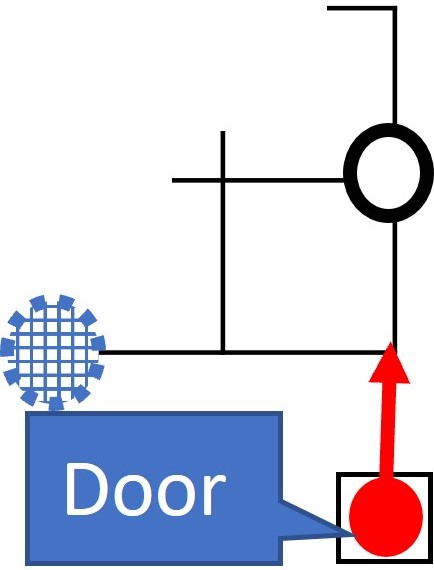}}
    \hspace{0.45cm}
\subfigure[{\sf DoorC}]{\includegraphics[height=1.8cm]{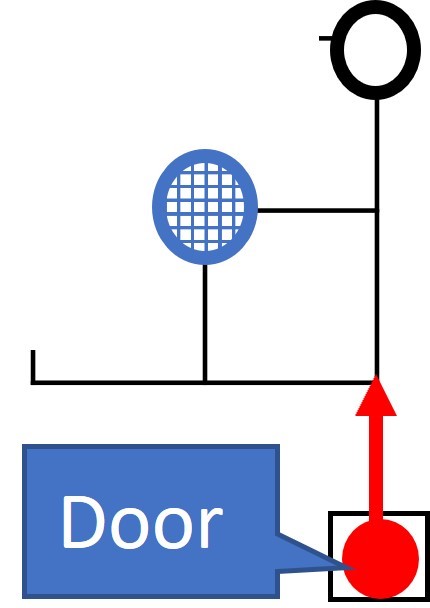}}
    \hspace{0.45cm}
\subfigure[{\sf DoorD}]{\includegraphics[height=1.8cm]{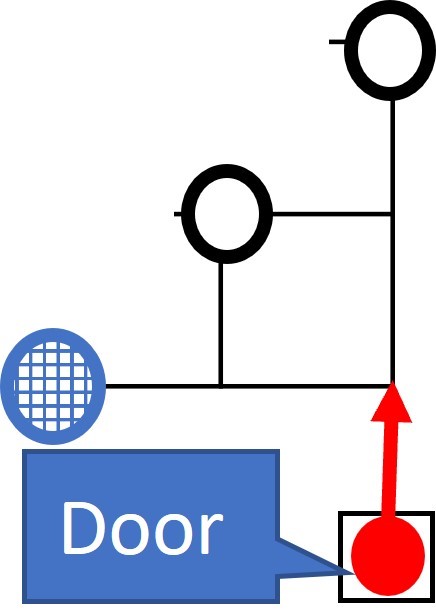}}
\caption{Definition of views in ${\it Door}^\prime$ while $p_1$.}
\label{fig:p1D3}
\end{figure}
In set ${\it MovP1}^\prime$ in Figure~\ref{fig:p1M3}, there are nine types of views:
\begin{itemize}
    \item {\sf MovP1A}: $r_i$ is on the Door corner, and can see only one non-Finished robot $r_j$ on the Door, with $k_j=(k_i+1)\mod 3$. 
    \item {\sf MovP1B}: $r_i$ on the Door corner, and can see only two non-Finished robots $r_j$ and $r_l$ such that $r_j$ is on the Door (with $k_j=(k_i+1)\mod 3$), and $r_l$ is three hops away on a border. 
    \item {\sf MovP1C} and {\sf MovP1D}: $r_i$ is on the Door corner, and can see two non-Finished robots $r_j$ and $r_l$ such that $r_j$ is on the Door (with $k_j=(k_i+1)\mod 3$) and $r_l$ is three hops away but not on a border. The difference between these view types is that, only in in the case of {\sf MovP1D}, $r_i$ can see a Finished robot not located on a border. 
    \item {\sf MovP1E}: $r_i$ is on the Door corner, can see a non-Finished robot $r_j$ on the Door (with $k_j=(k_i+1)\mod 3$), $r_i$ can see a Finished robot two hops away on a border, and the distance to other non-Finished robots is at least three.
    \item {\sf MovP1F}: $r_i$ is on a border, can see only at most two non-Finished robots on the border such that its successor $r_j$ (with $k_j>k_i$) is two hops away, and its predecessor $r_l$ (with $k_l<k_i$) is at least three hops away. 
    \item {\sf MovP1G} and {\sf MovP1H}: $r_i$ is on a border, can see two non-Finished robots such that its successor $r_j$ (with $k_j>k_i$) is two hops away, and its predecessor $r_l$ (with $k_l<k_i$) is three hops away. The difference between these view types is that, the distance from the line where $r_l$ stays is two (resp. one) in {\sf MovP1G} (resp. {\sf MovP1H}). 
    \item {\sf MovP1I}: $r_i$ is on a border, can see a Finished robot on the border two hops away, and its successor $r_j$ (with $k_j>k_i$) is two hops away. If $r_i$ can see a non-Finished predecessor $r_l$ (with $k_l<k_i$), then it is three hops away. 
\end{itemize}
\begin{figure}[p]
\centering
    \subfigure[{\sf MovP1A}]{\includegraphics[height=2.3cm]{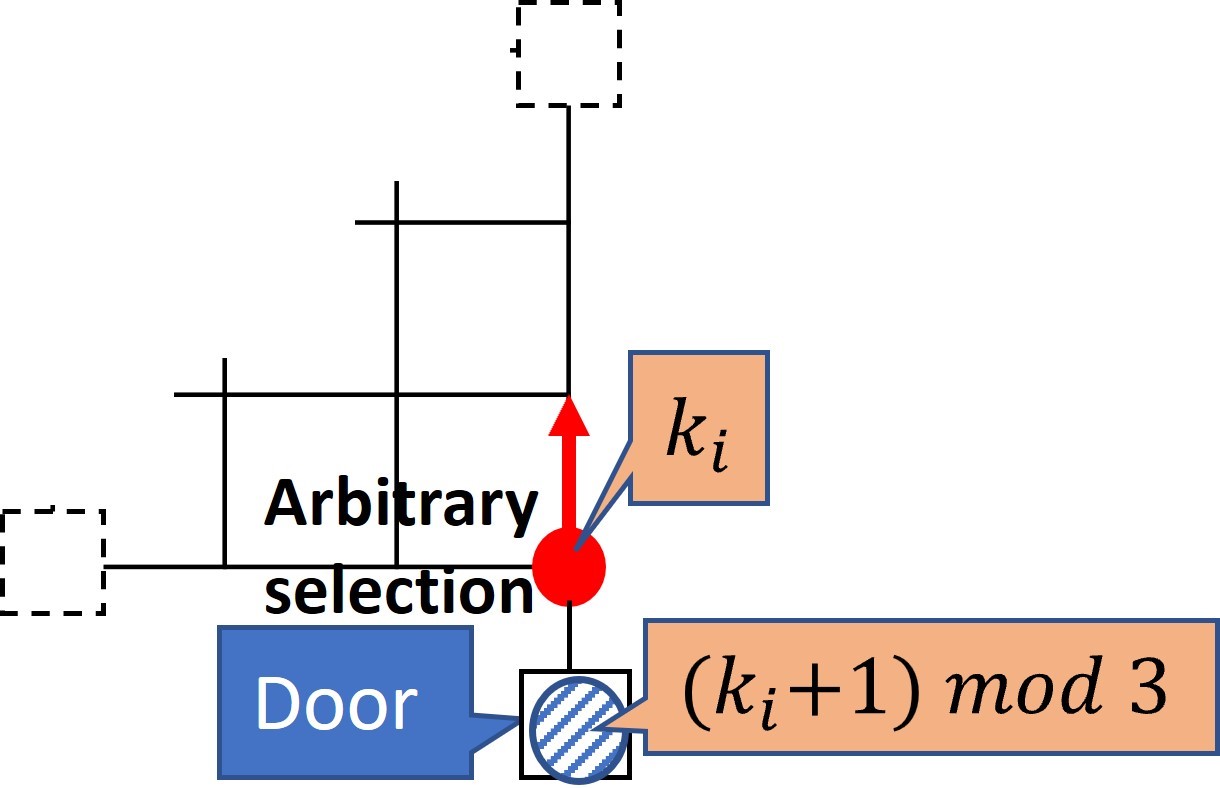}}\hspace{0.4cm}
    \subfigure[{\sf MovP1B}]{\includegraphics[height=2.3cm]{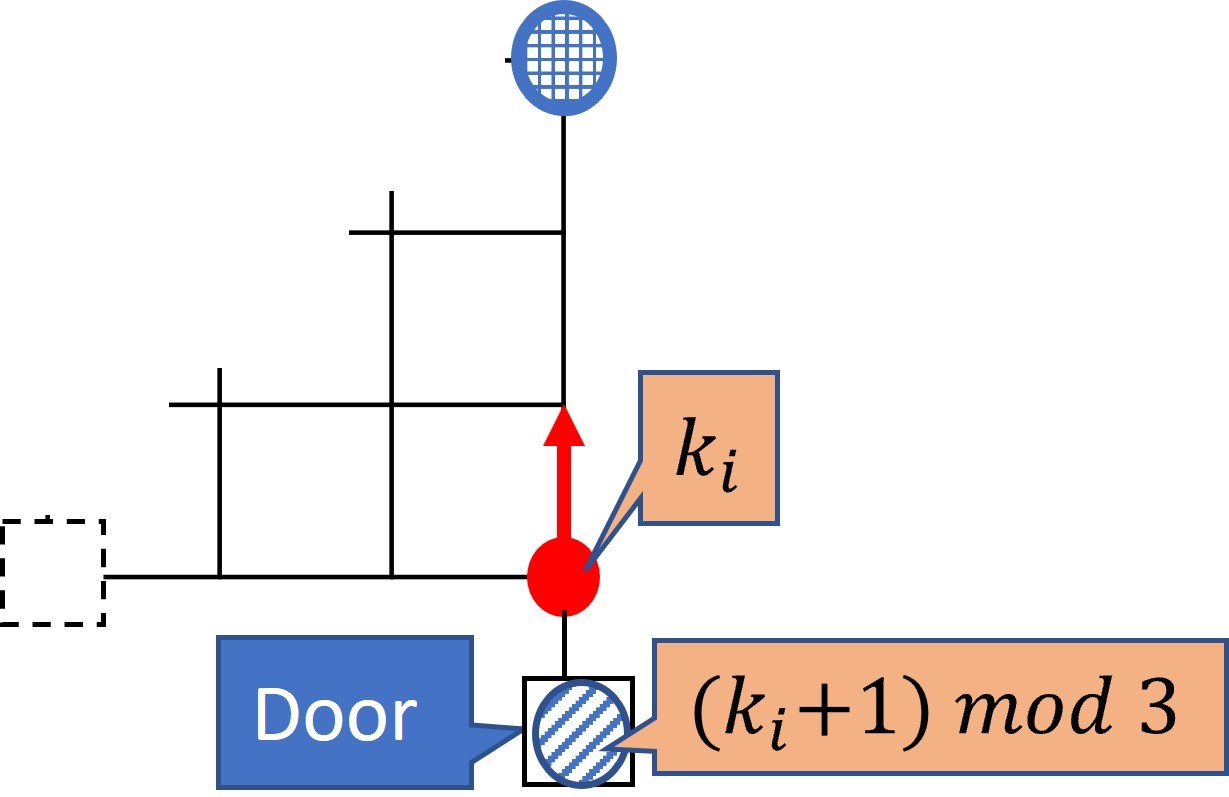}}
    \subfigure[{\sf MovP1C}]{\includegraphics[height=2.3cm]{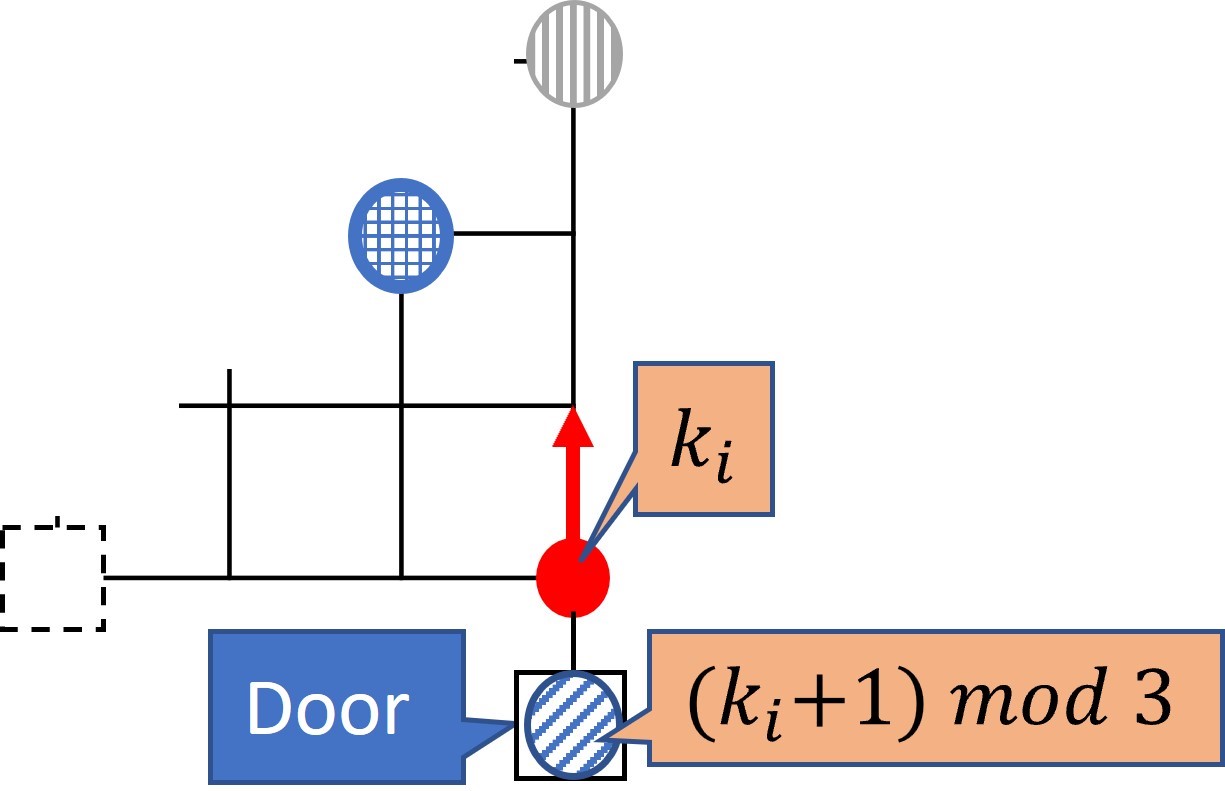}}\\
    \subfigure[{\sf MovP1D}]{\includegraphics[height=2.3cm]{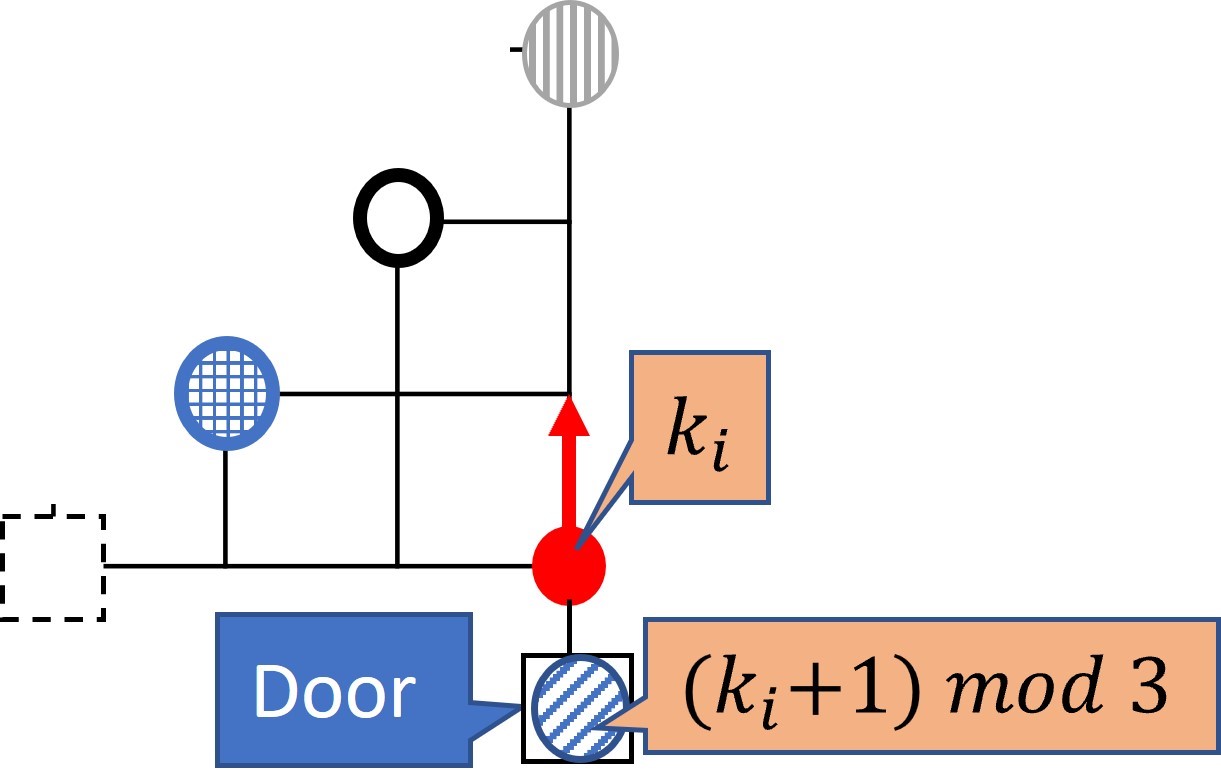}}\hspace{0.4cm}
    \subfigure[{\sf MovP1E}]{\includegraphics[height=2.3cm]{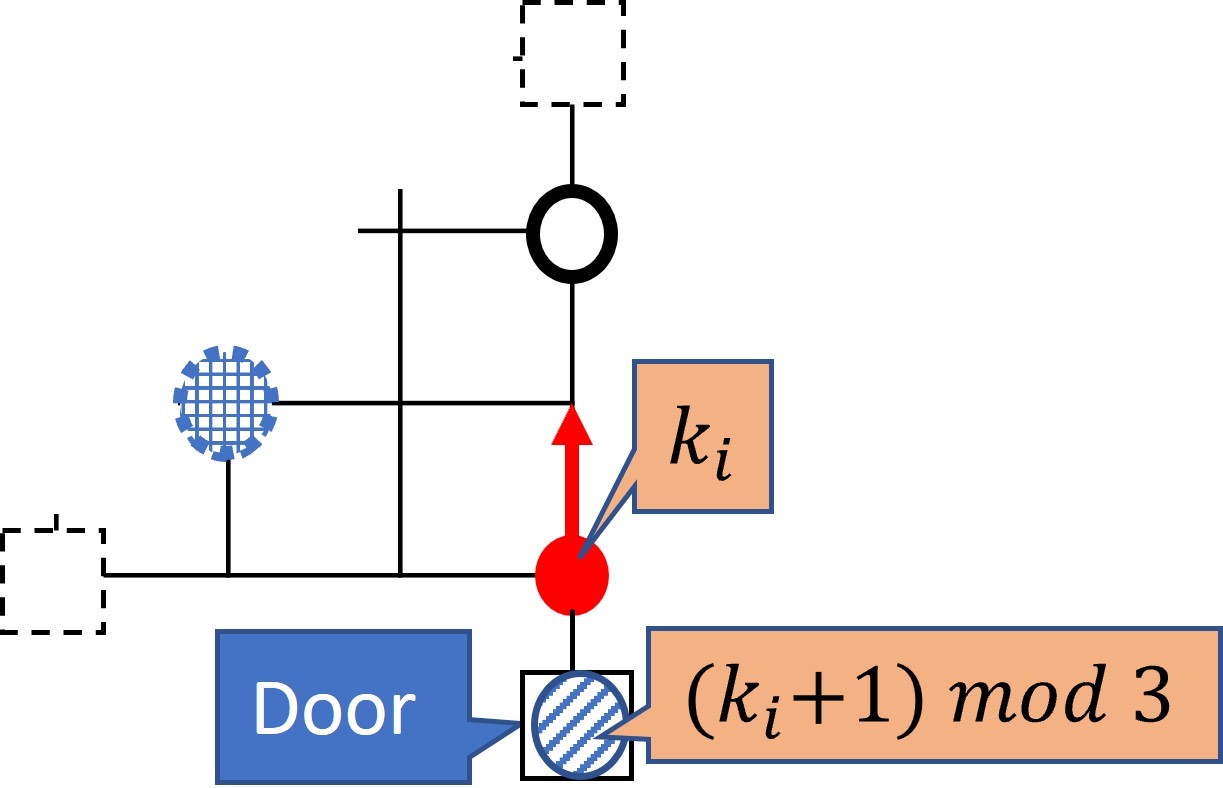}}\hspace{0.4cm}
    \subfigure[{\sf MovP1F}]{\includegraphics[height=3cm]{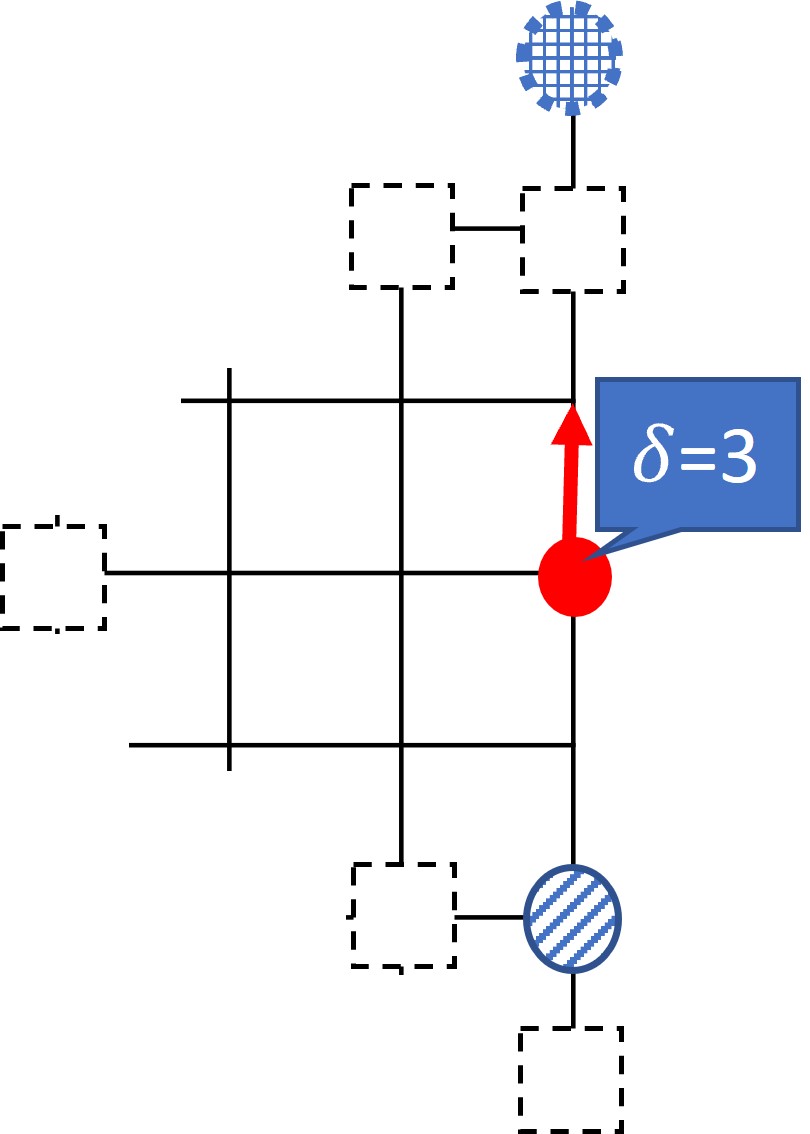}}\\
    \subfigure[{\sf MovP1G}]{\includegraphics[height=3cm]{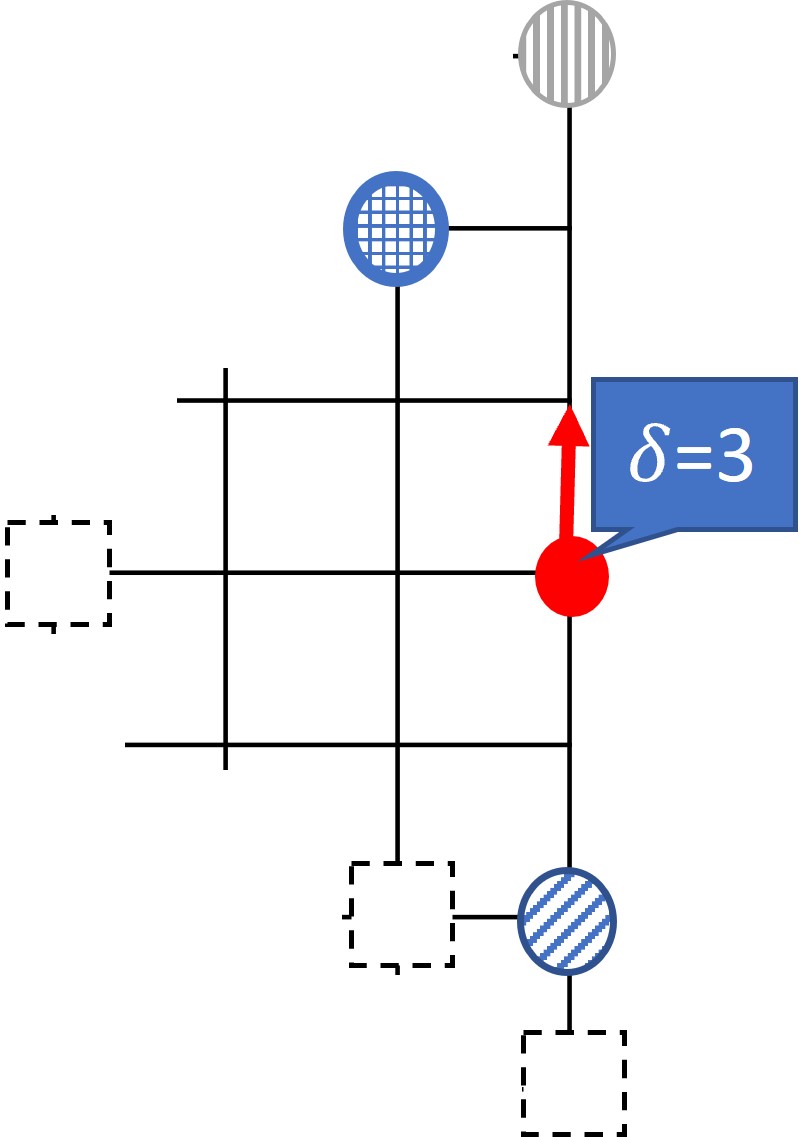}}\hspace{0.3cm}
    \subfigure[{\sf MovP1H}]{\includegraphics[height=3cm]{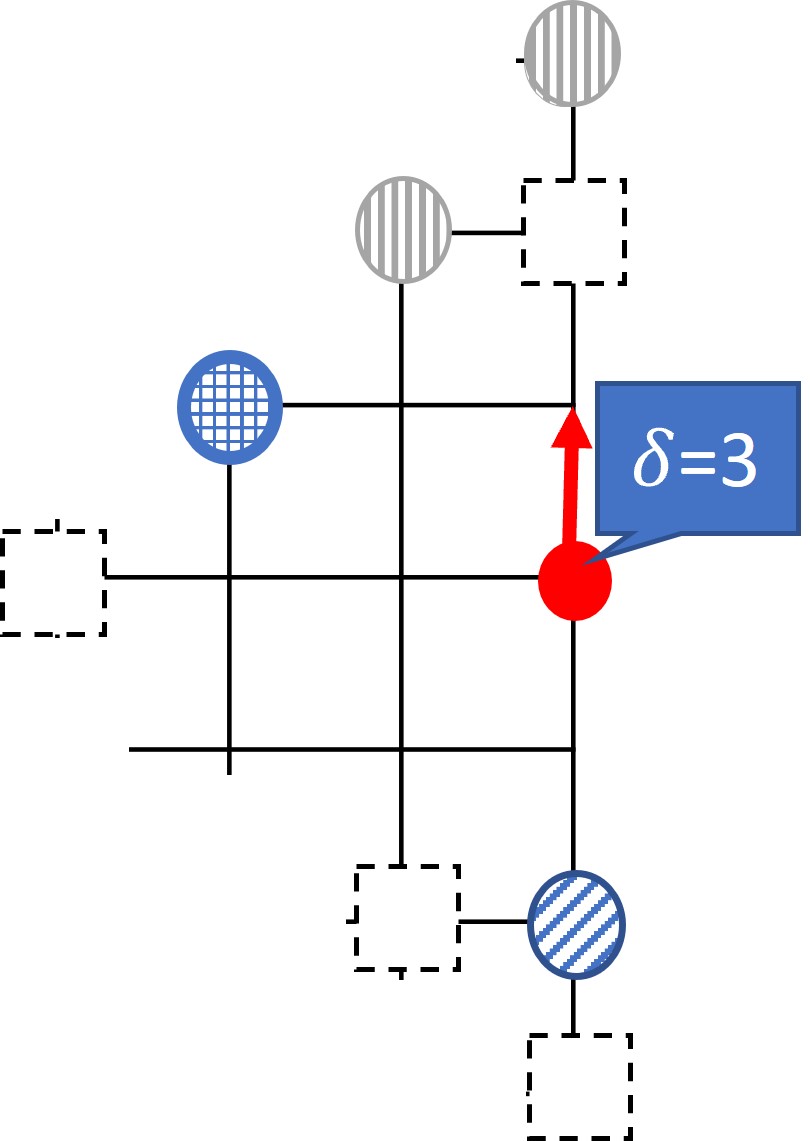}}\hspace{0.3cm}
    \subfigure[{\sf MovP1I}]{\includegraphics[height=3cm]{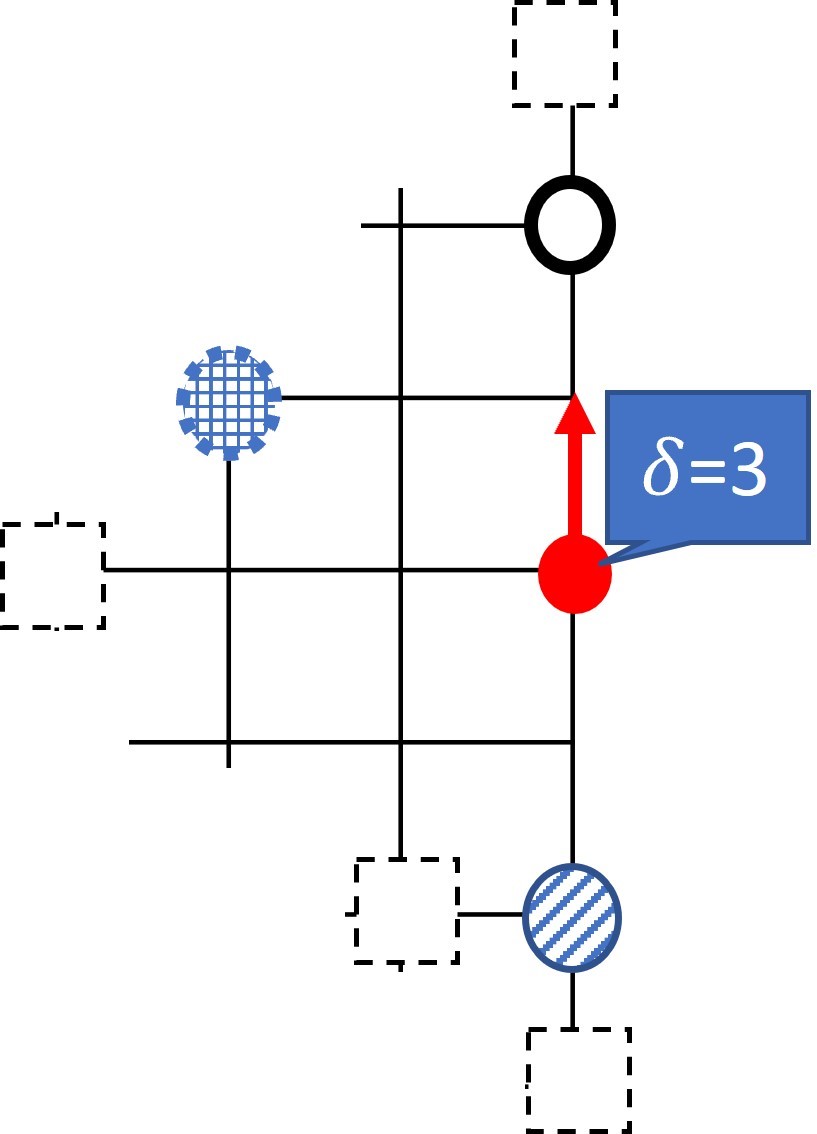}}
    \caption{Definition of views in ${\it MovP1}^\prime$ while $p_1$.}
    \label{fig:p1M3}
\end{figure}
In set ${\it StopP1}^\prime$ in Figure~\ref{fig:p1F3}, there are three types of views:
\begin{itemize}
    \item {\sf StopP1A}: $r_i$ is on the Door corner, and every node two hops away is occupied by a Finished robot.
    \item {\sf StopP1B}: $r_i$ is on a border, its successor is two hops away, and the node two hops away from $r_i$ on the same line as $r_i$ stays is occupied by a Finished robot.
    \item {\sf StopP1C}: $r_i$ is on a corner, its successor is two hops away on a border, and the node two hops away from $r_i$ on the other border is occupied by a Finished robot.
\end{itemize}
\begin{figure}[t]
\centering
\subfigure[{\sf StopP1A}]{\includegraphics[height=1.9cm]{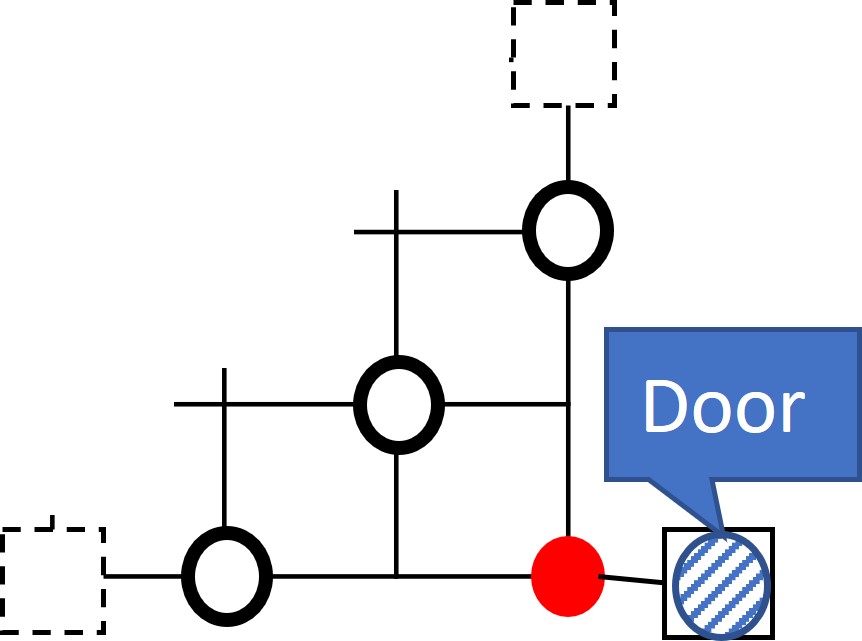}}\hspace{0.4cm}
\subfigure[{\sf StopP1B}]{\includegraphics[height=3cm]{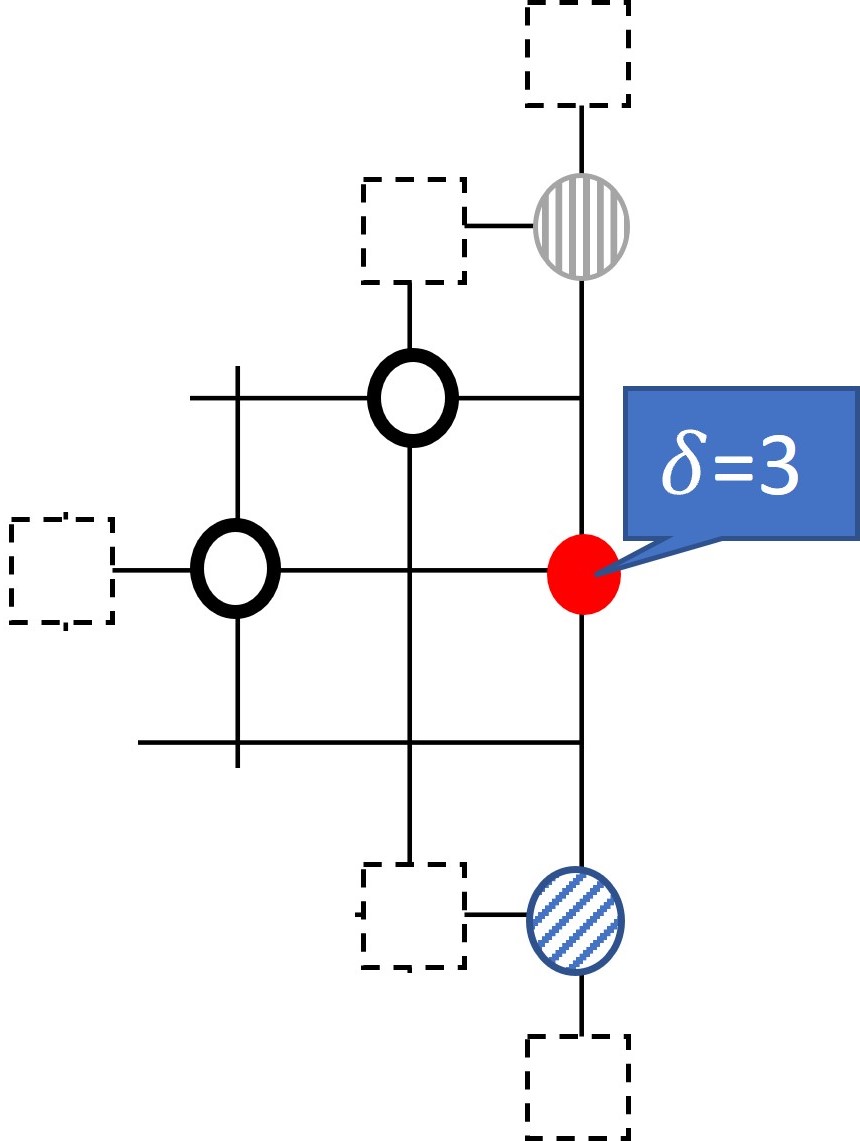}}\hspace{0.4cm}
\subfigure[{\sf StopP1C}]{\includegraphics[height=1.8cm]{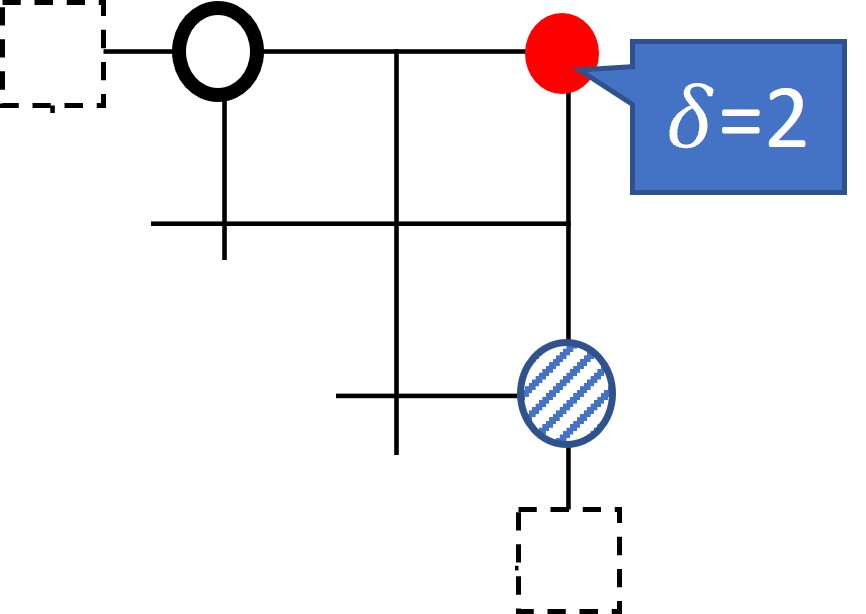}}
    \caption{Definition of views in ${\it StopP1}^\prime$ while $p_1$.}
\label{fig:p1F3}
\end{figure}
In set ${\it Change}^\prime$ in Figure~\ref{fig:p1C3}, there are five types of views:
\begin{itemize}
    \item {\sf ChangeA}: $r_i$ is on a corner, and its successor is two hops away. If $r_i$ can see its non-Finished predecessor, it is three hops away. 
    \item {\sf ChangeB}: $r_i$ is on the Door corner, can see a non-Finished robot $r_j$ on the Door (with $k_j=(k_i+1)\mod 3$), a Finished robot on a border is adjacent to $r_i$. 
    The other border is empty, or $r_i$'s non-Finished predecessor is three hops away on the border.
    \item {\sf ChangeC}: $r_i$ is on the Door corner, can see a non-Finished robot $r_j$ on the Door (with $k_j=(k_i+1)\mod 3$), two nodes two hops away are occupied by Finished robots. 
    A border is empty or $r_i$'s non-Finished predecessor is three hops away on the border. 
    \item {\sf ChangeD}: $r_i$ is on a border, can see its successor $r_j$ two hops away, and there is an adjacent Finished robot on the border on the opposite side of $r_j$.
    If $r_i$ can see its non-Finished predecessor, it is three hops away on the same line as $r_i$ stays.
    \item {\sf ChangeE}: $r_i$ is on a border, can see its successor $r_j$ two hops away, and a Finished robot two hops away on the adjacent line to $r_i$ stays on the opposite side of $r_j$.
    If $r_i$ can see its non-Finished predecessor, it is three hops away on the same line as $r_i$ stays.
\end{itemize}
\begin{figure}[t]
\centering
\subfigure[{\sf ChangeA}]{\includegraphics[height=1.8cm]{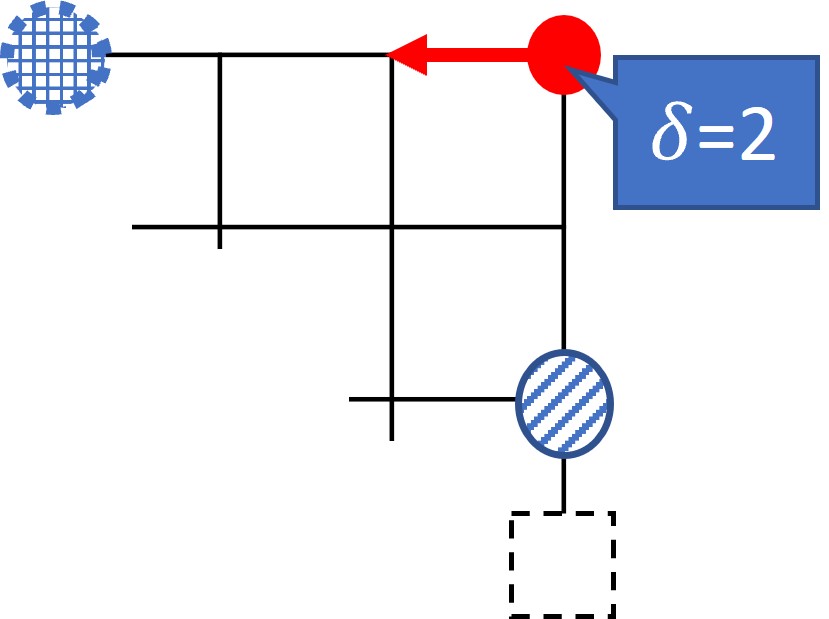}}\hspace{0.4cm}
\subfigure[{\sf ChangeB}]{\includegraphics[height=2.1cm]{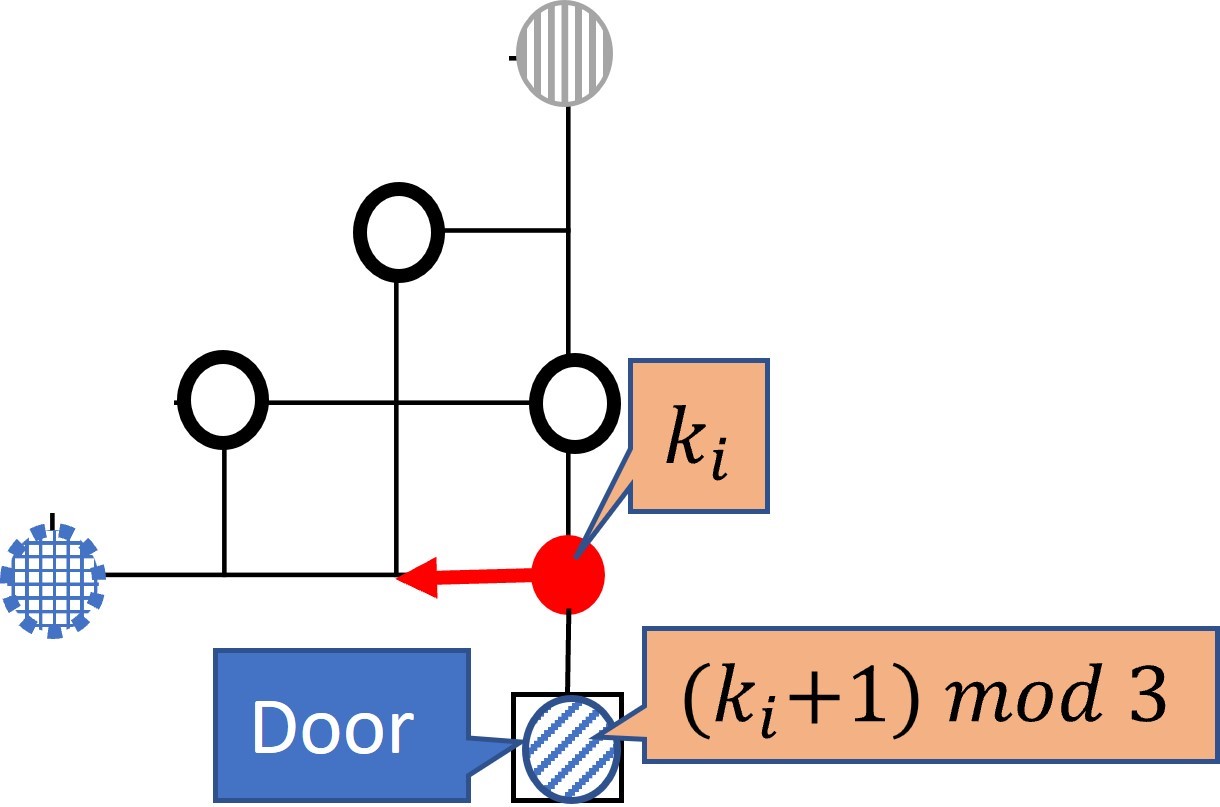}}\hspace{0.4cm}
\subfigure[{\sf ChangeC}]{\includegraphics[height=2.1cm]{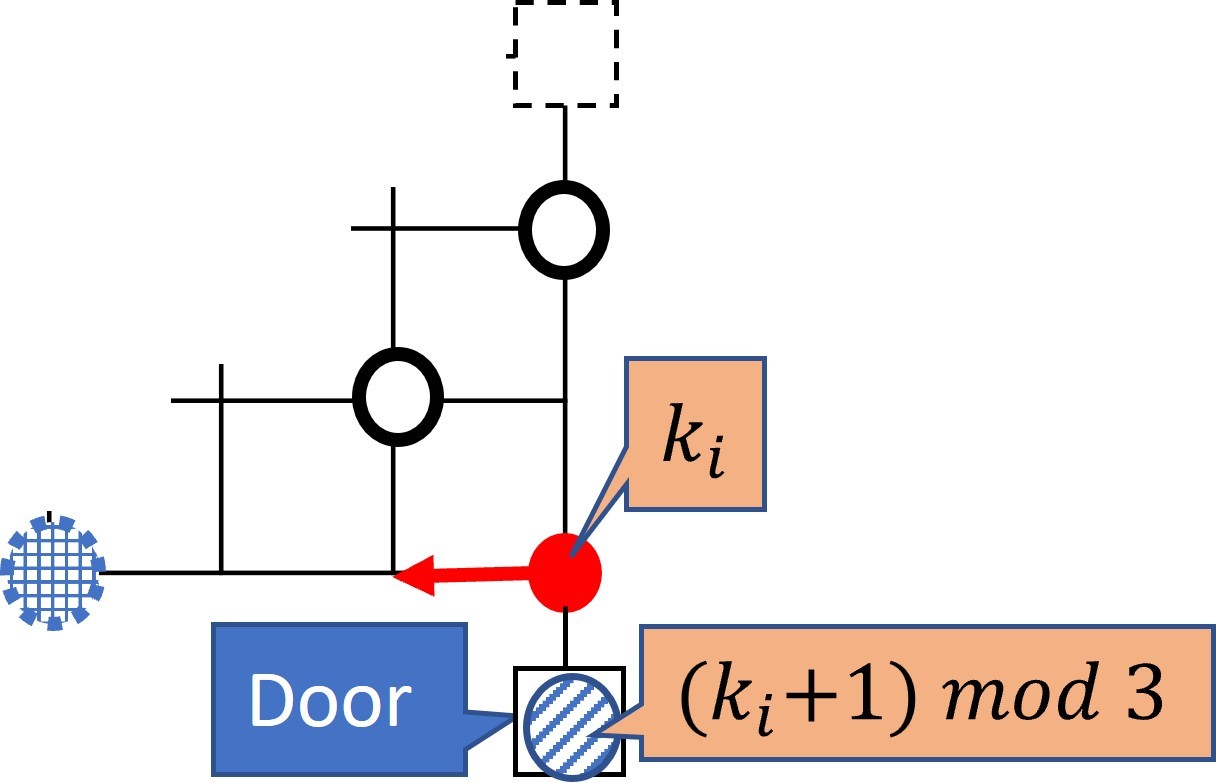}}\\
\subfigure[{\sf ChangeD}]{\includegraphics[height=3cm]{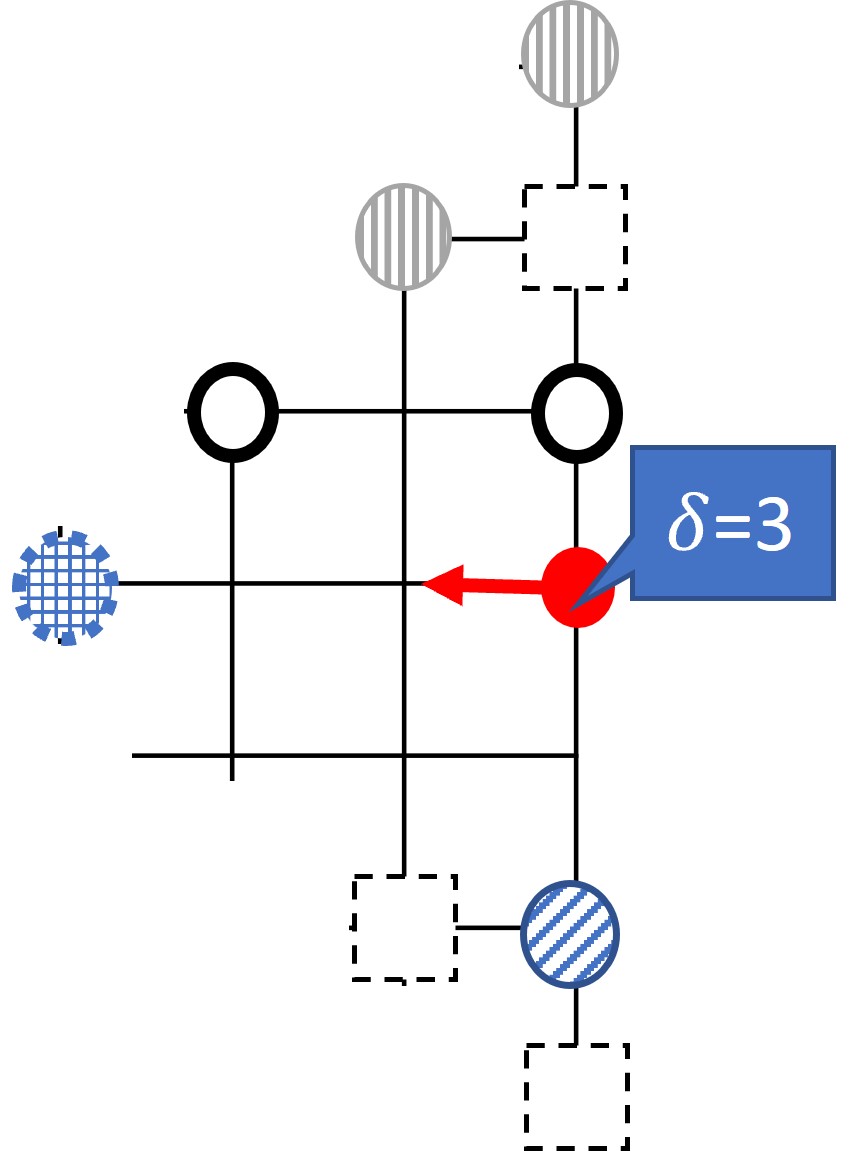}}\hspace{1cm}
\subfigure[{\sf ChangeE}]{\includegraphics[height=3cm]{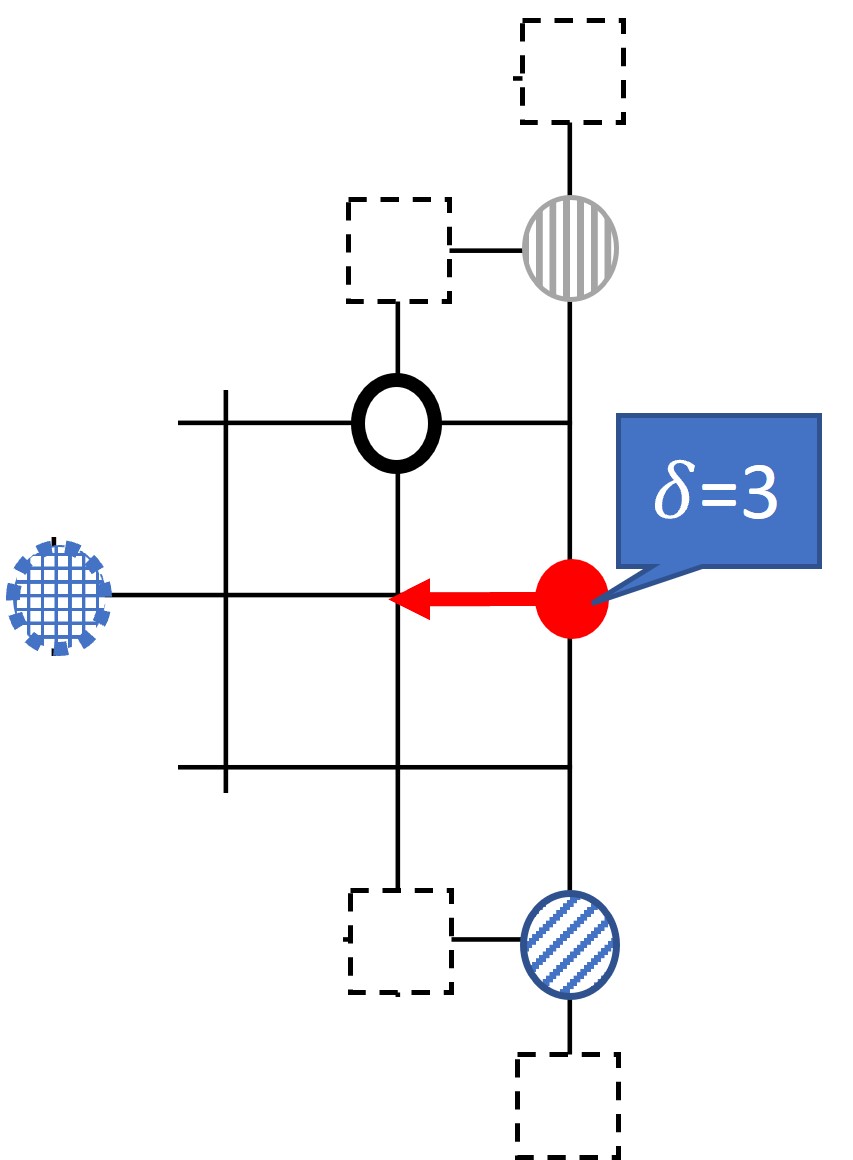}}
\caption{Definition of views in ${\it Change}^\prime$ while $p_1$.}
    \label{fig:p1C3}
\end{figure}
While $c(r_i)=(p_2,k_i)$, we use the view types in Figures~\ref{fig:p2M3}-\ref{fig:p2F3}.
In set ${\it MovP2}^\prime$ in Figure~\ref{fig:p2M3}, there are three types of views:
\begin{itemize}
    \item {\sf MovP2A}: $r_i$ is on a border, and can see its successor two hops away. If $r_i$ can see its predecessor, it is three hops away. 
    \item {\sf MovP2B} and {\sf MovP2C}: $r_i$ is on a line, can see its successor two hops away, and Finished robots on the adjacent line on one side. If $r_i$ can see its non-Finished predecessor on the same line, it is three hops away. The difference between these view types is that, only in {\sf MovP2C}, a Finished robot is adjacent to $r_i$.   
\end{itemize}
\begin{figure}[t]\centering
    \subfigure[{\sf MovP2A}]{\includegraphics[height=3cm]{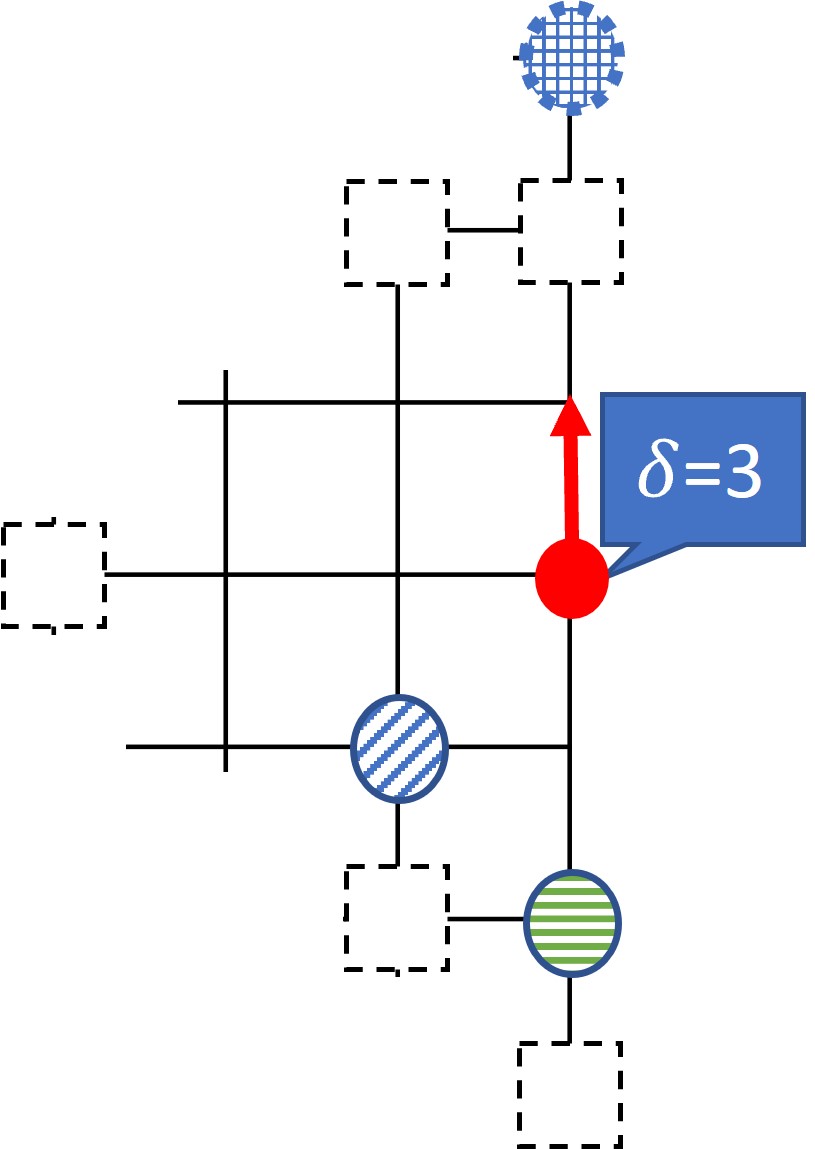}}\hspace{0.4cm}
    \subfigure[{\sf MovP2B}]{\includegraphics[height=3cm]{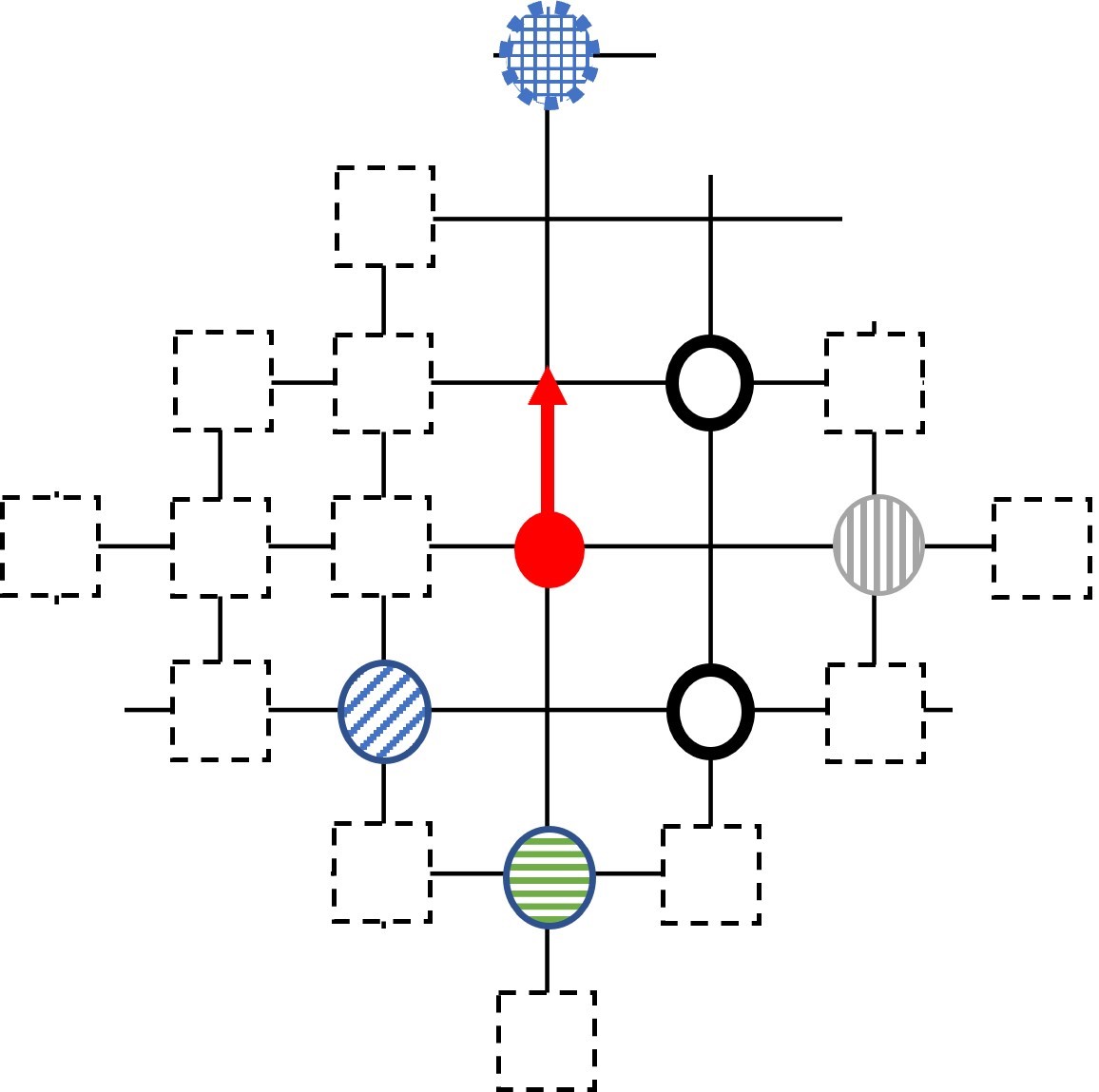}}\hspace{1cm}
    \subfigure[{\sf MovP2C}]{\includegraphics[height=3cm]{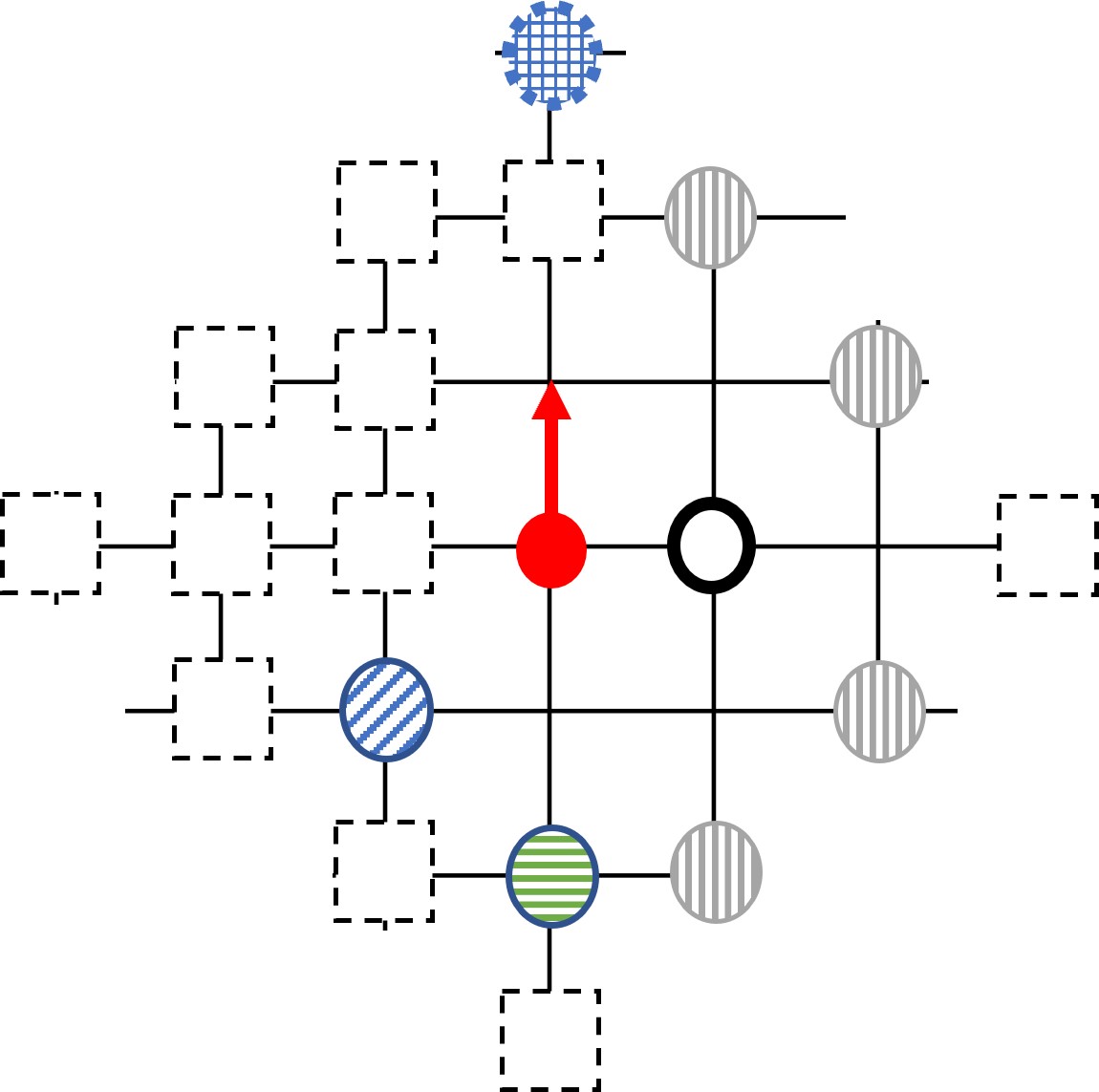}}
    \caption{Definition of views in ${\it MovP2}^\prime$ while $p_2$.}
    \label{fig:p2M3}
\end{figure}
In set ${\it StopP2}^\prime$ in Figure~\ref{fig:p2F3}, there are five types of views:
\begin{itemize}
    \item {\sf StopP2A}: $r_i$ is on a corner, and can see only its successor two hops away. 
    \item {\sf StopP2B}: $r_i$ is on a border, can see its successor two hops away, and a Finished robot two hops away on the opposite side from its successor. 
    \item {\sf StopP2C} (resp. {\sf StopP2D}): $r_i$ is on the end of the line (resp. at the node one hop from the end of the line), can see its successor two hops away, and Finished robots two hops away on an adjacent line.
    \item {\sf StopP2E}: $r_i$ is on a line, can see its successor two hops away, and Finished robots two hops away on the same line and an adjacent line.
\end{itemize}

\begin{figure}[t]
    \centering
    \subfigure[{\sf StopP2A}]{\includegraphics[height=1.6cm]{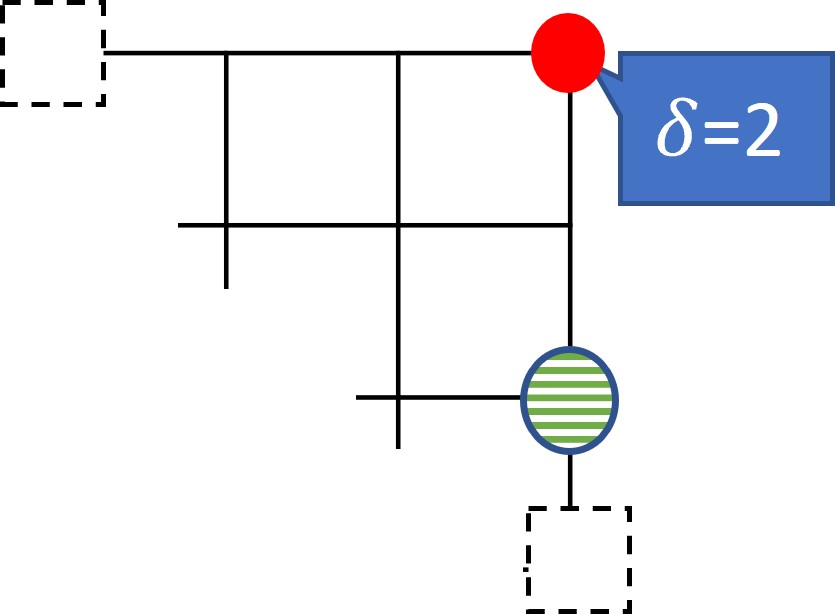}}\hspace{0.5cm}
    \subfigure[{\sf StopP2B}]{\includegraphics[height=3cm]{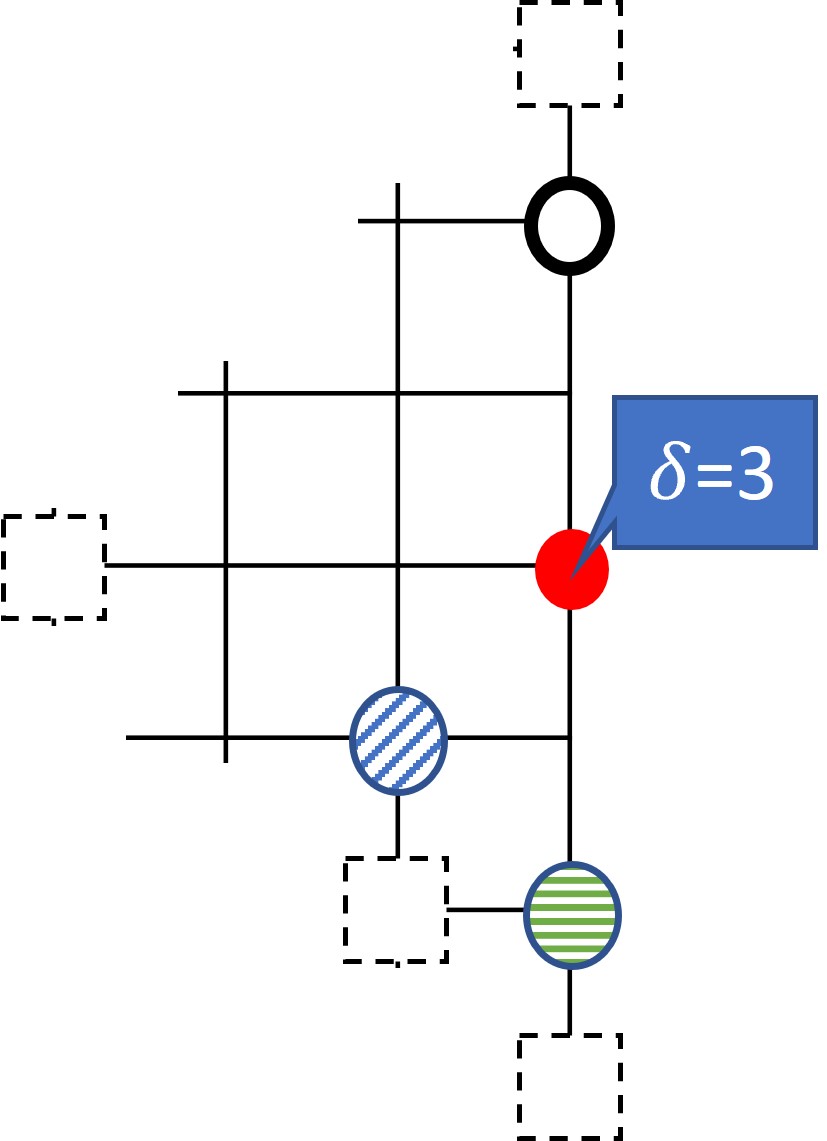}}\hspace{0.5cm}
    \subfigure[{\sf StopP2C}]{\includegraphics[height=2.1cm]{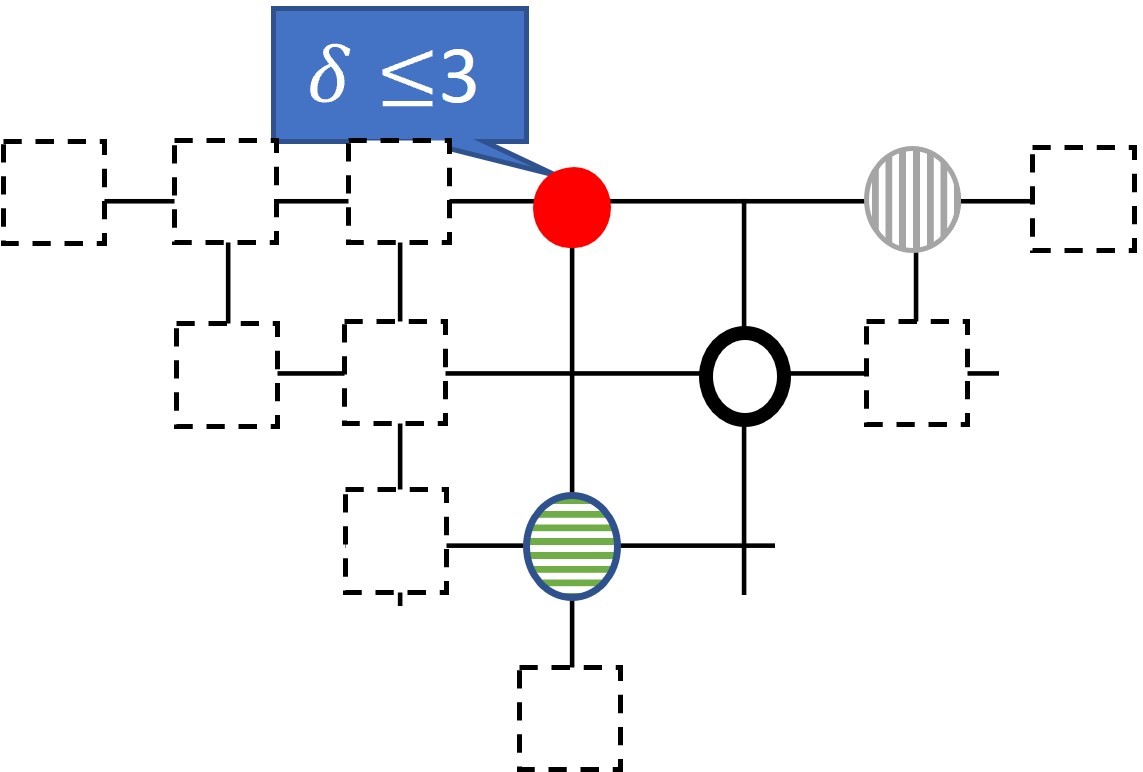}}\\
    \subfigure[{\sf StopP2D}]{\includegraphics[height=2.4cm]{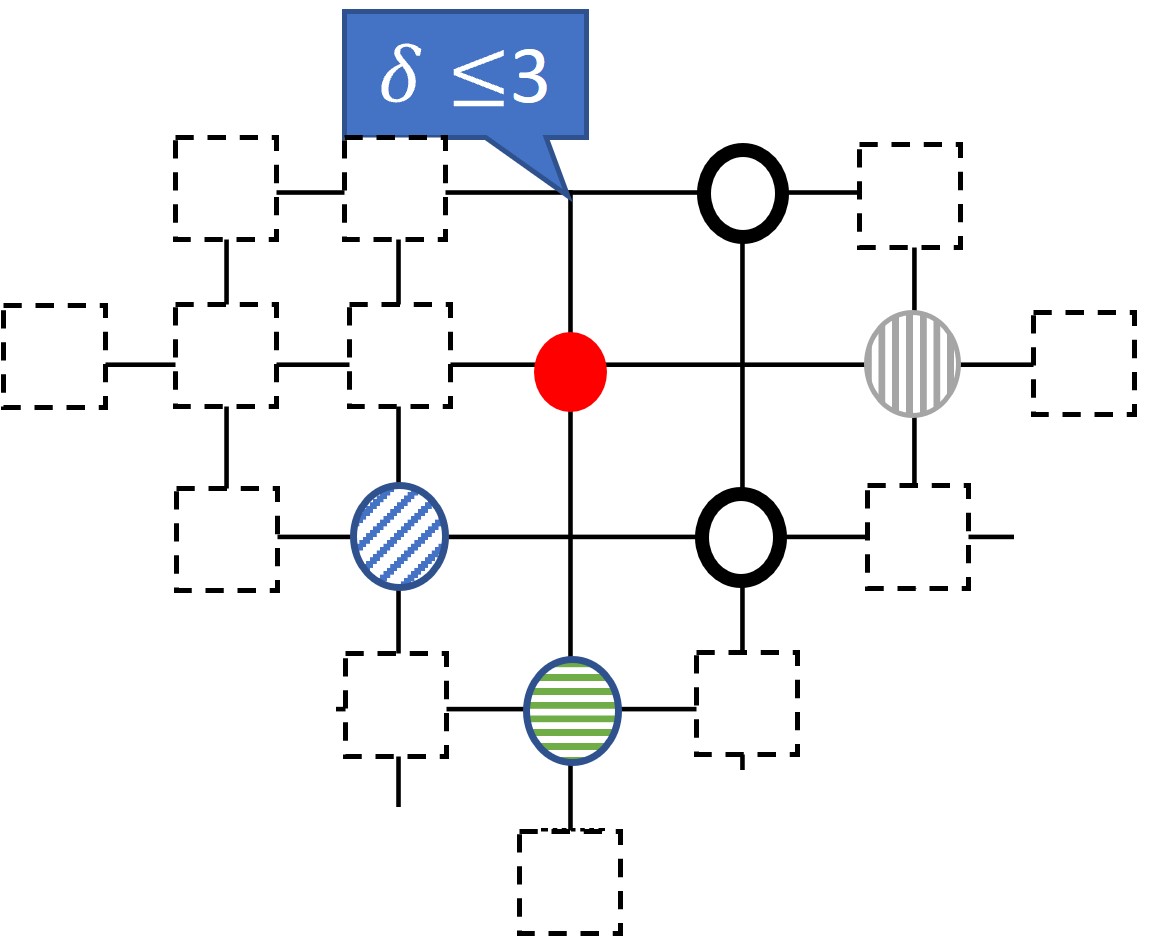}}\hspace{0.4cm}
    \subfigure[{\sf StopP2E}]{\includegraphics[height=3cm]{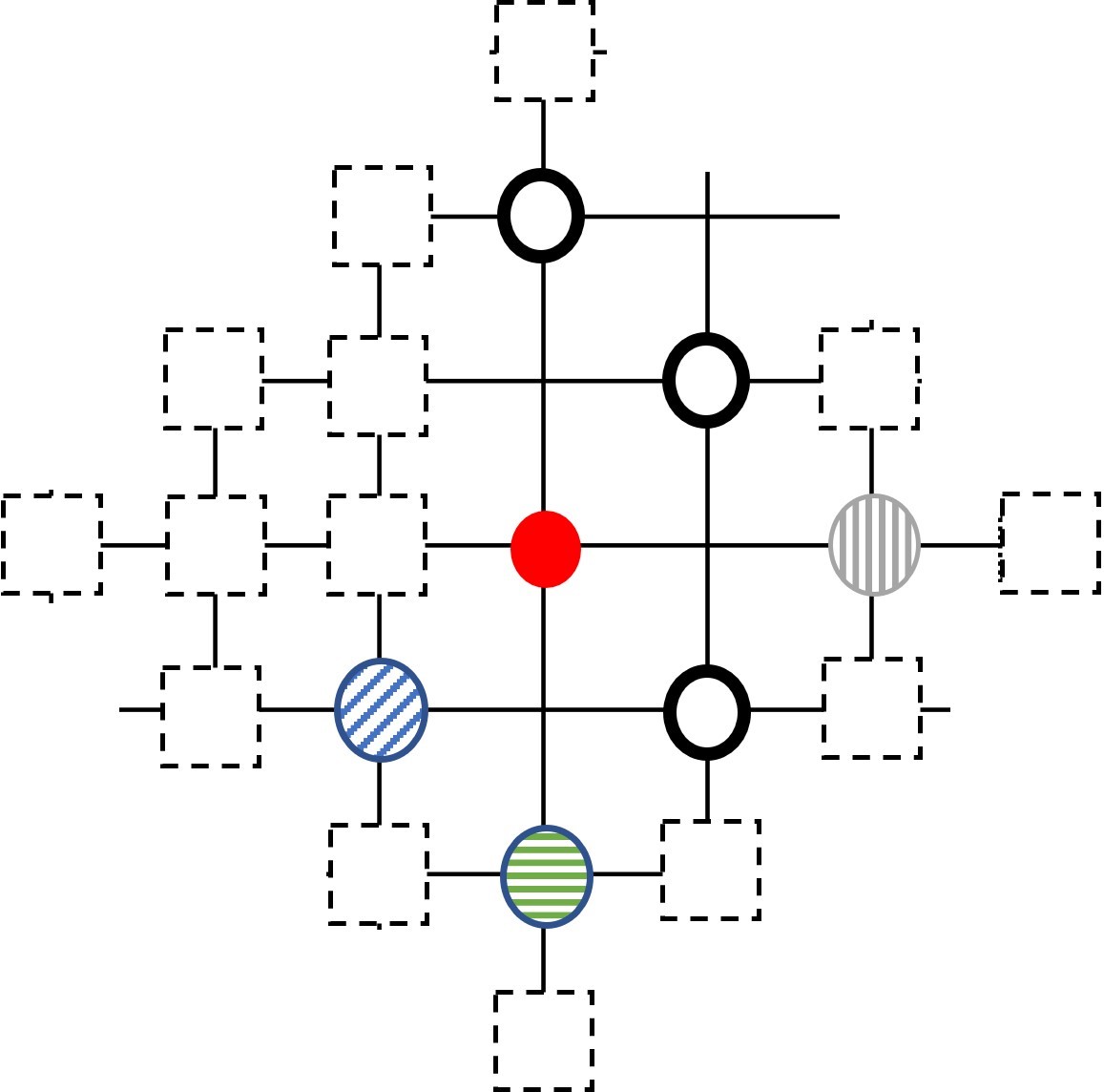}}
    \caption{Definition of views in ${\it StopP2}^\prime$ while $p_2$.}
    \label{fig:p2F3}
\end{figure}

The algorithm description of $\mathcal{MIS\_NOLAB}$ is in Algorithm~\ref{code:alg2}.
The number of each rule represents its priority, a smaller number denoting a higher priority.
In the following, we explain the algorithm in more detail.
Initially, each robot $r_i$ has color $c(r_i)=(p_1,0)$ on the Door node.
Then, if there is a predecessor on the Door corner, $r_i$ sets $k_i$ according to the light value of its predecessor ({\sf Door0} in Figure~\ref{fig:my_label}) by Rule 0.
After that, if ${\it view}(r_i)$ is in ${\it Door}^\prime$ (Figure~\ref{fig:p1D3}), $r_i$ enters the grid from the Door node by Rule 1.
Then, ${\it view}(r_i)$ belongs to ${\it MovP1}^\prime$ (Figure~\ref{fig:p1M3}).
In each type of ${\it MovP1}^\prime$, the direction of ${\it Move}(r_i)$ operation by Rule 4 is as follows:
\begin{itemize}
    \item {\sf MovP1A}: $r_i$ moves to an edge selected arbitrarily (that is, by an adversary).
    \item {\sf MovP1B}, {\sf MovP1C} and {\sf MovP1D}: $r_i$ moves on a border toward a non-Finished robot that is three hops away. 
    \item {\sf MovP1E}: $r_i$ moves toward the Finished robot on a border.
    \item {\sf MovP1F}, {\sf MovP1G}, {\sf MovP1H} and {\sf MovP1I}: $r_i$ moves along the border in the opposite direction from its successor.
\end{itemize}
Thus, $r_i$ moves on the first border by Rule 4, until ${\it view}(r_i)$ belongs to ${\it StopP1}^\prime$ (Figure~\ref{fig:p1F3}) or ${\it Change}^\prime$ (Figure~\ref{fig:p1C3}).
If ${\it view}(r_i)$ is in ${\it StopP1}^\prime$, $r_i$ changes its color to $F$ by Rule 2, and finishes its execution. 
If ${\it view}(r_i)$ belongs to ${\it Change}^\prime$, $r_i$ changes its color to $(p_2,k_i)$, and strays from the first border by Rule 3. 
After $r_i$ strays from the first border, ${\it view}(r_i)$ becomes ${\it MovP2}^\prime$ (Figure~\ref{fig:p2M3}), and $r_i$ moves to the second border or a line by Rule 6, until ${\it view}(r_i)$ belongs to ${\it StopP2}^\prime$ (Figure~\ref{fig:p2F3}).
That is, in each type of ${\it MovP2}^\prime$, $r_i$ moves in a direction away from its successor, while maintaining the same border or line.
If ${\it view}(r_i)$ belongs to ${\it StopP2}^\prime$, $r_i$ changes its color to $F$ by Rule 5, and finishes its execution. 

\begin{algorithm}[t]
\begin{tabbing}
\quad\=\quad\=\qquad\=\qquad\=\qquad\=\qquad\=\qquad\=\qquad\=\qquad\=\qquad\=\kill
{\bf Colors}~ $F$, $(p_1,k_i)$, $(p_2,k_i)$, where $k_i\in\{0, 1, 2\}$\\
{\bf Initialization}~ $c(r_i)=(p_1,0)$\\
{\bf Rules on node $v$ of robot $r_i$}\\
{\bf Rule 0}: {\bf if }${\it view}(r_i)={\sf Door0}$ {\bf then} ${\it SetC}(r_i)$;\\
{\bf if} $c(r_i)=(p_1,k_i)$ {\bf then}$\{$\\
\> {\bf Rule 1}: {\bf if }${\it view}(r_i)\in {\it Door}^\prime$ {\bf then} ${\it Enter\_Grid}(r_i)$;\\
\> {\bf Rule 2}: {\bf if }${\it view}(r_i)\in {\it StopP1}^\prime$ {\bf then} $c(r_i):=F$;\\
\> {\bf Rule 3}: {\bf if }${\it view}(r_i)\in{\it Change}^\prime$ {\bf then} $c(r_i):=(p_2,k_i)$; ${\it Move}(r_i)$;\\
\> {\bf Rule 4}: {\bf if }${\it view}(r_i)\in{\it MovP1}^\prime$ {\bf then} ${\it Move}(r_i)$;\\
$\}$\\
{\bf if} $c(r_i)=(p_2,k_i)$ {\bf then}$\{$\\
\> {\bf Rule 5}: {\bf if }${\it view}(r_i)\in{\it StopP2}^\prime$ {\bf then} $c(r_i):=F$;\\
\> {\bf Rule 6}: {\bf if }${\it view}(r_i)\in {\it MovP2}^\prime$ {\bf then} ${\it Move}(r_i)$;\\
$\}$
\end{tabbing}
\caption{Algorithm $\mathcal{MIS\_NOLAB}$ for a maximum independent set placement with 7 colors lights.}
\label{code:alg2}
\end{algorithm}

\subsubsection{Proof of Correctness}


In the following, we first show that each robot can recognize its successor and robots cannot collide.

\begin{lemma}\label{predecessor}
Each non-Finished robot except the first robot can recognize its predecessor and successor if it keeps two neighboring non-Finished robots.
\end{lemma}
\begin{proof}
By the assumption, each robot $r_i$ initializes its $c(r_i)$ to $(p_1,0)$ on the Door node.
The view of the first robot $r_1$ on the Door node is {\sf DoorA}, thus it moves to the Door corner with $c(r_1)=(p_1,0)$ by Rule 1.
After that, by the assumption, its successor $r_2$ appears on the Door node with $c(r_2)=(p_1,0)$.
Thus, $r_1$ cannot move until its view becomes {\sf MovP1A}, \emph{i.e.}, $r_2$ has $c(r_2)=(p_1,(k_1+1)\mod 3)=(p_1,(0+1)\mod 3)=(p_1,1)$.
Thus, ${\it view}(r_2)$ on the Door node becomes {\sf Door0} and $r_2$ sets $c(r_2)$ to $(p_1,1)$ by Rule 0.
Then, ${\it view}(r_1)$ becomes {\sf MovP1A}, $r_1$ selects the first border arbitrarily and moves on the border by Rule 4.

Consider a robot $r_i$ sets $c(r_i)$ to $(p_1,k_i)$ on the Door node in {\sf Door0}.
Then, its predecessor $r_h$ on the Door corner has $c(k_h)$ where $k_i=(k_h+1)\mod 3$ by the definition of ${\it SetC}(r_i)$.
After $r_i$ enters the grid, on the Door corner, it cannot move until its successor $r_j$ sets $c(r_j)$ to $(p_1,k_j)$ where $k_j=(k_i+1)\mod 3$ by the definitions of {\sf MovP1B}, {\sf MovP1C}, {\sf MovP1D}, {\sf MovP1E}, {\sf ChangeB} or {\sf ChangeC}.
Therefore, each robot $r_i$ on the grid has its order modulo 3 as its value $k_i$, and the value $k_i$ is not changed after that.
Thus, if each robot keeps two neighboring non-Finished robots, it can recognize its neighboring non-Finished robot with smaller (resp. larger) $k$ value than its own as its predecessor (resp. successor), lemma holds.
\end{proof}

\begin{lemma}\label{3collide1}
While the light color is $(p_1,k_i)$, each robot $r_i$ 
can recognize its successor, and cannot collide.
\end{lemma}
\begin{proof}
If there exists an outdated robot $r_o$ that is to move according to its outdated view, the view type is in ${\it Door}^\prime$, ${\it Change}^\prime$, ${\it MovP1}^\prime$, and ${\it MovP2}^\prime$ by the definition of the algorithm.
Thus, if a collision with $r_o$ occurs, then $r_o$'s view type is in ${\it Door}^\prime$, ${\it Change}^\prime$, ${\it MovP1}^\prime$, and ${\it MovP2}^\prime$.
In that case, because a Finished robot does not move at all, it could be that a non-Finished robot in $r_o$'s view moved, 
or that another non-Finished robot came into the visible region of $r_o$.

The first robot $r_1$ keeps its color $c(r_1)=(p_1,0)$.
On the Door node, the view of $r_1$ becomes {\sf DoorA}, and $r_1$ moves to the Door corner.
After that, $r_1$ can move only when its view becomes {\sf MovP1A}, \emph{i.e.}, $c(r_2)$ has to be set to $(p_1,1)$ by ${\it SetC}(r_2)$, where $r_2$ is $r_1$'s successor.
Thus, the view of $r_2$ on the Door node becomes {\sf Door0}, and eventually $r_2$ sets its color to $(p_1,1)$.
Then, $r_1$ selects one border as the first border arbitrarily in {\sf MovP1A} and moves.
Because the distance from $r_2$ becomes two, ${\it view}(r_1)$ becomes {\sf MovP1F} and $r_1$ moves one hop.
Then, because there is no rule to move for $r_1$ when the distance from $r_2$ is three, $r_1$ cannot move.
Thus, ${\it view}(r_2)$ becomes {\sf DoorA} and $r_2$ enters the grid.
After that, $r_1$ can move because the distance from $r_2$ is two, \emph{i.e.}, {\sf MovP1F} by Rule 4.
Thus, by Rule 4, $r_2$ can move from the Door corner only when ${\it view}(r_2)$ becomes {\sf MovP1B}.
That is, when $r_2$ can move, the distance between $r_1$ and $r_2$ is three, and $r_2$'s successor $r_3$ has $c(r_3)=(p_1,2)$ by the definition of {\sf MovP1B}.
By the definition of ${\it MovP1}^\prime$, they move only on the first border according to the degree of nodes until they arrive at the end of the first border.
While they move on the first border, only $r_1$ and $r_3$ are neighbors for $r_2$, and $r_3$'s successors also follow $r_3$ in the same way as $r_2$.
When the view of robots become {\sf ChangeA} (\emph{i.e.}, they arrive at the end of the first border), they change their colors from $(p_1,k_i)$ to $(p_2,k_i)$ by Rule 3 in the same order as they entered the grid.

By the same argument, each robot $r_i$ moves only on the first border using the degree of nodes while $c(r_i)=(p_1,k_i)$ holds, by the definition of the views in ${\it MovP1}^\prime$.
Then, on the first border, if $r_i$ is not on the Door node or the Door corner, $r_i$ can move only when the distance from its predecessor $r_h$ is three and from its successor $r_j$ is two. 
That is, while $r_i$ moves on the first border, $r_j$ follows $r_i$.
Then, by the definition of the views in ${\it MovP1}^\prime$, while $r_i$ moves on the first border, there are at most two non-Finished neighboring robots $r_h$ and $r_j$ for $r_i$ and they are kept by $r_i$'s movement, i.e., robots move on the first border keeping in the order they entered the grid.
By the definition of the algorithm, only when $r_h$ becomes Finished two hops away by ${\it StopP1}^\prime$ or ${\it StopP2}^\prime$,
the number of non-Finished neighboring robots for $r_i$ becomes one, but $r_i$ keeps $r_j$ with $k_i<k_j$ in its view and recognizes $r_j$ as its successor.
By Lemma~\ref{predecessor}, on the first border, each non-Finished robot can always recognize its successor, that is, each robot can recognize its direction to move.
Therefore, the first border is one-way.
Thus, while $c(r_i)=(p_1,k_i)$ holds, if the view is in ${\it Door}^\prime$, ${\it Change}^\prime$, or ${\it MovP1}^\prime$, $r_i$ cannot become outdated as any non-Finished robot cannot come into $r_i$'s visible region, and any non-Finished robots in $r_i$'s view cannot move.
That is, each robot cannot collide with other robots.

For each robot $r_i$, when its view becomes in ${\it StopP1}^\prime$ on the first border, $r_i$ changes its color to $F$ by Rule 2.
When its view belongs to ${\it Change}^\prime$ on the first border, $r_i$ changes its color from $(p_1,k_i)$ to $(p_2,k_i)$ and changes its direction to a line by Rule 3.
Thus lemma holds.
\end{proof}

\begin{lemma}\label{3collide2}
While the light color is $(p_2,k_i)$, each robot $r_i$ can recognize its successor, and cannot collide.
\end{lemma}
\begin{proof}
Consider the time $t$ when each robot $r_i$ changes its color to $(p_2,k_i)$ on the first border.
Then, its view is in ${\it Change}^\prime$ by Rule 3, and $r_i$ moves to a line. 
By the proof of Lemma~\ref{3collide1} and the definition of the views in ${\it Change}^\prime$, $r_i$'s successor $r_j$ is two hops behind at $t$ if $r_i$ is not on the Door corner. 
After that, by the definition of views in ${\it MovP2}^\prime$, $r_i$ can move by Rule 6 only when the distance from $r_j$ is two and the distance from its non-Finished predecessor $r_h$ (if exists) is three.
Thus, after $r_i$ moves by the view in ${\it MovP2}^\prime$, $r_i$ cannot move unless $r_j$ moves. 
\begin{itemize}
\item If ${\it view}(r_i)$ is {\sf ChangeA} at $t$, $r_i$ moves to the $0$-line (\emph{i.e.}, the second border) and $r_j$ also follows $r_i$. 
After that, ${\it view}(r_i)$ becomes {\sf MovP2A} until $r_i$ arrives at the diagonal corner (\emph{i.e.}, {\sf StopP2A}) or $r_h$ becomes Finished on $0$-line (\emph{i.e.}, {\sf StopP2B}). 
By the definition of {\sf MovP2A}, $r_i$ moves on the second border according to the degree of nodes.
By the definition of the algorithm, there is no rule to make $r_i$ stray from the second border.
Then, by the definition of {\sf MovP2A}, $r_i$ keeps the distance from $r_j$ two or three hops and has at most two non-Finished neighboring robots, while $r_i$ moves on the second border. 
By this distance, these non-Finished neighboring robots are kept by the movement. 
Because robots on the second border keep the same order as when they entered the grid, only when $r_h$ becomes Finished by {\sf StopP2B} or {\sf StopP2A}, or $r_i$ is the first robot, the number of non-Finished neighboring robots for $r_i$ becomes one. 
Then, $r_i$ can recognize $r_j$ as its successor, because $r_j$ is always in ${\it view}(r_i)$ and $k_i<k_j$ holds.
That is, by Lemma~\ref{predecessor}, $r_i$ can always recognize $r_j$ as its successor, 
and the second border is one way. 

\item If ${\it view}(r_i)$ is {\sf ChangeD} (resp. {\sf ChangeB}) at $t$, $r_i$ moves on a line except $0$-line and $(L-1)$-line (resp. on $(L-1)$-line) and $r_j$ also follows $r_i$. 
Without loss of generality, let the line be $m$-line where $m>0$.  
Then, by the definition of {\sf ChangeD} (resp. {\sf ChangeB}), robots on $(m-1)$-line are Finished and $(m+1)$-line is empty (if it exists on the grid). 
Thus, after that, because $r_j$ follows $r_i$, ${\it view}(r_i)$ becomes {\sf MovP2B} or {\sf MovP2C} until $r_i$ arrives at the end of the line (\emph{i.e.}, {\sf StopP2C} or {\sf StopP2D}) or $r_h$ becomes Finished on $m$-line (\emph{i.e.}, {\sf StopP2E}).
By the definition of the algorithm, there is no rule to make $r_i$ stray from $m$-line.
By the definitions of {\sf MovP2B} and {\sf MovP2C}, $r_i$ keeps the distance from $r_j$ two or three hops, and has at most two non-Finished neighboring robots, while $r_i$ moves on $m$-line.
By this distance, these non-Finished neighboring robots are kept by the movement. 
Because robots on $m$-line keep the same order as when they entered the grid, only when $r_h$ becomes Finished on $m$-line by Rule 5 (\emph{i.e.}, {\sf StopP2C}, {\sf StopP2D} or {\sf StopP2E}), or $r_i$ is the first robot for $m$-line (\emph{i.e.}, $r_h$ is Finished on the intersection of the first border and $(m-1)$-line in {\sf ChangeD} (resp. {\sf ChangeB})), the number of non-Finished neighboring robots for $r_i$ becomes one. Then, $r_i$ also recognizes $r_j$ as its successor, because $r_j$ is always in ${\it view}(r_i)$ and $k_i<k_j$ holds.
Thus, by Lemma~\ref{predecessor}, $r_i$ can always recognize $r_j$ as its successor, and $m$-line is one way.

\item If ${\it view}(r_i)$ is {\sf ChangeE} (resp. {\sf ChangeC}) at $t$, $r_i$ moves on a line except $0$-line and $(L-1)$-line (resp. on $(L-1)$-line) and $r_j$ also follows $r_i$. 
After that, ${\it view}(r_i)$ becomes {\sf MovP2B} or {\sf MovP2C} until $r_i$ arrives at the end of the line (\emph{i.e.}, {\sf StopP2C} or {\sf StopP2D}) or $r_h$ becomes Finished on the same line (\emph{i.e.}, {\sf StopP2E}).
By the same discussion as above, $r_i$ can always recognize $r_j$ as its successor, and the line is one way.
\end{itemize}
Therefore, in any case, while $r_i$ has $(p_2,k_i)$, 
if the view is in ${\it Change}^\prime$ or ${\it MovP2}^\prime$, then $r_i$ cannot become outdated as any non-Finished robots cannot come into $r_i$'s visible region, and any non-Finished robots in $r_i$'s view cannot move.
Thus, $r_i$ cannot collide with other robots while $r_i$ has $(p_2,k_i)$, and the lemma holds. 
\end{proof}

\begin{figure*}[t]
    \centering
    \subfigure[For $0$-line]{\includegraphics[width=0.45\textwidth]{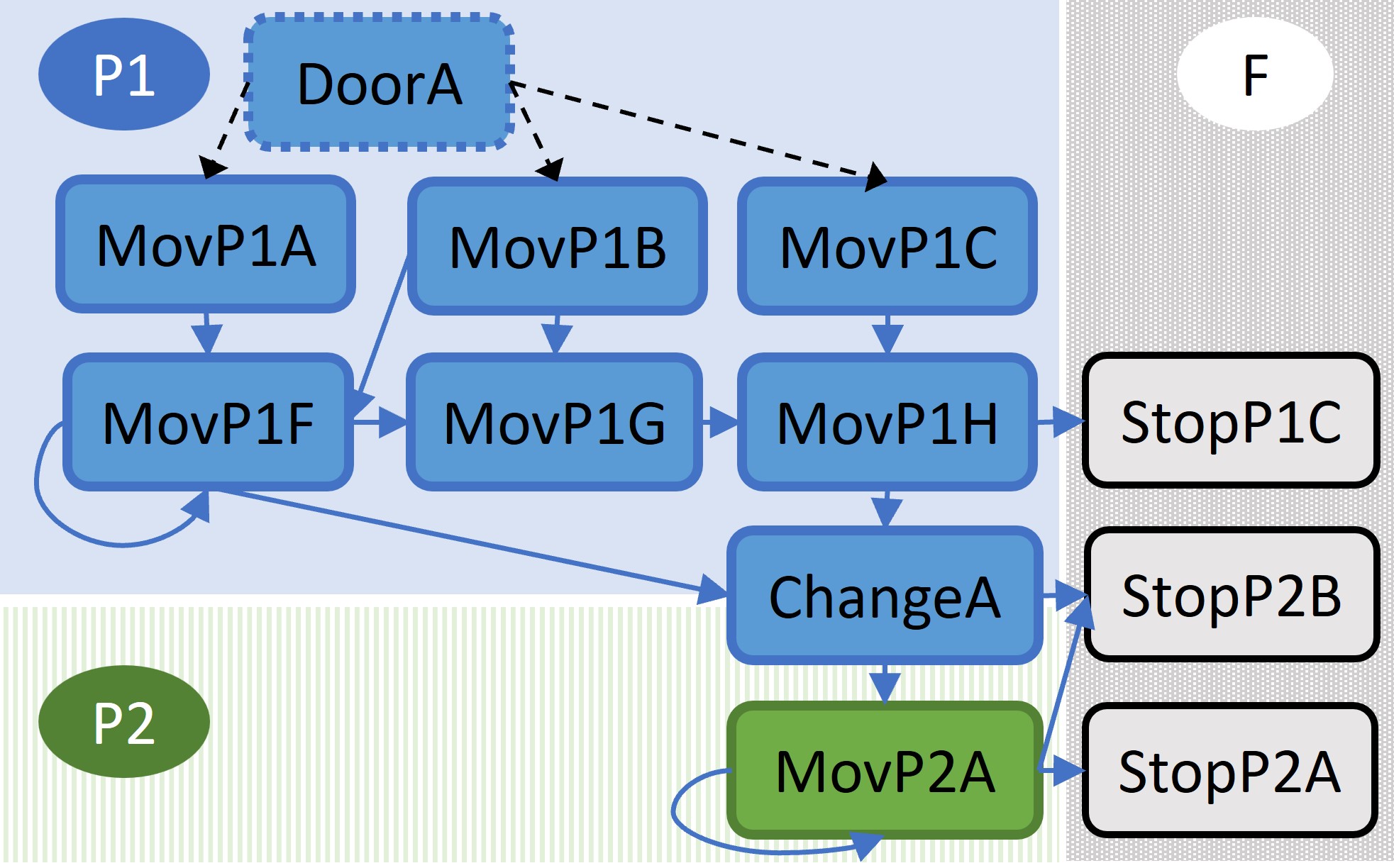}}\hspace{1cm}
    \subfigure[For $(L-2)$-line (any $L$) and $(L-3)$-line ($L>3$)]{\includegraphics[width=0.45\textwidth]{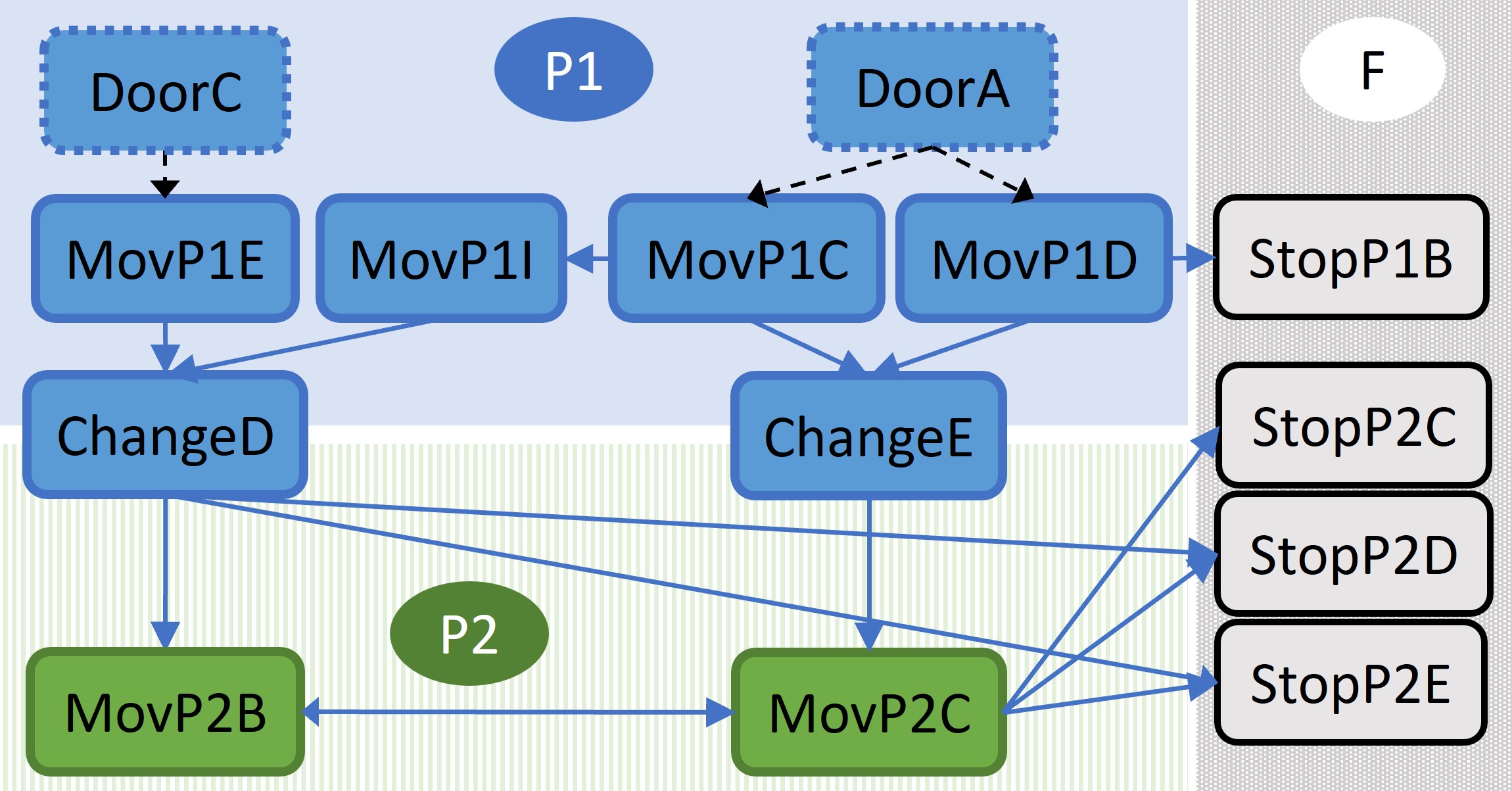}}\\
    \subfigure[For $(L-1)$-line (any $L$)]{\includegraphics[width=0.45\textwidth]{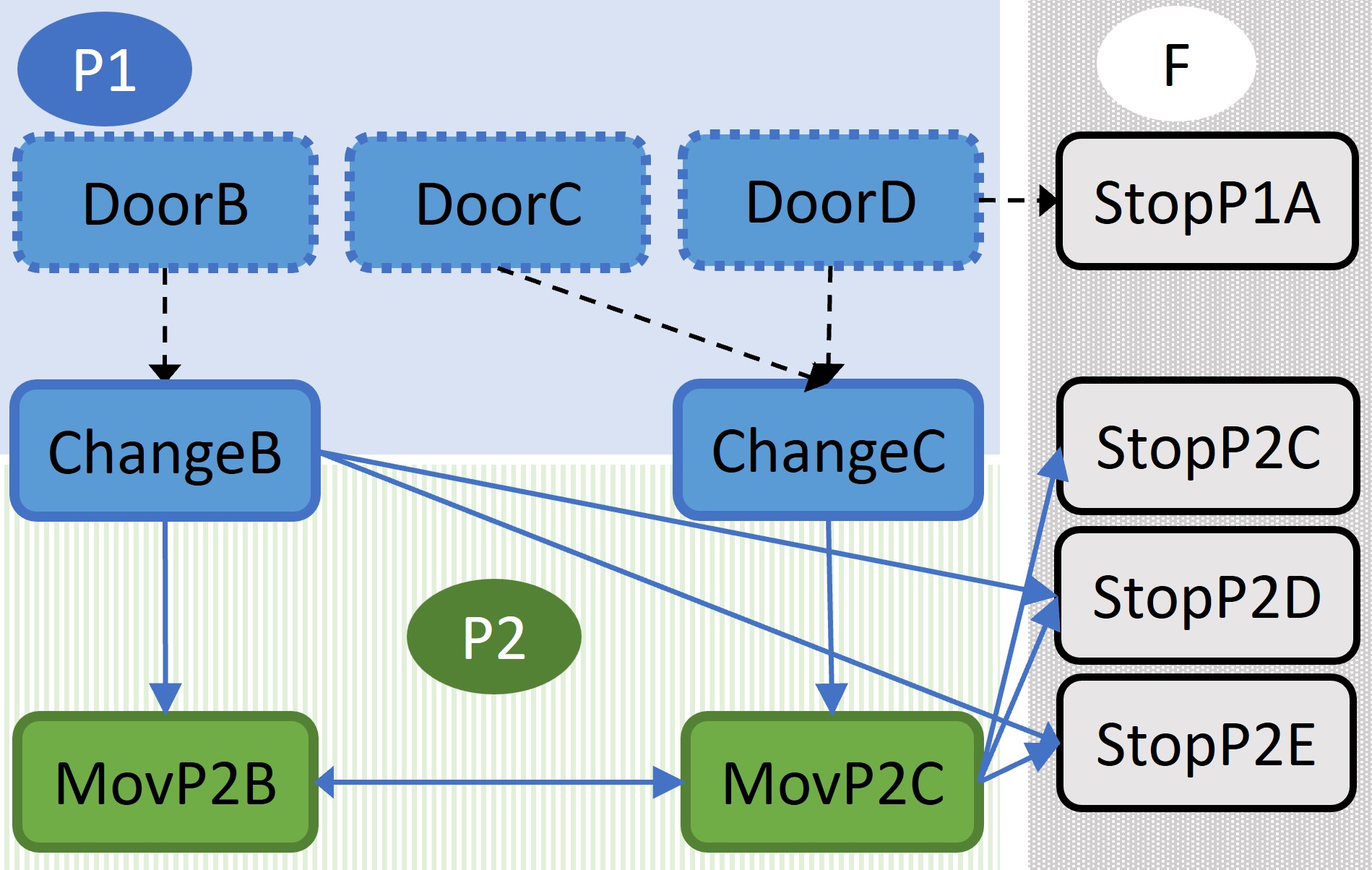}}\hspace{1cm}
    \subfigure[For other lines ($L>4$)]{\includegraphics[width=0.45\textwidth]{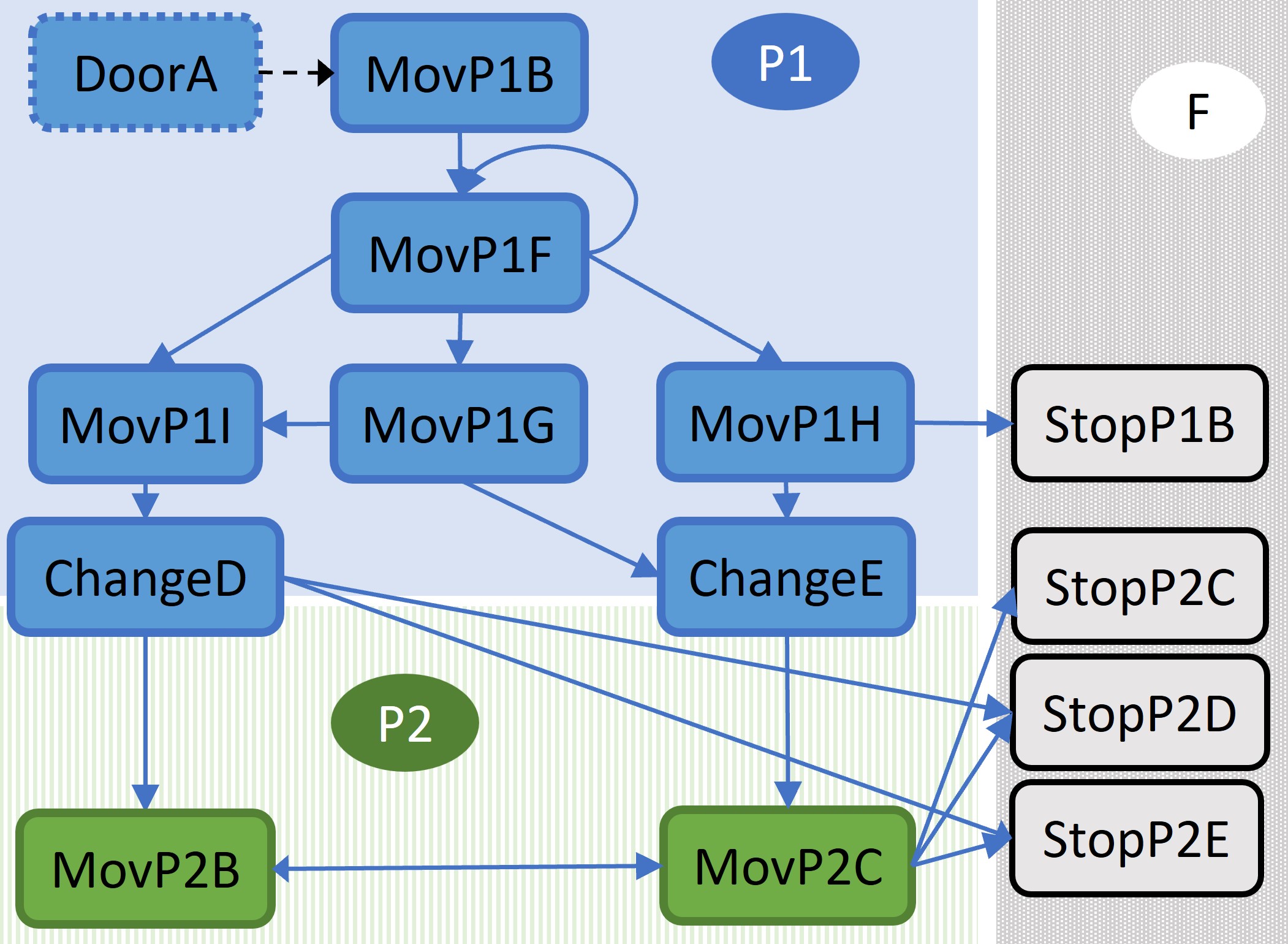}}
    \caption{View type transitions of $\mathcal{MIS\_NOLAB}$. Each solid arrow (resp. dotted arrow) represents a transition by a {\it Move} (resp. {\it Enter\_Grid}) operation. We omitted the view {\sf Door0}.}
    \label{fig:Alg2}
\end{figure*}

\begin{lemma}\label{predecessor2}
Each non-Finished robot can recognize its successor. 
\end{lemma}
\begin{proof}
By Lemma~\ref{3collide1} (resp. Lemma~\ref{3collide2}), each non-Finished robot $r_i$ can always recognize its successor while $c(r_i)=(p_1,k_i)$ (resp. $c(r_i)=(p_2,k_i)$) holds. 
Thus, the lemma holds.
\end{proof}

\begin{lemma}\label{3collide}
Robots cannot collide when executing $\mathcal{MIS\_NOLAB}$. 
\end{lemma}
\begin{proof}
By Lemma~\ref{3collide1} (resp. Lemma~\ref{3collide2}), while the light color is $(p_1,k_i)$ (resp. $(p_2,k_i)$), robots cannot collide.
Because each robot cannot move after it becomes Finished, the lemma holds.
\end{proof}

Next, we show that $\mathcal{MIS\_NOLAB}$ constructs a maximum independent set.

\begin{lemma}\label{31st}
The first robot $r_1$ moves to the diagonal corner, and $c(r_1)$ becomes $F$ on the corner. 
\end{lemma}
\begin{proof}
By the proofs of Lemmas~\ref{3collide1} and \ref{3collide2}, while robots move on the grid, they keep the order they entered the grid.

By the proof of Lemma~\ref{3collide1}, $r_1$ eventually arrives at the end of the first border, and then
$r_1$'s successor $r_2$ is on the node three hops behind. 
When the distance between $r_1$ and $r_2$ becomes two, then ${\it view}(r_1)$ becomes {\sf ChangeA} and $r_1$ changes its color from $(p_1,k_1)$ to $(p_2,k_1)$.

After that, by the proof of Lemma~\ref{3collide2}, $r_1$ eventually arrives at the diagonal corner because $r_1$ is the first robot.
When the distance between $r_1$ and $r_2$ becomes two, ${\it view}(r_1)$ becomes {\sf StopP2A}.  
Because $c(r_1)=(p_2,k_1)$, it changes its color to $F$ on the corner by Rule 5.

Thus, the lemma holds.
\end{proof}

\begin{lemma}\label{31stB}
The first $\lceil l/2 \rceil$ robots move to the second border, and their colors become $F$.
Additionally, nodes on the second border are occupied by a robot or empty alternately from the diagonal corner. 
\end{lemma}
\begin{proof}
By Lemma~\ref{31st}, the first robot $r_1$ eventually becomes Finished on the diagonal corner.
Then, by the definition of {\sf StopP2A} for $r_1$, the distance between $r_1$ and its successor $r_2$ is two.

Consider the execution of $r_2$ after $c(r_1)$ becomes $F$.
If $l$ is more than three, ${\it view}(r_2)$ becomes {\sf StopP2B} when the distance between $r_2$ and its successor $r_3$ becomes two.
Then, by Rule 5, $c(r_2)$ becomes $F$.
If $l$ is three, then $r_2$ is at the end of the first border, thus ${\it view}(r_2)$ becomes {\sf StopP1C} when the distance between $r_2$ and $r_3$ becomes two.
Then, $c(r_2)$ becomes $F$ by Rule 2.
Note that, in both cases, the distance between $r_1$ and $r_2$ remains two hops.

For the successors of $r_2$, we can discuss their movements in the same way as $r_2$.
By the definitions of {\sf StopP1C} and {\sf StopP2B}, when robots become Finished on the second border, the distance between a robot and its successor is two hops because there is no rule to move to the adjacent node of the occupied node on the border.
Therefore, on the second border, beginning with the diagonal corner, every even node is occupied, and the number of robots is $\lceil l/2 \rceil$.
If $l$ is odd, when the $\lceil l/2 \rceil$-th robot $r_i$ arrives at the end of the first border and the distance between $r_i$ and its successor becomes two, $r_i$'s view becomes {\sf StopP1C} and $r_i$ changes its color to $F$ by Rule 2.
Otherwise, $r_i$ changes its color to $(p_2,k_i)$ and moves to the second border by Rule 3.

Thus, the lemma holds.
\end{proof}

\begin{lemma}\label{31-line}
From the $(\lceil l/2 \rceil+1)$-th to the $l$-th robots, each robot moves to the $1$-line, and its color becomes $F$.
Additionally, nodes on the $1$-line are empty or occupied by a robot alternately, beginning with an empty node.
\end{lemma}
\begin{proof}
By Lemma~\ref{31stB}, $\lceil l/2 \rceil$ robots on $0$-line eventually become Finished.
By the definitions of ${\it StopP1}^\prime$ and ${\it StopP2}^\prime$, except on the Door corner, each robot can change its color to $F$ only when the distance from its successor is two.

Let $r_i$ be the $(\lceil l/2 \rceil+1)$-th robot, $r_h$ be the $(\lceil l/2 \rceil)$-th robot (\emph{i.e.}, $r_h$ is the predecessor of $r_i$), and $r_j$ be the $(\lceil l/2 \rceil+2)$-th robot (\emph{i.e.}, $r_j$ is the successor of $r_i$).
$r_i$ and $r_j$ move from the Door node in the same way as $r_h$ while $c(r_h)\neq F$.
Because robots on $0$-line (including $r_h$) become Finished eventually and then the distance between $r_i$ and $r_h$ is two, one of the following three cases occurs: 
\emph{(1)} When $r_i$ is on the Door corner, if $L=3$ and $l$ is odd, ${\it view}(r_i)$ becomes {\sf MovP1E}.
When the distance between $r_i$ and $r_j$ becomes two, \emph{(2)} if $l$ is odd, ${\it view}(r_i)$ becomes {\sf MovP1I}, because the end of the first border is occupied by $r_h$, or \emph{(3)} if $l$ is even, ${\it view}(r_i)$ becomes {\sf ChangeE}, because the end of the first border is empty but its adjacent node on the second border is occupied by $r_h$.

In cases \emph{(1)} and \emph{(2)}, 
by Rule 4, $r_i$ moves to the node in front of the end of the first border.
Then, after $r_j$ comes to the node two hops behind by {\sf MovP1F}, ${\it view}(r_i)$ becomes {\sf ChangeD}.
Then, by Rule 3, $c(r_i)$ becomes $(p_2,k_i)$ and $r_i$ moves to $1$-line.
After that, when $r_j$ comes to the node two hops away from $r_i$ by {\sf MovP1G}, if $l=3$, ${\it view}(r_i)$ becomes {\sf StopP2D} and $r_i$ changes its color to $F$ by Rule 5.
Otherwise, because $r_i$ can see Finished robots on $0$-line, ${\it view}(r_i)$ becomes {\sf MovP2B}.
Then, because the nodes on $0$-line are occupied alternately, ${\it view}(r_i)$ becomes {\sf MovP2B} and {\sf MovP2C} alternately by the execution of Rule 6.
Thus, $r_i$ moves toward the other side border that is parallel to the first border by Rule 6 and $r_j$ follows $r_i$.
Finally, ${\it view}(r_i)$ eventually becomes {\sf StopP2D} because the diagonal corner is occupied by a Finished robot (Lemma~\ref{31st}). 
Then, by Rule 5, $c(r_i)$ eventually becomes $F$.
Because $l$ is odd, $\lfloor l/2 \rfloor-1$ successors of $r_i$ follow $r_i$, and eventually their views become {\sf StopP2E}, and they change their colors to $F$ by Rule 5 on $1$-line.

In case \emph{(3)}, $r_i$ also changes its color to $(p_2,k_i)$, and moves to $1$-line by Rule 3.
After that, because $r_i$ can see Finished robots on $0$-line, ${\it view}(r_i)$ becomes {\sf MovP2C}. 
Then, in the same way as for cases \emph{(1)} and \emph{(2)}, $l/2-1$ robots including $r_i$ become Finished on $1$-line.
After that, the view of the next robot $r_l$ ($l$-th robot) becomes {\sf StopP1B} on the intersection of the first border and $1$-line, and $r_l$ becomes Finished by Rule 2. 

Thus, the lemma holds.
\end{proof}

\begin{lemma}\label{32hops}
The distance between any two robots on the grid is two hops after every robot becomes Finished.
\end{lemma} 
\begin{proof}
%
By the definitions of ${\it StopP1}^\prime$ and ${\it StopP2}^\prime$, the distance between a robot $r_i$ and its predecessor $r_j$ is two hops after each robot becomes Finished if $r_i$ and $r_j$ are on the same line.
Thus, when the robots on $m$-line ($0<m<L-1$) become Finished, if there are two adjacent Finished robots to the contrary, then there is a robot $r_r$ on $m$-line that cannot move from the node that is adjacent to a node occupied by a Finished robot on $(m-1)$-line.
However, by the same argument as in Lemmas~\ref{31stB} and \ref{31-line}, if $m$ is odd (resp. even), the nodes on $m$-line are occupied alternately beginning with an empty node (resp. occupied node) because the nodes on $(m-1)$-line are also occupied alternately beginning with an occupied node (resp. empty node).
Thus, before such $r_r$ becomes Finished, $r_r$ has a view of type {\sf MovP2C} and can move by Rule 6, \emph{i.e.}, such $r_r$ cannot exist.

Now, to consider the end of the execution of the algorithm, we consider $(L-1)$-line when nodes on $(L-2)$-line are occupied by Finished robots. 
The $(L-1)$-line is a border connected to the Door corner.
Then, if both $l$ and $L$ are odd or both are even, the view from the Door node becomes {\sf DoorC} or {\sf DoorD}, otherwise {\sf DoorB} (See the left figures of Figure~\ref{fig:checkers}). 
Note that, in the case of {\sf DoorC}, the last robot $r_i$ on the $(L-2)$-line becomes Finished by {\sf StopP2E} when its successor $r_j$ arrives at the Door corner.
Then, after $r_j$ moves two hops (\emph{i.e.}, by {\sf ChangeC} and {\sf MovP2C}), the view from the Door node becomes {\sf DoorD} for $r_j$'s successor.
\begin{itemize}
\item If the view from the Door node is {\sf DoorC} or {\sf DoorD}, the robot on the Door node moves to the Door corner by Rule 1. 
Then, the view from the Door corner is {\sf ChangeC}.
By the same discussion as above (FIGURE~\ref{fig:checkers}), the view from the Door corner eventually becomes {\sf StopP1A}, thus the final robot on the Door corner becomes Finished by Rule 2.
Then, any other robot cannot enter into the grid because there is no such rule.
\item If the view from the Door node is {\sf DoorB}, the view from the Door corner is {\sf ChangeB}.
Then, the empty node $v$ that is adjacent to the Door corner is eventually occupied by a Finished robot on $(L-1)$-line (FIGURE~\ref{fig:checkers}).
After that, 
any other robot on the Door node cannot enter into the grid because there is no such rule. 
\end{itemize}
Thus, the lemma holds.
\end{proof}

\begin{lemma} 
Every robot on the grid is eventually Finished.
\end{lemma}
\begin{proof}
By the proofs of Lemmas~\ref{predecessor}-\ref{32hops}, the transitions of the view type of each robot are shown as FIGURE~\ref{fig:Alg2}. 
Thus, the lemma holds.
\end{proof}

By Lemma~\ref{32hops}, distances between any two occupied nodes are two.
Thus, by the same discussion as Theorem~\ref{1MIS},
we are now able to state our main result:

\begin{theorem}\label{MIS}
Algorithm $\mathcal{MIS\_NOLAB}$ constructs a maximum independent set by occupied locations on the grid.
\end{theorem}

By the proofs of Lemmas~\ref{31stB}--\ref{32hops}, nodes on the even-numbers (resp. odd-numbers) lines are occupied by $\lceil l/2\rceil$ (resp. $\lfloor l/2 \rfloor$) robots.
By the same discussion as Lemma~\ref{number}, the following lemma holds.
\begin{lemma}\label{3number}
When a maximum independent set is constructed, $\lceil n/2 \rceil$ robots are on the grid.
\end{lemma}

Each robot $r_i$ sets its value $k_i$ at most once by Rule 0.
By the same discussion of Theorem~\ref{time}, the following theorem holds.
\begin{theorem}\label{3time}
The time complexity of $\mathcal{MIS\_NOLAB}$ is $O(n(L+l))$ steps.
\end{theorem}

\section{Conclusion}
We proposed two algorithms to construct a maximum independent set on an unknown size grid in the case that the Door node is connected to a corner node.
One of our algorithms uses only three colors for each robot light and $\phi=2$, but it assumes edge-labeling.
The other uses seven colors for each robot light and $\phi=3$, and it executes in a completely anonymous graph.
Both of the time complexity are $O(n(L+l))$ steps.


Some interesting questions remain open:
\begin{itemize}
    \item Are there any algorithms for the case where each robot has no light or two light colors? Following the results by Hesari et al.~\cite{HFNOS14c} for the continuous line setting, we conjecture their impossibility result for oblivious (\emph{a.k.a.} no-light robots) can be extended to the discrete asynchronous and unoriented setting. 
\item Are there any algorithms for other assumptions of the Door node? For example, the Door node can be connected to another node, and there may be multiple Door nodes.
    \item Are there any algorithms that can tolerate maximum independent set reconfiguration in the case of robot crashes? We conjecture that, assuming a failing robot turns off its light (that is, crash failures can be detected by other robots), it is possible to extend our algorithm to adjust the remaining robots and introduce new ones so that the maximum independent set is reconstructed. 
\end{itemize}
Additionally, we plan to design algorithms for the case of a maximal independent set placement, and minimum dominating set placement, that requires fewer robots.

\bibliographystyle{splncs04}
\bibliography{biblio}

@ARTICLE{SY99,
  author    = {Ichiro Suzuki and Masafumi Yamashita},
  title     = {Distributed Anonymous Mobile Robots: Formation of Geometric Patterns},
  journal   = sicomp,
  volume    = 28,
  number    = 4,
  pages     = {1347--1363},
  year      = {1999},
  
}

@ARTICLE{formation-survey,
author ={Yukiko Yamauchi},
title = {A survey on pattern formation of autonomous mobile robots: asynchrony, obliviousness and visibility},
journal = {Journal of Physics: Conference Series},
volume = {473},
	pages = {012016},
	year={2013}
	}

@Article{barriere11ijfcs,
  author    = {Lali Barri{\`{e}}re and Paola Flocchini and Eduardo Mesa Barrameda and Nicola Santoro},
  journal   = {International Journal of Foundations of Computer Science},
  title     = {Uniform Scattering of Autonomous Mobile Robots in a Grid},
  year      = {2011},
  number    = {3},
  pages     = {679--697},
  volume    = {22},
}

@Article{BAKS2020,
  author  = {Kaustav Bose and Ranendu Adhikary and Manash Kumar Kundu and Buddhadeb Sau},
  journal = tcs,
  title   = {Arbitrary Pattern Formation on Infinite Grid by Asynchronous Oblivious Robots},
  year    = {2020},
  pages   = {213--227},
  volume  = {815},
}

@Article{BT17j,
  author    = {Quentin Bramas and S{\'{e}}bastien Tixeuil},
  journal   = {International Journal of Foundations of Computer Science},
  title     = {The Random Bit Complexity of Mobile Robots Scattering},
  year      = {2017},
  number    = {2},
  pages     = {111--134},
  volume    = {28},
}

@Article{casteigts12cc,
  author    = {Arnaud Casteigts and Jeremie Albert and Serge Chaumette and Amiya Nayak and Ivan Stojmenovic},
  journal   = {Computer Communications},
  title     = {Biconnecting a network of mobile robots using virtual angular forces},
  year      = {2012},
  number    = {9},
  pages     = {1038--1046},
  volume    = {35},
}

@Article{CDPIM10,
  author    = {Julien Cl{\'{e}}ment and Xavier D{\'{e}}fago and Maria Gradinariu Potop{-}Butucaru and Taisuke Izumi and St{\'{e}}phane Messika},
  journal   = ipl,
  title     = {The cost of probabilistic agreement in oblivious robot networks},
  year      = {2010},
  number    = {11},
  pages     = {431--438},
  volume    = {110},
}

@Article{CP08j,
  author    = {Reuven Cohen and David Peleg},
  journal   = tcs,
  title     = {Local spreading algorithms for autonomous robot systems},
  year      = {2008},
  number    = {1-2},
  pages     = {71--82},
  volume    = {399},
}

@ARTICLE{DSKN16,
  author    = {Gianlorenzo D'Angelo and
               Gabriele Di Stefano and
               Ralf Klasing and
               Alfredo Navarra},
  title     = {Gathering of robots on anonymous grids and trees without multiplicity
               detection},
  journal   = tcs,
  volume    = {610},
  pages     = {158--168},
  year      = {2016}
}

@ARTICLE{DFPSY16,
  author    = {Shantanu Das and
               Paola Flocchini and
               Giuseppe Prencipe and
               Nicola Santoro and
               Masafumi Yamashita},
  title     = {Autonomous mobile robots with lights},
  journal   = tcs,
  volume    = {609},
  pages     = {171--184},
  year      = {2016}
}

@Article{SAFPS20,
  author  = {St{\'{e}}phane Devismes and Anissa Lamani and Franck Petit and Pascal Raymond and S{\'{e}}bastien Tixeuil},
  journal = cj,
  title   = {Terminating Exploration Of A Grid By An Optimal Number Of Asynchronous Oblivious Robots},
  year    = {2020},
  month   = jan,
  number  = {1},
  pages   = {132–154},
  volume  = {64},
}

@Article{DP09j,
  author    = {Yoann Dieudonn{\'{e}} and Franck Petit},
  journal   = {Parallel Processing Letters},
  title     = {Scatter of Robots},
  year      = {2009},
  number    = {1},
  pages     = {175--184},
  volume    = {19},
}

@ARTICLE{DE2016,
author={Mattia D'Emidio and Daniele Frigioni and Alfredo Navarra},
title={Synchronous Robots vs Asynchronous Lights-Enhanced Robots on Graphs},
journal=ENTCS,
year={2016},
pages={169--180},
volume={322}
}

@Article{FPS08j,
  author    = {Paola Flocchini and Giuseppe Prencipe and Nicola Santoro},
  journal   = tcs,
  title     = {Self-deployment of mobile sensors on a ring},
  year      = {2008},
  number    = {1},
  pages     = {67--80},
  volume    = {402},
}

@Article{GP13,
  author  = {Samuel Guilbault and Andrzej Pelc},
  journal = tcs,
  title   = {Gathering asynchronous oblivious agents with local vision in regular bipartite graphs},
  year    = {2013},
  pages   = {86--96},
  volume  = {509},
}

@inproceedings{filling2008,
 author={Eduardo Mesa Barrameda and Shantanu Das and Nicola Santoro},
 title={Deployment of Asynchronous Robotic Sensors in Unknown Orthogonal Environments},
 booktitle={Proceedings of {ALGOSENSORS}2008},
 year={2008},
 pages={125--140},
 address={Reykjavik, Iceland},
 month ={12 July},
 publisher={Springer-Verlag, Berlin}
}

@InProceedings{bonnet11opodis,
  author    = {Fran{\c{c}}ois Bonnet and Alessia Milani and Maria Potop{-}Butucaru and S{\'{e}}bastien Tixeuil},
  booktitle = {Proceedings of {OPODIS}2011},
  title     = {Asynchronous Exclusive Perpetual Grid Exploration without Sense of Direction},
  year      = {2011},
  address={Toulouse, France},
  month     = {12-16 December},
  pages     = {251--265},
  publisher = {Springer, Cham},
}

@InProceedings{BDL19,
  author    = {Quentin Bramas and St{\'e}phane Devismes and Pascal Lafourcade},
  booktitle = {Proceedings of {NETYS}2020},
  title     = {Infinite Grid Exploration by Disoriented Robots},
  year      = {2020},
  pages     = {129--145},
  address={Marrakech, Morocco},
  month = {3-5 June},
  publisher={Springer, Cham}
}

@InProceedings{DLLP13,
  author    = {Ajoy Kumar Datta and Anissa Lamani and Lawrence L. Larmore and Franck Petit},
  booktitle = {Proceedings of {SAFECOMP}2013},
  title     = {Ring Exploration with Oblivious Myopic Robots},
  year      = {2013},
  pages     = {335--342},
  month = {24-27 September},
  address={Toulouse, France},
  publisher={Springer-Verlag, Berlin},
}

@InProceedings{DLPRT12,
  author    = {St{\'e}phane Devismes and Anissa Lamani and Franck Petit and Pascal Raymond and S{\'e}bastien Tixeuil},
  booktitle = {Proceedings of {SSS}2012},
  title     = {Optimal Grid Exploration by Asynchronous Oblivious Robots},
  year      = {2012},
  pages     = {64--76},
  address={Toronto, Canada},
  publisher={Springer-Verlag, Berlin},
  month={1-4 Octorber}
}

@InProceedings{HVST2020,
  author    = {Rory Hector and Ramachandran Vaidyanathan and Gokarna Sharma and Jerry L. Trahan},
  booktitle = {Proceedings of {IPDPS}2020},
  title     = {Optimal Convex Hull Formation on a Grid by Asynchronous Robots with Lights},
  year      = {2020},
  address   = {New Orleans, USA},
  month     = {18-22 May},
  pages     = {1051--1060},
  publisher = {IEEE, USA},
}

@InProceedings{heriban18icdcn,
  author    = {Adam Heriban and Xavier D{\'{e}}fago and S{\'{e}}bastien Tixeuil},
  booktitle = {Proceedings of {ICDCN}2018},
  title     = {Optimally Gathering Two Robots},
  year      = {2018},
  address   = {Varanasi, India},
  pages     = {1--10},
  publisher = {ACM, USA},
  month={4-7 January}
}

@InProceedings{HFNOS14c,
  author    = {Mohsen Eftekhari Hesari and Paola Flocchini and Lata Narayanan and Jaroslav Opatrny and Nicola Santoro},
  booktitle = {Proceedings of {SIROCCO}2014},
  title     = {Distributed Barrier Coverage with Relocatable Sensors},
  year      = {2014},
  address={Takayama, Japan},
  month= {23-25 July},
  pages     = {235--249},
  publisher={Springer, Cham},
}

@inproceedings{filling2017,
author={Attila Hideg and Tamas Lukovszki},
title={Uniform Dispersal of Robots with Minimum Visibility Range},
booktitle={Proceedings of {ALGOSENSORS}2017},
year={2017},
pages={155--167},
address={Vienna, Austria},
month={4-8 September},
publisher={Springer, Cham}
}

@Inproceedings{filling2,
  author      = {Attila Hideg and Tamas Lukovszki},
  title       = {Asynchronous Filling by Myopic Luminous Robots},
  year        = {2020},
  booktitle   = {Proceedings of {ALGOSENSORS}2020},
  pages={108--123},
  address={Pisa, Italy},
  month={9-10 September},
  publisher={Springer, Cham}
}

@inproceedings{filling2018,
author={Attila Hideg and Tamas Lukovszki and Bertalan Forstner},
title={Filling Arbitrary Connected Areas by Silent Robots with Minimum Visibility Range},
booktitle={Proceedings of {ALGOSENSORS}2018},
year={2018},
pages={193--205},
address={Helsinki, Finland},
month={20-24 August},
publisher={Springer-Verlag, Berlin}
}

@InProceedings{ABKS2018,
  author    = {Ranendu Adhikary and Kaustav Bose and Manash Kumar Kundu and Buddhadeb Sau},
  booktitle = {Proceedings of {ALGOSENSORS}2018},
  title     = {Mutual Visibility by Asynchronous Robots on Infinite Grid},
  year      = {2018},
  address   = {Helsinki, Finland},
  month     = {20-24 August},
  pages     = {83--101},
  publisher = {Springer-Verlag, Berlin},
}

@InProceedings{KLO14,
  author    = {Sayaka Kamei and Anissa Lamani and Fukuhito Ooshita},
  booktitle = {Proceedings of {SRDSW}2014},
  title     = {Asynchronous Ring Gathering by Oblivious Robots with Limited Vision},
  year      = {2014},
  pages     = {46--49},
  address={Nara, Japan},
  publisher={IEEE, USA},
  month={6-9 Octorber}
}

@InProceedings{OPODIS19,
  author    = {Sayaka Kamei and Anissa Lamani and Fukuhito Ooshita and S{\'{e}}bastien Tixeuil and Koichi Wada},
  booktitle = {Proceedings of {OPODIS}2019},
  title     = {Gathering on Rings for Myopic Asynchronous Robots With Lights},
  year      = {2019},
  pages={27:1--27:17},
  publisher={LIPICS, Germany},
  address ={Neuchâtel, Switzerland},
  month={17-19 December},
}

@InProceedings{WALCOM20,
  author    = {Ajay Kshemkalyani and Anisur Rahaman Molla and Gokarna Sharma},
  booktitle = {Proceedings of {WALCOM}2020},
  title     = {Dispersion of Mobile Robots on Grids},
  year      = {2020},
  address   = {Singapore},
  month     = {31 March-2 April},
  pages     = {183--197},
  publisher = {Springer, Cham},
}

@InProceedings{NOI19,
  author    = {Shota Nagahama and Fukuhito Ooshita and Michiko Inoue},
  booktitle = {Proceedings of {SSS}2019},
  title     = {Ring Exploration of Myopic Luminous Robots with Visibility More than One},
  year      = {2019},
  pages     = {256--271},
  address ={Pisa, Italy},
  month ={22-25 October},
  publisher={Springer, Cham}
}

@InProceedings{APDCM2021,
author={Shota Nagahama and Fukuhito Ooshita and Michiko Inoue},
title={Terminating Grid Exploration with Myopic Luminous Robots},
booktitle={Proceedings of {IPDPSW}2021},
pages={586--595},
year={2021},
address={Portland, USA},
month= {17 May},
publisher = {IEEE Computer Society, USA},
}

@InProceedings{SF-Ex,
  author    = {Ooshita, Fukuhito and Tixeuil, S{\'e}bastien},
  booktitle = {Proceedings of {SSS}2018},
  title     = {Ring Exploration with Myopic Luminous Robots},
  year      = {2018},
  address={Tokyo, Japan},
  month={4-7 November},
  pages     = {301--316},
  publisher ={Springer, Cham}
}

@InProceedings{PS2020,
author={Pavan Poudel and Gokarna Sharma},
title={Fast Uniform Scattering on a Grid for Asynchronous Oblivious Robots},
booktitle={Proceedings of {SSS}2020},
year={2020},
pages={211--228},
address={Austin, USA},
month= {18-21 November},
publisher={Springer, Cham}
}

@InProceedings{G13,
  author    = {Giovanni Viglietta},
  booktitle = {Proceedings of {ALGOSENSOR}2013},
  title     = {Rendezvous of Two Robots with Visible Bits},
  year      = {2013},
  pages     = {291--306},
  address ={Sophia Antipolis, France},
  month = {5-6 September},
  publisher = {Springer-Verlag, Berlin},
}

@inproceedings{filling2013,
author={Eduardo Mesa Barrameda and Shantanu Das and Nicola Santoro},
title={Uniform Dispersal of Asynchronous Finite-State Mobile Robots in Presence of Holes},
booktitle={Proceedings of {ALGOSENSORS}2013},
year={2013},
pages={228--243},
address={Sophia Antipolis, France},
month= {5-6 September},
publisher={Springer-Verlag, Berlin}
}

@InProceedings{balabonski19netys,
  author    = {Thibaut Balabonski and Pierre Courtieu and Robin Pelle and Lionel Rieg and S{\'{e}}bastien Tixeuil and Xavier Urbain},
  booktitle = {Proceedings of {NETYS}2019},
  title     = {Continuous vs. Discrete Asynchronous Moves: {A} Certified Approach for Mobile Robots},
  year      = {2019},
  pages     = {93--109},
  publisher = {Springer, Cham},
  month={19-21 June},
  address={Morocco},
}

@book{FPS19,
  editor    = {Paola Flocchini and
               Giuseppe Prencipe and
               Nicola Santoro},
  title     = {Distributed Computing by Mobile Entities, Current Research in Moving and Computing},
  publisher={Springer},
  address={Cham, Switzerland},
  year      = {2019}
}

@inbook{FPS19p1,
  author    = {Paola Flocchini and
               Giuseppe Prencipe and
               Nicola Santoro},
editor    = {Paola Flocchini and
               Giuseppe Prencipe and
               Nicola Santoro},
  booktitle     = {Distributed Computing by Mobile Entities, Current Research in Moving and Computing},
  title={Moving and Computing Models: Robots},
  publisher={Springer},
  address={Cham, Switzerland},
  year      = {2019},
  pages={3--14}
}

@inbook{FPS19p2,
  author    = {Serafino Cicerone and Gabriele Di Stefano and Alfredo Navarra},
editor    = {Paola Flocchini and
               Giuseppe Prencipe and
               Nicola Santoro},
  booktitle     = {Distributed Computing by Mobile Entities, Current Research in Moving and Computing},
  title={Asynchronous Robots on Graphs: Gathering},
  publisher={Springer},
  address={Cham, Switzerland},
  year      = {2019},
  pages={184--217}
}

@inbook{FPS19p3,
  author    = {David Ilcinkas},
editor    = {Paola Flocchini and
               Giuseppe Prencipe and
               Nicola Santoro},
  booktitle     = {Distributed Computing by Mobile Entities, Current Research in Moving and Computing},
  title={Oblivious Robots on Graphs: Exploration},
  publisher={Springer},
  address={Cham, Switzerland},
  year      = {2019},
  pages={218--233}
}

@inbook{FPS19p4,
  author    = {Xavier D{\'{e}}fago and Maria Potop-Butucaru and S{\'{e}}bastien Tixeuil},
editor    = {Paola Flocchini and
               Giuseppe Prencipe and
               Nicola Santoro},
  booktitle     = {Distributed Computing by Mobile Entities, Current Research in Moving and Computing},
  title={Fault-Tolerant Mobile Robots},
  publisher={Springer},
  address={Cham, Switzerland},
  year      = {2019},
  pages={234--251}
}

@inbook{FPS19p5,
  author    = {Maria Potop-Butucaru and Nathalie Sznajder and S{\'{e}}bastien Tixeuil and Xavier Urbain},
editor    = {Paola Flocchini and
               Giuseppe Prencipe and
               Nicola Santoro},
  booktitle     = {Distributed Computing by Mobile Entities, Current Research in Moving and Computing},
  title={Formal Methods for Mobile Robots},
  publisher={Springer},
  address={Cham, Switzerland},
  year      = {2019},
  pages={278--313}
}

@InBook{Hsiang,
  author    = {Tien-Ruey Hsiang and Esther M. Arkin and Michael A. Bender and Sándor P. Fekete and Joseph S. B. Mitchell},
  booktitle = {Algorithmic Foundations of Robotics V.},
  publisher = {Springer},
  title     = {Algorithms for Rapidly Dispersing Robot Swarms in Unknown Environments},
  year      = {2004},
  address   = {Berlin, Heidelberg},
  editor    = {Jean-Daniel Boissonnat and Joel Burdick and Ken Goldberg and Seth Hutchinson},
  pages     = {77--93},
}

@InBook{F16bc,
  author    = {Paola Flocchini},
  booktitle = {Encyclopedia of Algorithms},
  publisher = {Springer},
  title     = {Uniform Covering of Rings and Lines by Memoryless Mobile Sensors},
  year      = {2016},
  address   = {Boston, MA},
  editor    = {Ming-Yang Kao},
  pages     = {2297--2301},
}

\section*{Data availability}
No new data were generated or analysed in support of this research.
\end{document}